\documentclass{vldb}
\usepackage{graphicx}
\usepackage{balance}  % for  \balance command ON LAST PAGE  (only there!)

\usepackage{amsthm}
\usepackage{amssymb}
\usepackage{url}

\usepackage{amsthm}
\usepackage{thmtools}
\usepackage{thm-restate}

\newtheorem{definition}{Definition}

\newtheorem{lemma}{Lemma}

\usepackage{amssymb}
\usepackage{bm}

\usepackage{color}
\usepackage{ifthen}

\newif\ifcommentson\commentsontrue
\newcommand{\colorB}[1]{\textcolor{black}{#1}}
\newcommand{\colorBB}[1]{\textcolor{black}{#1}}

\newif\ifconferenceon\conferenceonfalse
\ifconferenceon
\newcommand{\conference}[1]{#1}
\newcommand{\arxiv}[1]{}
\else
\newcommand{\conference}[1]{}
\newcommand{\arxiv}[1]{#1}
\fi

\allowdisplaybreaks[1]

\newcommand{\LDP}{LDP}
\newcommand{\OSLDP}{OSLDP}
\newcommand{\ULDP}{ULDP}

\newcommand{\nngreals}{\mathbb{R}_{\ge0}}
\newcommand{\argmax}{\operatornamewithlimits{argmax}}

\newcommand{\expect}[1]{\operatornamewithlimits{\displaystyle\mathbb{E}}_{#1}}
\newcommand{\expectym}{\expect{Y^n\sim \bmm^n}}
\newcommand{\var}{\text{Var}}
\newcommand{\normal}[1]{\mathcal{N}\!\left(#1\right)}

\newcommand{\calC}{\mathcal{C}}

\newcommand{\calM}{\mathcal{M}}

\newcommand{\calW}{\mathcal{W}}
\newcommand{\calX}{\mathcal{X}}
\newcommand{\calY}{\mathcal{Y}}
\newcommand{\calZ}{\mathcal{Z}}

\newcommand{\bbE}{\mathbb{E}}

\newcommand{\uid}[1]{^{(#1)}}

\newcommand{\bmX}{\mathbf{X}}
\newcommand{\bmY}{\mathbf{Y}}

\newcommand{\bmQ}{\mathbf{Q}}
\newcommand{\bmQxy}{\mathbf{Q}(y|x)}
\newcommand{\bmQxdy}{\mathbf{Q}(y|x')}
\newcommand{\bmQixy}{\mathbf{Q}^{(i)}(y|x)}

\newcommand{\bmQRR}{\bmQ_{\it RR}}
\newcommand{\bmQRAPPOR}{\bmQ_{\it RAP}}
\newcommand{\bmQrestRR}{\bmQ_{\it uRR}}
\newcommand{\bmQrestRAPPOR}{\bmQ_{\it uRAP}}

\newcommand{\bmp}{\mathbf{p}}
\newcommand{\hbmp}{\hat{\mathbf{p}}}

\newcommand{\bmm}{\mathbf{m}}
\newcommand{\hbmm}{\hat{\mathbf{m}}}
\newcommand{\bmr}{\mathbf{r}}
\newcommand{\hbmr}{\hat{\mathbf{r}}}

\newcommand{\bmpUS}{\mathbf{p}_{U_{\!S}}}
\newcommand{\bmpUN}{\mathbf{p}_{U_{\!N}}}

\newcommand{\bmpi}{\mathbf{\pi}}
\newcommand{\bmt}{\mathbf{t}}

\newcommand{\calXS}{\mathcal{X}_S}
\newcommand{\calXN}{\mathcal{X}_N}
\newcommand{\calYP}{\mathcal{Y}_P}
\newcommand{\calYI}{\mathcal{Y}_I}
\newcommand{\calZS}{\mathcal{Z}_S}
\newcommand{\calZP}{\mathcal{Z}_P}
\newcommand{\calZI}{\mathcal{Z}_I}

\vldbTitle{Utility-Optimized Local Differential Privacy Mechanisms for Distribution Estimation}
\vldbAuthors{Takao Murakami and Yusuke Kawamoto}
\vldbDOI{https://doi.org/TBD}
\vldbVolume{12}
\vldbNumber{xxx}
\vldbYear{2019}

\begin{document}

\title{Utility-Optimized Local Differential Privacy Mechanisms for Distribution Estimation \thanks{This study was supported by JSPS KAKENHI 
JP19H04113, JP17K12667, and by Inria under the project LOGIS.}}

\numberofauthors{2} %  in this sample file, there are a *total*

\author{
\alignauthor
Takao Murakami\\
       \affaddr{AIST, Tokyo, Japan}\\
       \email{takao-murakami at aist.go.jp}
\alignauthor
Yusuke Kawamoto\\
       \affaddr{AIST, Tsukuba, Japan}\\
       \email{yusuke.kawamoto at aist.go.jp}
}
\date{30 July 1999}

\maketitle

\begin{abstract}
LDP (Local Differential Privacy) 
has been widely studied to estimate statistics of personal data 
(e.g., distribution underlying the data) 
while protecting users' privacy. 
Although LDP does not require 
a trusted third party, 
it regards all personal data equally sensitive, which causes excessive obfuscation hence the loss of 
utility.

In this paper, we 
introduce the notion of 
\emph{ULDP (Utility-optimized LDP)}, 
which provides a privacy guarantee equivalent to LDP only for sensitive data. 
We first 
consider 
the setting where 
all users use the same obfuscation mechanism, 
and propose two mechanisms providing ULDP: 
\emph{utility-optimized randomized response} 
and 
\emph{utility-optimized RAPPOR}. 
We then consider 
the setting where 
the distinction between sensitive and non-sensitive data can be different from user to user. 
For this setting, 
we propose 
a \emph{personalized ULDP mechanism with semantic tags} 
to estimate the distribution of personal data with high utility while keeping secret what is sensitive for each user. 
We show, both theoretically and experimentally, that our mechanisms provide 
much higher utility than the existing LDP mechanisms when there are a lot of non-sensitive data. We also show that when most of the data are non-sensitive, our mechanisms even provide almost the same utility as non-private mechanisms in the low privacy regime. 
\end{abstract}

\section{Introduction}
\label{sec:intro}

\colorB{DP (Differential Privacy) \cite{Dwork_TCC06,DP} 
is becoming a gold standard for data privacy; 
it enables 
big data 
analysis 
while protecting users' privacy 
against adversaries with arbitrary background knowledge. 
According to the underlying architecture, 
DP 
can be categorized into 
the one in the \emph{centralized model} 
and 
the one in the \emph{local model} \cite{DP}.}
In the centralized model, 
a ``trusted'' database administrator, who can access to all users' personal data, obfuscates the data 
(e.g., by adding noise, generalization) 
before providing them to a (possibly malicious) 
data analyst. 
Although DP was extensively studied 
for the centralized model at the beginning, 
the original personal data in this model can be leaked from the database by illegal access or internal fraud. 
This issue is critical in recent years, because the number of data breach incidents is increasing \cite{data_breach}.

The local model 
does not require a ``trusted'' administrator, 
and therefore 
does not suffer from the data leakage issue 
explained above. 
In this model, each user obfuscates her personal data by herself, and 
sends the obfuscated data to a data collector (or data analyst). 
Based on the obfuscated data, 
the data collector can estimate some statistics (e.g., histogram, heavy hitters \cite{Qin_CCS16}) of the personal data. 
DP in the local model, which is called 
\emph{LDP} (\emph{Local Differential Privacy}) \cite{Duchi_FOCS13}, 
has recently attracted much attention in the academic field 
\cite{Avent_USENIX17,Cormode_SIGMOD18_2,Fanti_PoPETs16,Kairouz_ICML16,Kairouz_JMLR16,Murakami_PoPETs18,Pastore_ISIT16,Qin_CCS16,Sei_TIFS17,Wang_ICDE18,Ye_ISIT17}, 
and has also been adopted by 
\colorBB{industry} 
\cite{Ding_NIPS17,Thakurta_USPatent17,Erlingsson_CCS14}. 

However, LDP mechanisms regard all personal data as equally sensitive, and leave a lot of room for increasing data utility. 
For example, consider questionnaires such as: ``Have you ever cheated in an exam?'' and ``Were you with a prostitute in the last month?'' \cite{Cohen_SSSIPP89}. 
Obviously, ``Yes'' is a sensitive response to these questionnaires, whereas ``No'' is not sensitive. 
A 
RR (Randomized Response) method 
proposed 
by Mangat \cite{Mangat_JRSS94} utilizes this fact. 
Specifically, it reports ``Yes'' or ``No'' as follows: 
if the true answer is ``Yes'', always report ``Yes''; otherwise, report ``Yes'' and ``No'' with probability $p$ and $1-p$, respectively. 
Since the reported answer ``Yes'' may come from both the true answers ``Yes'' and ``No'', the confidentiality of the user reporting ``Yes'' is not violated. %\cite{Mangat_JRSS94}. 
Moreover, since the reported answer ``No'' is always come from the true answer ``No'', the data collector can estimate a distribution of true answers with higher accuracy than Warner's 
RR
\cite{Warner_JASA65}, which simply flips ``Yes'' and ''No'' with probability $p$. 
However, Mangat's 
RR 
does not provide LDP, since LDP regards both ``Yes'' and ``No'' as equally sensitive. 

There are a lot of ``non-sensitive'' data 
for other types of data. For example, locations such as hospitals and home can be sensitive, whereas visited sightseeing places, restaurants, and coffee shops are non-sensitive for many users. 
Divorced people may want to keep their divorce secret, while 
the others may not care about their marital status. 
The distinction between sensitive and non-sensitive data can also be different from user to user 
(e.g., home address is different from user to user; 
some people might want to keep secret even the sightseeing places). 
\colorB{To explain more about this issue, 
we briefly review related work on 
LDP and variants of DP.}

\smallskip
\noindent{\colorB{\textbf{Related work.}}}~~\colorB{Since Dwork \cite{Dwork_TCC06} introduced 
DP, 
a number of its variants have been studied to provide different types of privacy guarantees; 
e.g., 
LDP \cite{Duchi_FOCS13}, 
$d$-privacy \cite{Chatzikokolakis_PETS13}, 
Pufferfish privacy \cite{Kifer_TODS14}, 
dependent DP \cite{Liu_NDSS16}, 
Bayesian DP \cite{Yang_SIGMOD15}, 
mutual-information DP \cite{Cuff_CCS16},
R\'{e}nyi DP \cite{Mironov_CSF17}, and
distribution privacy \cite{Kawamoto:18:arxiv}. 
In particular, LDP \cite{Duchi_FOCS13} has been widely studied in the literature. 
For example, 
Erlingsson \emph{et al.} \cite{Erlingsson_CCS14} proposed the RAPPOR as an obfuscation mechanism providing LDP, and implemented it in Google Chrome browser. 
Kairouz \emph{et al.} \cite{Kairouz_ICML16} showed that under the $l_1$ and $l_2$ losses, 
the randomized response (generalized to multiple alphabets) and RAPPOR 
are order optimal among all LDP mechanisms in the low and high privacy regimes, respectively. 
Wang \textit{et al.} \cite{Wang_USENIX17} generalized the RAPPOR and  
a random projection-based method \cite{Bassily_STOC15}, 
and found parameters that minimize the variance of the estimate.} 

Some studies also attempted to address the non-uniformity of privacy requirements among 
records (rows) 
or among items (columns) 
in the centralized DP: 
Personalized DP \cite{Jorgensen_ICDE15}, 
Heterogeneous DP \cite{Alaggan_JPC16}, and 
One-sided DP \cite{Doudalis_arXiv17}. 
However, 
obfuscation mechanisms 
that address the non-uniformity among input values
in the ``local'' DP 
have not been studied, to our knowledge. 
In this paper, we show 
that data utility can be significantly increased by designing such local mechanisms.

\smallskip
\noindent{\textbf{Our contributions.}}~~\colorB{The goal of this paper is to design obfuscation mechanisms in the local model that achieve high 
data utility while providing DP for sensitive data. 
To achieve this, 
we introduce the notion of 
\emph{ULDP (Utility-optimized LDP)}, 
which provides a privacy guarantee equivalent to LDP only for sensitive data, 
and obfuscation mechanisms providing ULDP.} 
As a task for the data collector, we consider \emph{discrete distribution estimation} 
\cite{Agrawal_SIGMOD05,Fanti_PoPETs16,Huang_ICDE08,Kairouz_ICML16,Murakami_PoPETs18,Sei_TIFS17,Erlingsson_CCS14,Ye_ISIT17}, 
where 
personal data 
take 
discrete values. 
Our contributions are as follows:
\begin{itemize}
\item 
We first consider the setting in which all users use the same obfuscation mechanism, 
and propose 
two 
ULDP mechanisms: 
\emph{utility-optimized RR} 
and \emph{utility-optimized RAPPOR}. 
\colorB{We 
prove that when there are a lot of non-sensitive data, 
our 
mechanisms provide much higher utility than 
two state-of-the-art LDP mechanisms: 
the RR (for multiple alphabets) \cite{Kairouz_ICML16,Kairouz_JMLR16} and RAPPOR \cite{Erlingsson_CCS14}. 
We also prove that when most of the data are non-sensitive, our mechanisms even provide almost the same utility as a non-private mechanism that does not obfuscate the personal data in the low privacy regime 
where the privacy budget is $\epsilon = \ln |\calX|$ for a set $\calX$ of personal data.}
\item 
We then consider the setting in which 
the distinction between sensitive and non-sensitive data can be different from user to user, 
and 
propose 
a \emph{PUM (Personalized ULDP Mechanism) with semantic tags}. 
The PUM keeps secret what is sensitive for each user, 
while enabling the data collector to estimate a distribution using some background knowledge about the distribution conditioned on each tag (e.g., geographic distributions of homes).
We also theoretically analyze 
the data utility of 
the 
PUM.
\item We finally show that \colorB{our mechanisms 
are very promising in terms of utility} using two large-scale datasets.
\end{itemize}
\colorBB{The proofs of all statements in the paper 
are given in \arxiv{the appendices.}\conference{the extended version of the paper \cite{arxiv:ULDP}.}}

\smallskip
\noindent{\colorBB{\textbf{Cautions and limitations.}}}~~\colorBB{Although ULDP is meant to protect sensitive data, there are some cautions and limitations.}

\colorBB{First, we assume that each user sends 
a single 
datum 
and that each 
user's personal data 
is independent (see Section~\ref{sub:notations}). 
This is reasonable for a variety of personal data (e.g., locations, age, sex, marital status), where each user's data is  irrelevant to most others' one.
However, for some types of personal data (e.g., flu status \cite{Song_SIGMOD17}), 
each user can be highly influenced by others. 
There might also be a correlation between sensitive data and non-sensitive data when a user sends multiple 
data 
(on a related note, non-sensitive attributes may lead to re-identification of a record \cite{Narayanan_CACM10}). 
A possible solution to these problems would be to incorporate ULDP with \emph{Pufferfish privacy}~\cite{Kifer_TODS14,Song_SIGMOD17}, which is used to protect correlated data.
We leave this as future work (see Section~\ref{sec:discussions_extension} for 
discussions on the case of multiple 
data 
per user 
and the correlation issue).}

\colorBB{We focus on a scenario in which it is easy for users to decide what is sensitive (e.g., 
cheating experience, location of home). 
However, there is also a scenario in which users do not know what is sensitive. 
For the latter scenario, we cannot use ULDP but can simply apply LDP.
}

\colorBB{Apart from the sensitive/non-sensitive data issue, 
there are scenarios in which ULDP does not cover. 
For example, 
ULDP does not protect users who have a sensitivity about ``information disclosure'' itself (i.e., those who will not disclose any information). 
We assume that users have consented to information disclosure. 
To collect as much data as possible, 
we can provide an incentive for the information disclosure; e.g., provide a reward or point-of-interest (POI) information nearby a reported location. 
We also assume that the data collector obtains a consensus from users before providing reported data to third parties. 
Note that these cautions are common to LDP.}

\colorBB{There might also be a risk of discrimination; e.g., the data collector might discriminate against all users that provide a yes-answer, and have no qualms about small false positives. 
False positives decrease with increase in 
$\epsilon$. 
We note that LDP also suffer from this 
attack; the false positive probability is the same for both ULDP and LDP with the same $\epsilon$.
}

\colorBB{In summary, 
ULDP provides a privacy guarantee equivalent to LDP for sensitive data under the assumption of the data independence. 
We consider our work as a building-block of broader DP approaches or the basis for further development.}

\section{Preliminaries}

\subsection{Notations}
\label{sub:notations}

Let 
$\nngreals$ be the set of non-negative real numbers. 
Let 
$n$ be the number of users, 
\colorBB{$[n] = \{ 1, 2, \ldots, n\}$},
$\calX$ 
(resp.~$\calY$) 
be a finite set of personal (resp.~obfuscated) data. 
We assume continuous data are discretized into bins in advance (e.g., a location map is divided into some regions).
We use the superscript ``${(i)}$'' to represent the $i$-th user. 
Let 
$X\uid{i}$ (resp.~$Y^{(i)}$) 
be a random variable representing personal (resp.~obfuscated) data of 
the $i$-th user. 
The $i$-th user 
obfuscates her personal data 
$X\uid{i}$ 
via 
her obfuscation mechanism $\bmQ\uid{i}$, 
which maps $x \in \calX$ to $y \in \calY$ with probability $\bmQixy$, 
and sends the obfuscated data $Y\uid{i}$ to a data collector. 
Here we assume that each user sends 
a single 
\colorBB{datum}. 
We discuss the case 
\colorBB{of multiple data} 
in Section~\ref{sec:discussions_extension}. 

We divide personal data into two types: 
\emph{sensitive data} and \emph{non-sensitive data}. 
Let $\calXS \subseteq \calX$ be a set of sensitive data common to all users, and $\calXN = \calX \setminus \calXS$ be the remaining personal data.
Examples of such ``common'' sensitive data $x \in \calXS$ are the regions including public sensitive locations (e.g., hospitals) and 
obviously sensitive responses to questionnaires 
described in Section~\ref{sec:intro}\footnote{Note that these data might be sensitive for many/most users but not for all in practice (e.g., some people might not care about their cheating experience). However, we can regard these data as sensitive for all users (i.e., be on the safe side) by allowing a small loss of data utility.}.

Furthermore, let $\calXS\uid{i} \subseteq \calXN$ ($i\in [n]$) be 
a set of 
sensitive data specific to the $i$-th user (here we do not include $\calXS$ into $\calXS\uid{i}$ because $\calXS$ is protected for all users in our mechanisms). 
$\calXS\uid{i}$ is a set of personal data that is possibly non-sensitive for many users but sensitive for the $i$-th user. 
Examples of such ``user-specific'' sensitive data $x \in \calXS\uid{i}$ are the regions including private locations such as their home and workplace. (Note that the majority of 
working population can be uniquely identified from their home/workplace location pairs \cite{Golle_Pervasive09}.)

In Sections~\ref{sec:ULDP} and \ref{sec:rest_mechanism}, 
we consider the case where 
all users divide $\calX$ into the same sets of sensitive data and of non-sensitive data, 
i.e., $\calXS\uid{1} = \cdots = \calXS\uid{n} = \emptyset$, 
and use the same obfuscation mechanism $\bmQ$ (i.e., $\bmQ = \bmQ\uid{1} = \cdots = \bmQ\uid{n}$). 
In Section~\ref{sec:per_rest}, we consider a general setting 
that can deal with the user-specific sensitive data $\calXS\uid{i}$ and user-specific mechanisms $\bmQ\uid{i}$. 
We call the former case a \textit{common-mechanism scenario} 
and the latter a \textit{personalized-mechanism scenario}. 

We assume that each user's personal data $X\uid{i}$ is independently and identically distributed (i.i.d.) with a 
probability distribution 
$\bmp$, 
which generates $x \in \calX$ with probability $\bmp(x)$. 
Let $\bmX = (X\uid{1}, \cdots, X\uid{n})$ and $\bmY = (Y\uid{1}, \cdots, Y\uid{n})$ 
be tuples of all personal data and all obfuscated data, respectively. 
The data collector estimates $\bmp$ from $\bmY$ by a method %introduced 
described 
in Section~\ref{sub:distribution_estimation}. 
We denote by $\hbmp$ the estimate of $\bmp$. 
We further denote by $\calC$ the probability simplex; 
i.e., $\calC = \{\bmp | \sum_{x \in \calX} \bmp(x) = 1, \bmp(x) \geq 0 \text{ for any } x \in \calX\}$.

\arxiv{In Appendix~\ref{sec:notations}, we also show the basic notations in Table~\ref{tab:notations}.}

\subsection{Privacy Measures}
\label{sub:privacy_metrics}

LDP (Local Differential Privacy) \cite{Duchi_FOCS13} is defined as follows: 
\begin{definition} [$\epsilon$-\LDP{}] \label{def:LDP} 
Let $\epsilon \in \nngreals$. 
An obfuscation mechanism 
$\bmQ$ 
from $\calX$ to $\calY$ 
provides \emph{$\epsilon$-\LDP{}} 
if for any $x,x' \in \calX$ and any $y \in \calY$, 
\begin{align}
\bmQxy \leq e^\epsilon \bmQxdy.
\label{eq:LDP}
\end{align}
\end{definition}
LDP guarantees that an adversary who has observed $y$ cannot determine, 
for any pair of $x$ and $x'$, 
whether it is come from $x$ or $x'$ with a certain degree of confidence. 
As the privacy budget $\epsilon$ approaches to $0$, all of the data in $\calX$ become almost equally likely. 
Thus, a user's privacy is strongly protected when $\epsilon$ is small.

\subsection{Utility Measures}
\label{sub:utility_metrics}
In this paper, we use 
the $l_1$ loss (i.e., absolute error) and 
the $l_2$ loss (i.e., squared error) 
as utility 
measures. 
Let 
$l_1$ 
(resp.~$l_2^2$) 
be the $l_1$ (resp.~$l_2$) loss function, which maps
the estimate $\hbmp$ and the true distribution $\bmp$ to 
the loss; i.e., 
$l_1(\hbmp,\bmp) = \sum_{x \in \calX} |\hbmp(x) - \bmp(x)|$, 
$l_2^2(\hbmp,\bmp) = \sum_{x \in \calX} (\hbmp(x) - \bmp(x))^2$. 
It should be noted that 
$\bmX$ is generated from $\bmp$ and 
$\bmY$ is generated from $\bmX$ using 
$\bmQ\uid{1}, \cdots, \bmQ\uid{n}$. 
Since $\hbmp$ is computed from $\bmY$, 
both 
the $l_1$ and $l_2$ losses
depend on $\bmY$. 

In our theoretical analysis in Sections~\ref{sec:rest_mechanism} and \ref{sec:per_rest}, 
we take the expectation 
of the $l_1$ loss over 
all possible realizations of 
$\bmY$. 
In our experiments in Section~\ref{sec:exp}, we replace the expectation of 
the $l_1$ loss 
with the sample mean over multiple realizations of $\bmY$ and divide it by $2$ 
to evaluate 
the TV (Total Variation). 
In 
Appendix~\ref{sec:proofs_utility_l2loss}, 
we also show that the $l_2$ loss has similar results to the ones in Sections~\ref{sec:rest_mechanism} and \ref{sec:exp} by evaluating the expectation of the $l_2$ loss and the MSE (Mean Squared Error), respectively. 

\subsection{Obfuscation Mechanisms}
\label{sub:obfuscation_mechanisms}

\colorB{We describe the RR (Randomized Response) \cite{Kairouz_ICML16,Kairouz_JMLR16} and 
a generalized version of the RAPPOR \cite{Wang_USENIX17} as follows.}

\smallskip
\noindent{\textbf{Randomized response.}}~~The RR 
for $|\calX|$-ary alphabets was studied in \cite{Kairouz_ICML16,Kairouz_JMLR16}.
Its output range 
is identical to the input domain;
i.e., $\calX = \calY$. 

Formally, 
given $\epsilon \in \nngreals$, 
the 
\textit{$\epsilon$-RR} is an obfuscation mechanism that 
maps $x$ to $y$ with the probability:
\begin{align}
\bmQRR(y | x) = 
\begin{cases}
 \frac{e^\epsilon}{|\calX|+e^\epsilon-1} & \text{(if $y = x$)}\\
 \frac{1}{|\calX|+e^\epsilon-1} & \text{(otherwise)}.\\
\end{cases} 
\label{eq:RR}
\end{align}
It is easy to check by (\ref{eq:LDP}) and (\ref{eq:RR}) that $\bmQRR$ provides $\epsilon$-\LDP{}.

\smallskip
\noindent{\textbf{\colorB{Generalized} RAPPOR.}}~~The 
RAPPOR (Randomized Aggregatable Privacy-Preserving Ordinal Response) \cite{Erlingsson_CCS14} is an obfuscation mechanism implemented in Google Chrome browser. 
\colorB{Wang \textit{et al.} \cite{Wang_USENIX17} extended its simplest configuration called the basic one-time RAPPOR by generalizing two probabilities in perturbation. Here we call it the \textit{generalized RAPPOR} and describe its algorithm in detail.}

\colorB{The generalized} RAPPOR 
is an obfuscation mechanism 
with the input alphabet $\calX = \{x_1, x_2, \cdots, x_{|\calX|}\}$ and the output alphabet $\mathcal{Y} = \{0,1\}^{|\mathcal{X}|}$. 
It 
first deterministically maps $x_i \in \mathcal{X}$ to 
$e_i \in \{0,1\}^{|\mathcal{X}|}$, where $e_i$ is the $i$-th standard basis vector. 
It then 
probabilistically 
flips each bit of $e_i$ 
to obtain obfuscated data 
$y = (y_1, y_2, \cdots, y_{|\calX|}) \in \{0,1\}^{|\mathcal{X}|}$, 
where $y_i \in \{0,1\}$ is the $i$-th element of $y$. 
\colorB{Wang \textit{et al.} \cite{Wang_USENIX17} 
 compute $\epsilon$ from two parameters $\theta \in [0,1]$ (representing the probability of keeping $1$ unchanged) and $\psi \in [0,1]$ (representing the probability  of flipping $0$ into $1$).
In this paper, we compute $\psi$ from two parameters $\theta$ and $\epsilon$.}

\colorB{Specifically, given $\theta \in [0,1]$ and $\epsilon \in \nngreals$, the ($\theta,\epsilon$)\textit{-generalized RAPPOR}} 
maps $x_i$ to 
$y$ 
with the probability:
\begin{align*}
\bmQRAPPOR(y | x_i) &= \textstyle{\prod_{1 \le j \le |\calX|} \Pr(y_j | x_i),}
\end{align*}
where 
\colorB{$\Pr(y_j | x_i) = \theta$ if $i = j$ and $y_j=1$, and
$\Pr(y_j | x_i) = 1-\theta$ if $i = j$ and $y_j=0$, and
$\Pr(y_j | x_i) = \psi = \frac{\theta}{(1 - \theta) e^\epsilon + \theta}$ if $i \ne j$ and $y_j=1$, and
$\Pr(y_j | x_i) = 1- \psi$ otherwise.
The basic one-time RAPPOR \cite{Erlingsson_CCS14} is a special case of the generalized RAPPOR where $\theta = \frac{e^{\epsilon/2}}{e^{\epsilon/2} + 1}$.}
$\bmQRAPPOR$ also provides $\epsilon$-\LDP{}. 

\subsection{Distribution Estimation Methods}
\label{sub:distribution_estimation}
Here we explain 
the empirical estimation method \cite{Agrawal_SIGMOD05,Huang_ICDE08,Kairouz_ICML16} 
and the EM reconstruction method \cite{Agrawal_PODS01,Agrawal_SIGMOD05}. 
Both of them assume that the data collector knows the obfuscation mechanism $\bmQ$ used to generate $\bmY$ from $\bmX$. 

\smallskip
\noindent{\textbf{Empirical estimation method.}}~~The empirical estimation method 
\cite{Agrawal_SIGMOD05,Huang_ICDE08,Kairouz_ICML16} 
computes an empirical estimate $\hbmp$ of $\bmp$ 
using an empirical distribution $\hbmm$ of the obfuscated data $\bmY$. 
Note that $\hbmp$, $\hbmm$, and $\bmQ$ can be represented as 
an $|\calX|$-dimensional vector, $|\calY|$-dimensional vector, and $|\calX| \times |\calY|$ matrix, 
respectively. They have the following equation: 
\begin{align}
\hbmp \bmQ = \hbmm.
\label{eq:empirical_estimate}
\end{align}
The empirical estimation method computes $\hbmp$ by solving (\ref{eq:empirical_estimate}). 

\colorBB{Let $\bmm$ be the true distribution of obfuscated data; i.e., $\bmm = \bmp \bmQ$.} 
As the number of users $n$ increases, 
the empirical distribution $\hbmm$ converges to 
$\bmm$. 
Therefore, the empirical estimate $\hbmp$ also converges to $\bmp$. 
\colorB{However, when the number of users $n$ is small, many elements in $\hbmp$ can be negative. 
To address this issue, 
the studies in \cite{Erlingsson_CCS14,Wang_USENIX17} 
kept only estimates above a significance threshold determined via Bonferroni correction, and
discarded the remaining estimates.}

\smallskip
\noindent{\textbf{EM reconstruction method.}}~~The 
EM (Expectation-Maximization) reconstruction method  \cite{Agrawal_PODS01,Agrawal_SIGMOD05} 
(also called the iterative Bayesian technique \cite{Agrawal_SIGMOD05}) 
regards $\bmX$ as a hidden variable 
and estimates $\bmp$ from $\bmY$ using the EM %(Expectation-Maximization) 
algorithm \cite{learning} 
(for details of the algorithm, see \cite{Agrawal_PODS01,Agrawal_SIGMOD05}). 
Let $\hbmp_{EM}$ be an estimate of $\bmp$ by the EM reconstruction method. 
The feature of this algorithm is that 
$\hbmp_{EM}$ is equal to the maximum likelihood estimate in the probability simplex $\calC$ (see \cite{Agrawal_PODS01} for the proof). 
Since this property holds irrespective of the number of users $n$, the elements in $\hbmp_{EM}$ are always non-negative. 

\smallskip
\colorB{In this paper, 
our theoretical analysis uses the empirical estimation method 
for simplicity, 
while 
our experiments use 
the empirical estimation method, 
the one with the significance threshold, and 
the EM reconstruction method.} 

\section{Utility-Optimized LDP (ULDP)}
\label{sec:ULDP}
In this section, we focus on the common-mechanism scenario 
(outlined in Section~\ref{sub:notations}) 
and introduce 
ULDP (Utility-optimized Local Differential Privacy), 
which provides a privacy guarantee equivalent to $\epsilon$-\LDP{} only for sensitive data. 
Section~\ref{sub:ULDP_definition} provides the definition of ULDP. 
Section~\ref{sub:ULDP_theoretical} shows 
some theoretical properties of ULDP.

\subsection{Definition}
\label{sub:ULDP_definition}

Figure~\ref{fig:overview_rest} 
shows an overview of 
ULDP. 
An obfuscation mechanism 
providing ULDP, 
which we call the utility-optimized mechanism, 
divides obfuscated data into \textit{protected data} and \textit{invertible data}. 
Let $\calYP$ be a set of protected data, and $\calYI = \calY \setminus \calYP$ be a set of invertible data. 

\begin{figure}
\centering
\includegraphics[width=0.72\linewidth]{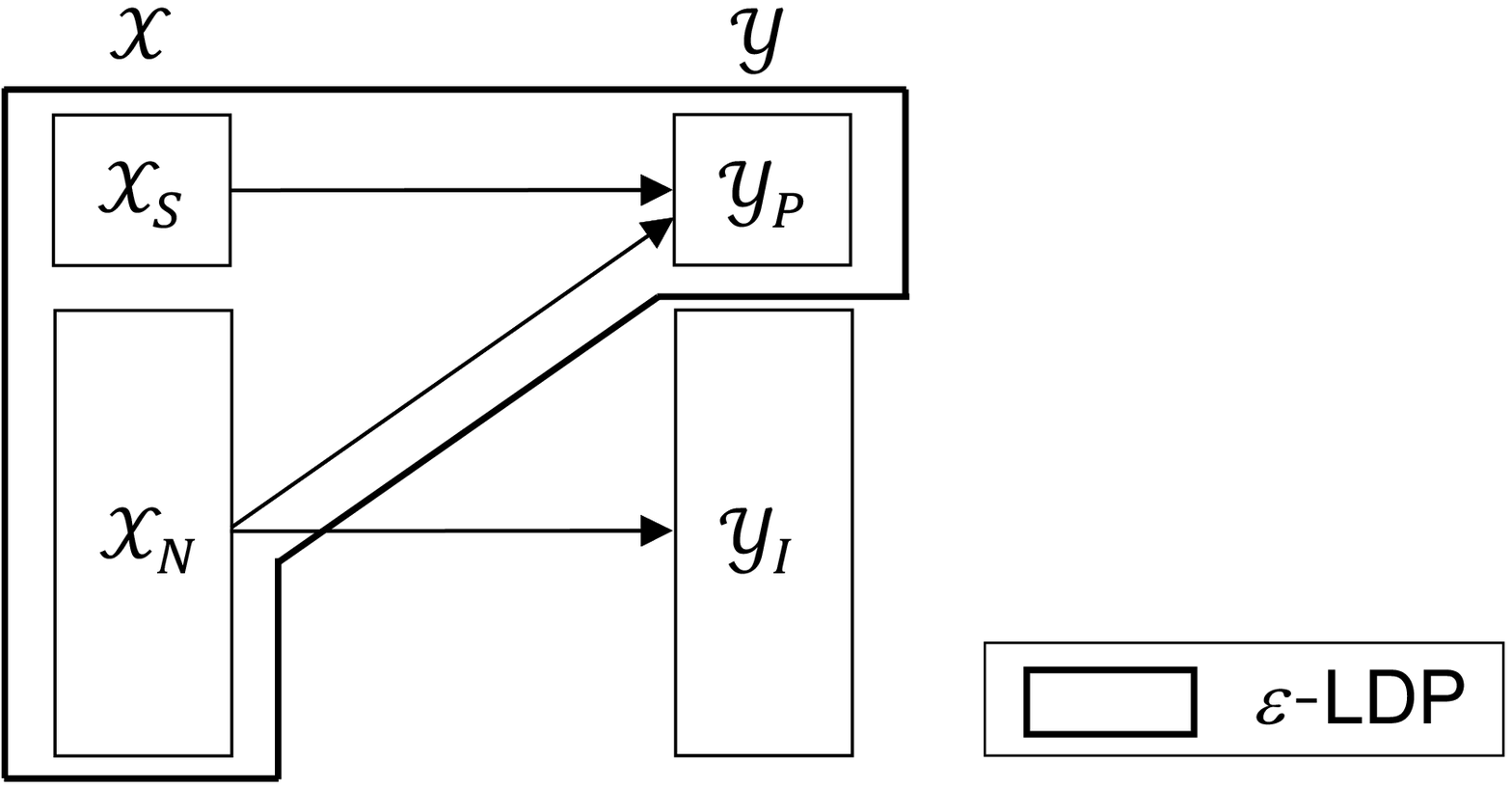}
\vspace{-3mm}
\caption{Overview of ULDP. 
It has no transitions from $\calXS$ to $\calYI$, and 
every output in $\calYI$ reveals the corresponding input in $\calXN$. 
It also provides $\epsilon$-\LDP{} for $\calYP$. 
}
\label{fig:overview_rest}
\end{figure}

The feature of 
the utility-optimized mechanism 
is that it maps 
sensitive data $x \in \calXS$ to only protected data $y \in \calYP$. 
In other words, 
\emph{it restricts the output set, given the input $x \in \calXS$, to $\calYP$}. 
Then it 
provides $\epsilon$-\LDP{} for $\calYP$; 
i.e., $\bmQxy \leq e^\epsilon \bmQxdy$ 
for any $x,x' \in \calX$ and any $y \in \calYP$. 
By this property, 
a privacy guarantee equivalent to $\epsilon$-\LDP{} is provided for any sensitive data $x \in \calXS$, 
since the output set corresponding to $\calXS$ is restricted to $\calYP$. 
In addition, 
every output in $\calYI$ reveals the corresponding input in $\calXN$ 
(as in Mangat's randomized response \cite{Mangat_JRSS94}) 
to optimize the estimation accuracy.

We now formally define ULDP and the utility-optimized mechanism:
\begin{definition} [$(\calXS,\calYP,\epsilon)$-\ULDP{}] \label{def:rest} 
Given $\calXS \subseteq \calX$, $\calYP \subseteq \calY$, 
and 
$\epsilon \in \nngreals$, 
an obfuscation mechanism $\bmQ$ 
from $\calX$ to $\calY$ 
provides $(\calXS,\calYP,\epsilon)$-\ULDP{} 
if it 
satisfies the following properties:
\begin{enumerate}
\item For any $y \in \calYI$, there exists 
an $x \in \calXN$ such that 
\begin{align}
\bmQxy > 0 ~~\text{and}~~\bmQxdy = 0~~\text{for any } x'\neq x.
\label{eq:xs_ys_0}
\end{align}
\item For any $x,x' \in \calX$ and any $y \in \calYP$,
\begin{align}
\bmQxy \leq e^\epsilon \bmQxdy.
\label{eq:epsilon_OSLDP_whole}
\end{align}
\end{enumerate}
We refer to an obfuscation mechanism $\bmQ$ providing 
$(\calXS, \calYP,$ $\epsilon)$-\ULDP{} 
as the $(\calXS,\calYP,\epsilon)$-utility-optimized mechanism. 
\end{definition}

\noindent{\textbf{Example.}}~~For an intuitive understanding of Definition~\ref{def:rest}, 
we show that Mangat's randomized response \cite{Mangat_JRSS94} provides $(\calXS,\calYP,\epsilon)$-\ULDP{}. 
As described in Section~\ref{sec:intro}, 
this mechanism considers binary alphabets (i.e., $\calX = \calY = \{0,1\}$), and regards the value $1$ as sensitive (i.e., $\calXS = \calYP =\{1\}$). 
If the input value is $1$, it always reports $1$ as output. 
Otherwise, 
it reports $1$ and $0$ with probability $p$ and $1-p$, respectively. 
Obviously, this mechanism does not provide $\epsilon$-\LDP{} for any $\epsilon \in [0,\infty)$.
However, it provides $(\calXS, \calYP, \ln \frac{1}{p})$-\ULDP{}.

\smallskip
$(\calXS,\calYP,\epsilon)$-\ULDP{} provides a privacy guarantee equivalent to $\epsilon$-\LDP{} for any sensitive data $x \in \calXS$, as explained above. 
On the other hand, 
no privacy guarantees are provided for non-sensitive data $x \in \calXN$ 
because every output in $\calYI$ reveals the corresponding input in $\calXN$.
However, it does not matter since non-sensitive data need not be protected. 
Protecting only minimum necessary data is the key to 
achieving locally private distribution estimation with high data utility. 

We can apply any $\epsilon$-\LDP{} mechanism to the sensitive data in $\calXS$ to 
provide $(\calXS,\calYP,\epsilon)$-\ULDP{} as a whole. 
In Sections~\ref{sub:restRR} and \ref{sub:restRAPPOR}, we 
propose a utility-optimized RR 
(Randomized Response) 
and utility-optimized RAPPOR, which 
apply the $\epsilon$-RR and $\epsilon$-RAPPOR, respectively, 
to the sensitive data $\calXS$. 

\conference{In Appendix~\ref{sec:ULDP_OSLDP}, 
we also consider OSLDP (One-sided LDP), a local model version of OSDP introduced in a preprint \cite{Doudalis_arXiv17}, and explain the reason for using ULDP in this paper.}

\colorBB{It might be better to generalize ULDP so that different levels of $\epsilon$ can be assigned to different sensitive data.
We leave introducing such granularity as future work.}

\smallskip
\noindent{\textbf{Remark.}}~~It should also be noted that 
the data collector needs to know 
$\bmQ$ to estimate $\bmp$ from $\bmY$ (as described in Section~\ref{sub:distribution_estimation}), and 
that the $(\calXS,\calYP,\epsilon)$-utility-optimized mechanism $\bmQ$ itself 
includes 
the information on  \emph{what is sensitive for users} 
(i.e., 
the data collector learns
whether each $x \in \calX$ belongs to $\calXS$ or not 
by checking the values of $\bmQ(y|x)$ for all $y \in \calY$). 
This does not matter in the common-mechanism scenario, 
since the set $\calXS$ of sensitive data is common to all users 
(e.g., public hospitals). 
However, in the personalized-mechanism scenario, 
the $(\calXS \cup \calXS\uid{i},\calYP,\epsilon)$-utility-optimized mechanism $\bmQ\uid{i}$, 
which expands the set $\calXS$ of personal data to $\calXS \cup \calXS\uid{i}$, 
includes 
the information on 
\emph{what is sensitive for the $i$-th user}. 
Therefore, 
the data collector learns whether each $x \in \calXN$ belongs to $\calXS\uid{i}$ or not by 
checking the values of $\bmQ\uid{i}(y|x)$ for all $y \in \calY$, 
despite the fact that the $i$-th user wants to hide her user-specific sensitive data $\calXS\uid{i}$  (e.g., home, workplace). 
We address this issue in Section~\ref{sec:per_rest}. 

\subsection{Basic Properties of ULDP}
\label{sub:ULDP_theoretical}
Previous work showed some basic properties of differential privacy (or its variant), such as compositionality~\cite{DP} 
and immunity to post-processing~\cite{DP}. 
We briefly explain theoretical properties of ULDP including 
the ones above.

\smallskip
\noindent{\textbf{Sequential composition.}}~~ULDP 
is preserved under adaptive sequential composition 
when the composed obfuscation mechanism maps sensitive data to pairs of protected data. 
Specifically, 
consider two mechanisms 
$\bmQ_0$ from $\calX$ to $\calY_0$ and 
$\bmQ_1$ from $\calX$ to $\calY_1$ such that 
$\bmQ_0$ (resp.~$\bmQ_1$) maps sensitive data $x\in\calXS$ to protected data $y_0\in\calY_{0P}$ (resp.~$y_1\in\calY_{1P}$). 
Then the sequential composition of $\bmQ_0$ and $\bmQ_1$ maps sensitive data $x\in\calXS$ to pairs $(y_0, y_1)$ of protected data ranging over:
\begin{align*}
(\calY_0\times\calY_1)_P &= 
 \left\{ (y_0,y_1)\in\calY_0\times\calY_1
 \mid y_0\in\calY_{0P} \mbox{ and } y_1\in\calY_{1P}
 \right\}. 
\end{align*}
Then we obtain the following compositionality.

\begin{restatable}[Sequential composition]{prop}{SequentialComposition}
\label{prop:composition_main}
Let $\varepsilon_0,\varepsilon_1 \ge 0$.
If $\bmQ_0$ provides $(\calXS,\calY_{0P},\varepsilon_0)$-\ULDP{} and 
$\bmQ_1(y_0)$ provides $(\calXS,$ $\calY_{1P},\varepsilon_1)$-\ULDP{} 
for each $y_0\in\calY_0$, then the sequential composition of $\bmQ_0$ and $\bmQ_1$ provides $(\calXS,(\calY_0\times\calY_1)_P,\varepsilon_0+\varepsilon_1)$-\ULDP{}.
\end{restatable}

For example, 
if we apply an obfuscation mechanism providing $(\calXS,\calYP,\epsilon)$-\ULDP{} for $t$ times, 
then we obtain 
\arxiv{$(\calXS,\allowbreak(\calYP)^t,\epsilon t)$-\ULDP{}}\conference{$(\calXS,(\calYP)^t,\epsilon t)$-\ULDP{}} 
in total (this is derived by repeatedly using Proposition~\ref{prop:composition_main}).

\smallskip
\noindent{\textbf{Post-processing.}}~~ULDP 
is immune to the post-processing by a randomized algorithm 
\colorBB{that \textit{preserves data types}: protected data or invertible data.}
\colorBB{Specifically, if a mechanism $\bmQ_0$ provides $(\calXS,\calYP,\varepsilon)$-\ULDP{} and a randomized algorithm $\bmQ_1$ maps protected data over $\calYP$ (resp. invertible data) to protected data over $\calZP$ (resp. invertible data), then the composite function $\bmQ_1\circ\bmQ_0$ provides $(\calXS,\calZP,\varepsilon)$-\ULDP{}.} 

\colorBB{Note that 
$\bmQ_1$ needs to preserve data types for utility; i.e., to make all $y \in \calYI$ invertible 
(as in Definition~\ref{def:rest}) 
after post-processing.
The DP guarantee for $y \in \calYP$ is preserved by any post-processing algorithm. 
See Appendix~\ref{sub:post-process} for details.}

\smallskip
\noindent{\colorB{\textbf{Compatibility with LDP.}}}~~\colorBB{Assume that data collectors A and B adopt 
a mechanism 
providing 
\ULDP{} 
and 
a mechanism 
providing 
\LDP{}, respectively. 
In this case, 
all protected data in the data collector A 
can be 
combined with 
all obfuscated data in the data collector B (i.e., data integration) 
to perform data analysis under LDP. 
See Appendix~\ref{sub:compatibility} for details.}

\smallskip
\noindent{\textbf{Lower bounds on the $l_1$ and $l_2$ losses.}}~~We present lower bounds on the $l_1$ and $l_2$ losses of any \ULDP{} mechanism by using the fact that \ULDP{} provides 
(\ref{eq:epsilon_OSLDP_whole}) for any $x,x' \in \calXS$ and any $y \in \calYP$. 
Specifically, Duchi \textit{et al.} \cite{Duchi_arXiv13} showed that for $\epsilon \in [0,1]$, the lower bounds on the $l_1$ and $l_2$ losses (minimax rates) of any $\epsilon$-\LDP{} mechanism can be expressed as $\Theta(\frac{|\calX|}{\sqrt{n \epsilon^2}})$ and $\Theta(\frac{|\calX|}{n \epsilon^2})$, respectively. 
By directly applying these bounds to $\calXS$ and $\calYP$, 
the lower bounds on the $l_1$ and $l_2$ losses of any $(\calXS,\calYP,\epsilon)$-\ULDP{} mechanisms for $\epsilon \in [0,1]$ 
can be expressed as $\Theta(\frac{|\calXS|}{\sqrt{n \epsilon^2}})$ and $\Theta(\frac{|\calXS|}{n \epsilon^2})$, respectively. 
In Section~\ref{sub:utility_analysis}, 
we show that our utility-optimized RAPPOR achieves these lower bounds 
when $\epsilon$ is close to $0$ (i.e., high privacy regime).

\section{Utility-Optimized Mechanisms}
\label{sec:rest_mechanism}
In this section, we focus on the common-mechanism scenario 
and propose 
the \emph{utility-optimized RR (Randomized Response)} 
and \emph{utility-optimized RAPPOR} (Sections~\ref{sub:restRR} and \ref{sub:restRAPPOR}). 
We then analyze the data utility of these mechanisms (Section~\ref{sub:utility_analysis}). 

\subsection{Utility-Optimized Randomized Response}
\label{sub:restRR}
We propose the utility-optimized RR,  %(Randomized Response)}, 
which is a generalization of Mangat's randomized response \cite{Mangat_JRSS94} to $|\calX|$-ary alphabets with $|\calXS|$ sensitive symbols. 
As with the RR, the output range of the utility-optimized RR is identical to the input domain; i.e., $\calX = \calY$. 
In addition, we divide the output set 
in the same way as the input set; 
i.e., $\calXS = \calYP$, $\calXN = \calYI$. 

Figure~\ref{fig:rest_RR} shows the utility-optimized RR with $\calXS = \calYP = \{x_1, x_2, \allowbreak x_3\}$ and 
$\calXN = \calYI = \{x_4, x_5, x_6\}$. 
The utility-optimized RR applies the $\epsilon$-RR to 
$\calXS$.
It 
maps $x \in \calXN$ to $y \in \calYP$ ($=\calXS$) with the probability $\bmQxy$ so that (\ref{eq:epsilon_OSLDP_whole}) is satisfied, and maps $x \in \calXN$ to itself with the remaining probability. 
Formally, we define the utility-optimized RR 
\colorBB{(uRR)} as follows:

\begin{figure}
\centering
\includegraphics[width=0.82\linewidth]{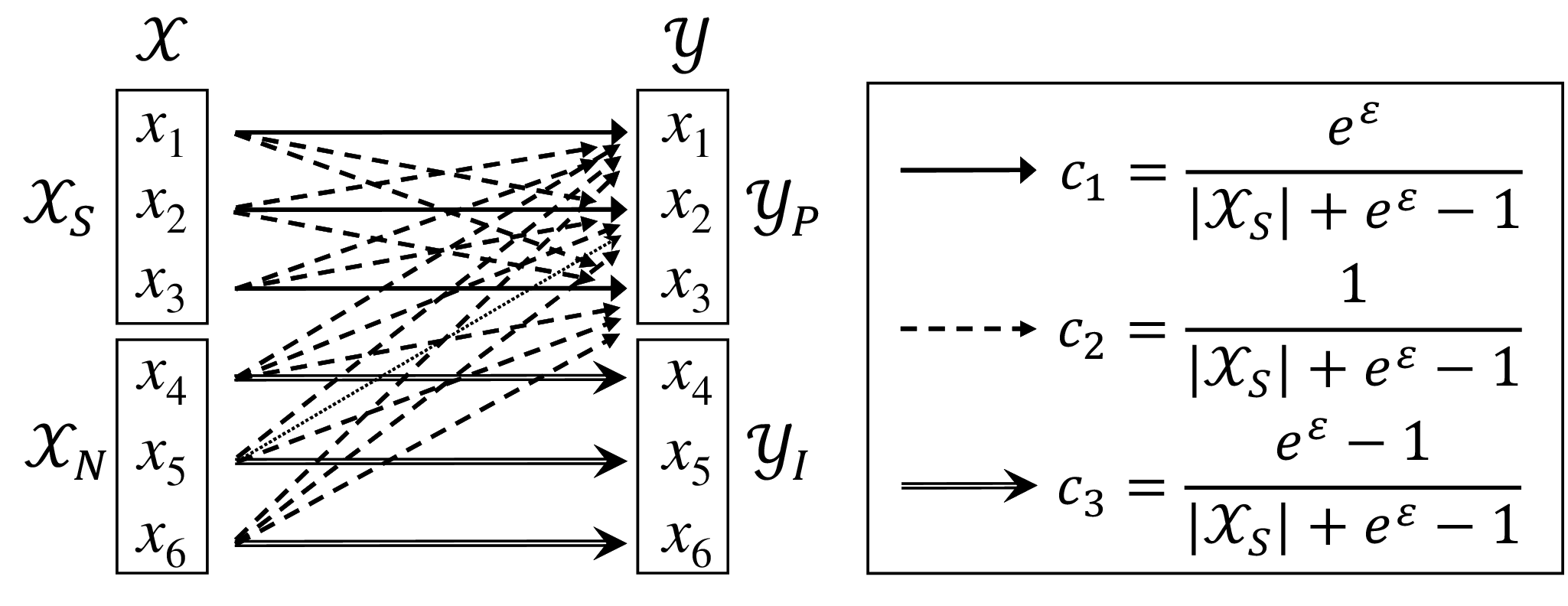}
\vspace{-5mm}
\caption{Utility-optimized RR 
in the case where $\calXS = \calYP =\{x_1, x_2, x_3\}$ and $\calXN = \calYI =\{x_4, x_5, x_6\}$.} 
\label{fig:rest_RR}
\end{figure}

\begin{definition}[$(\calXS,\epsilon)$-utility-optimized RR]\label{def:restRR}
Let $\calXS \subseteq \calX$ and $\epsilon \in \nngreals$. 
Let $c_1 = \frac{e^{\epsilon}}{|\calXS| + e^{\epsilon} - 1}$,
$c_2 = \frac{1}{|\calXS| + e^{\epsilon} - 1}$, and
\colorB{$c_3 = 1 - |\calXS|c_2 
= \frac{e^\epsilon - 1}{|\calXS| + e^{\epsilon} - 1}$}.
Then the $(\calXS,\epsilon)$-utility-optimized RR 
\colorBB{(uRR)} 
is an obfuscation mechanism that 
maps $x \in \calX$ to $y \in \calY$ ($=\calX$) 
with the probability $\bmQrestRR(y | x)$ defined as follows:

\begin{align}
\colorBB{\bmQrestRR(y | x) = 
\begin{cases}
 c_1 & \text{(if $x \in \calXS$, $y = x$)}\\
 c_2 & \text{(if $x \in \calXS$, $y \in \calXS\setminus\{x\}$)}\\
 c_2 & \text{(if $x \in \calXN$, $y \in \calXS$)}\\
 c_3 & \text{(if $x \in \calXN$, $y = x$)}\\
 0 & \text{(otherwise)}.
\end{cases}} 
\label{eq:restRR_XS}
\end{align}

\end{definition}

\begin{restatable}{prop}{ULDPrestRR}
\label{prop:restRR_ULDP} 
The 
$(\calXS,\epsilon)$-\colorBB{uRR} 
provides $(\calXS,\calXS,$ $\epsilon)$-\ULDP{}.
\end{restatable}

\subsection{Utility-Optimized RAPPOR}
\label{sub:restRAPPOR}
Next, we 
propose the utility-optimized RAPPOR 
with the input alphabet $\calX = \{x_1, x_2, \cdots, x_{|\calX|}\}$ and the output alphabet $\mathcal{Y} = \{0,1\}^{|\calX|}$. 
Without loss of generality, 
we assume 
that $x_1, \cdots, x_{|\calXS|}$ 
are 
sensitive 
and $x_{|\calXS|+1}, \cdots, x_{|\calX|}$ 
are 
non-sensitive; 
i.e., 
$\calXS = \{x_1, \cdots, x_{|\calXS|}\}$, 
\arxiv{$\calXN = \{x_{|\calXS|+1}, \allowbreak \cdots, x_{|\calX|}\}$.}\conference{$\calXN = \{x_{|\calXS|+1}, \cdots, x_{|\calX|}\}$.} 

Figure~\ref{fig:rest_RAP} shows the utility-optimized RAPPOR with $\calXS =\{x_1, \cdots, \allowbreak x_4\}$ and $\calXN =\{x_5, \cdots, x_{10}\}$. 
The utility-optimized RAPPOR first deterministically 
maps $x_i \in \calX$ to the $i$-th standard basis vector $e_i$. 
It should be noted that 
if $x_i$ is sensitive data (i.e., $x_i \in \calXS$), 
then the last $|\calXN|$ elements in $e_i$ are always zero 
(as shown in the upper-left panel of Figure~\ref{fig:rest_RAP}). 
Based on this fact, the utility-optimized RAPPOR 
regards obfuscated data $y = (y_1, y_2, \ldots, y_{|\calX|}) \in \{0,1\}^{|\calX|}$ such that 
$y_{|\calXS|+1} = \cdots = y_{|\calX|} = 0$ as 
protected 
data; i.e., 
\begin{align}
\calYP &= \{ (y_1, \ldots, y_{|\calXS|}, 0, \cdots, 0) |  y_1, \ldots, y_{|\calXS|} \in \{0,1\} \}. \label{eq:rRAP_YS}%\\
\end{align}
Then it applies 
the ($\theta,\epsilon$)-generalized RAPPOR 
to 
$\calXS$, 
and maps $x \in \calXN$ to $y \in \calYP$ (as shown in the lower-left panel of Figure~\ref{fig:rest_RAP}) with the probability $\bmQxy$ so that (\ref{eq:epsilon_OSLDP_whole}) is satisfied. 
We 
formally define the utility-optimized RAPPOR \colorBB{(uRAP)}:

\begin{figure}
\centering
\includegraphics[width=0.9\linewidth]{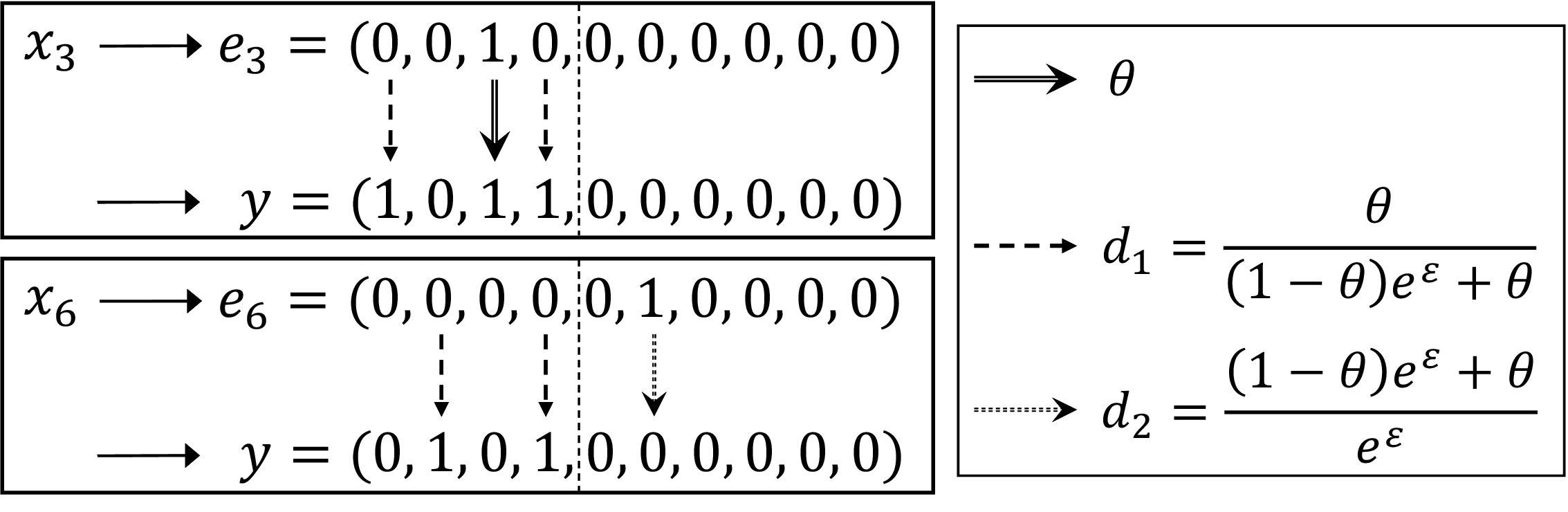}
\vspace{-4mm}
\caption{Utility-optimized RAPPOR in the case where $\calXS =\{x_1, \cdots, x_4\}$ and $\calXN =\{x_5, \cdots, x_{10}\}$.} 
\label{fig:rest_RAP}
\end{figure}

\begin{definition}[$(\calXS,\colorB{\theta,}\epsilon)$-utility-optimized RAPPOR]\label{def:restRAPPOR}
Let $\calXS \subseteq \calX$, 
\colorB{$\theta \in [0,1]$,} 
and $\epsilon \in \nngreals$. 
Let 
\colorB{$d_1 = \frac{\theta}{(1 - \theta) e^\epsilon + \theta}$, 
$d_2 = \frac{(1 - \theta) e^\epsilon + \theta}{e^\epsilon}$}. 
Then the $(\calXS,\colorB{\theta,}\epsilon)$-utility-optimized RAPPOR \colorBB{(uRAP)} is an obfuscation mechanism that 
maps $x_i \in \calX$ to $y \in \calY = \{0,1\}^{|\calX|}$ 
with the probability $\bmQrestRAPPOR(y | x)$ given by:
\begin{align}
\bmQrestRAPPOR(y | x_i) &= \textstyle{\prod_{1 \le j \le |\calX|} \Pr(y_j | x_i),}
\label{eq:restRAPPOR}
\end{align}
where $\Pr(y_j | x_i)$ is written as follows:
\begin{enumerate}
\renewcommand{\labelenumi}{(\roman{enumi})}
\item if $1 \leq j \leq |\calXS|$:
\begin{align}
\Pr(y_j | x_i) &= 
\colorB{
\begin{cases}
 1 - \theta & \text{(if $i = j$, $y_j = 0$)}\\
 \theta & \text{(if $i = j$, $y_j = 1$)}\\
 1 - d_1 & \text{(if $i \neq j$, $y_j = 0$)}\\
 d_1 & \text{(if $i \neq j$, $y_j = 1$)}.
\end{cases} 
 }
\label{eq:restRAPPOR_calXS}
\end{align}
\item if $|\calXS|+1 \leq j \leq |\calX|$:
\begin{align}
\Pr(y_j | x_i) &= 
\begin{cases}
 d_2 & \text{(if $i = j$, $y_j = 0$)}\\
 1 - d_2 & \text{(if $i = j$, $y_j = 1$)}\\
 1 & \text{(if $i \neq j$, $y_j = 0$)}\\
 0 & \text{(if $i \neq j$, $y_j = 1$)}.\\
\end{cases} 
\label{eq:restRAPPOR_calXN}
\end{align}
\end{enumerate}
\end{definition}

\begin{restatable}{prop}{ULDPrestRAP}
\label{prop:restRAPPOR_ULDP} 
The $(\calXS,\theta,\epsilon)$-\colorBB{uRAP} 
provides $(\calXS,\calYP,\epsilon)$-\ULDP{}, where 
$\calYP$ is given by (\ref{eq:rRAP_YS}).
\end{restatable}
\colorB{
Although we used the generalized RAPPOR in $\calXS$ and $\calYP$ in Definition~\ref{def:restRAPPOR}, hereinafter 
we set $\theta = \frac{e^{\epsilon/2}}{e^{\epsilon/2} + 1}$ in the same way as the original RAPPOR \cite{Erlingsson_CCS14}. 
There are two reasons for this. 
First, 
it achieves ``order'' optimal data utility among all \colorBB{$(\calXS,\calYP,\epsilon)$-\ULDP{}} 
mechanisms in the high privacy regime, as shown in Section~\ref{sub:utility_analysis}. 
Second, it maps $x_i \in \calXN$ to $y \in \calYI$ with probability $1 - d_2 = 1 - e^{-\epsilon/2}$, which is close to $1$ when $\epsilon$ is large (i.e., low privacy regime). 
Wang \textit{et al.} \cite{Wang_USENIX17} showed that the generalized RAPPOR with parameter $\theta = \frac{1}{2}$ minimizes the variance of the estimate. However, our 
\colorBB{uRAP} 
with parameter $\theta = \frac{1}{2}$ maps $x_i \in \calXN$ to $y \in \calYI$ with probability 
$1 - d_2 = \frac{e^\epsilon - 1}{2 e^\epsilon}$ 
which is 
less than $1 - e^{-\epsilon/2}$ for any $\epsilon > 0$ and is 
less than $\frac{1}{2}$ even when $\epsilon$ goes to infinity. 
Thus, 
our 
\colorBB{uRAP} 
with $\theta = \frac{e^{\epsilon/2}}{e^{\epsilon/2} + 1}$ 
maps $x_i \in \calXN$ to $y \in \calYI$ with higher probability, and therefore achieves a smaller estimation error over all non-sensitive data.
We also consider that 
an optimal $\theta$ for our 
\colorBB{uRAP} 
is different from the optimal $\theta$ ($= \frac{1}{2}$) for the generalized RAPPOR.
We leave 
finding the optimal 
$\theta$ for 
our 
\colorBB{uRAP} 
(with respect to 
the estimation error 
over all personal data) 
as future work.}

We 
refer to 
the 
$(\calXS,\theta,\epsilon)$-\colorBB{uRAP} 
with $\theta = \frac{e^{\epsilon/2}}{e^{\epsilon/2} + 1}$ 
in shorthand 
as the \textit{$(\calXS,\epsilon)$-\colorBB{uRAP}}.

\subsection{Utility Analysis}
\label{sub:utility_analysis}

We 
evaluate 
the $l_1$ loss 
of 
the \colorBB{uRR} 
and 
\colorBB{uRAP} 
when the empirical estimation method is used for 
distribution estimation\footnote{\colorBB{We note that we use the empirical estimation method in the same way as \cite{Kairouz_ICML16}, and that it might be possible that 
other mechanisms have better utility with a different estimation method. 
However, we emphasize that even with the empirical estimation method, 
the \colorBB{uRAP} 
achieves the lower bounds on the $l_1$ and $l_2$ losses of any \ULDP{} mechanisms 
when $\epsilon \approx 0$, 
and 
the \colorBB{uRR} and  \colorBB{uRAP} 
achieve almost the same utility as a non-private mechanism 
when $\epsilon = \ln |\calX|$ 
and most of the data are non-sensitive.}}. 
In particular, we evaluate the $l_1$ loss
when $\epsilon$ is close to $0$ (i.e., high privacy regime) and 
$\epsilon = \ln |\calX|$ (i.e., low privacy regime). 
Note that 
ULDP provides a natural interpretation of the latter value of $\epsilon$. 
Specifically, it follows from (\ref{eq:epsilon_OSLDP_whole}) that 
if $\epsilon = \ln |\calX|$, then 
for any $x \in \calX$, 
the likelihood that the input data is $x$ is almost equal to the sum of the likelihood that the input data is $x' \neq x$. 
This is consistent with the fact that 
the $\epsilon$-RR with parameter $\epsilon = \ln |\calX|$ 
sends true data (i.e., $y = x$ in (\ref{eq:RR})) with probability about $0.5$ and 
false data (i.e., $y \neq x$) 
with probability about $0.5$, 
and hence provides plausible deniability \cite{Kairouz_ICML16}.

\smallskip
\noindent{\textbf{\colorBB{uRR} 
in the general case.}}~~We 
begin with 
the \colorBB{uRR:}

\begin{restatable}[$l_1$ loss of the \colorBB{uRR}]{prop}{propLoneRestRR}
\label{prop:l1_restRR}
Let $\epsilon \in \nngreals$,
$u = |\calXS| + e^\epsilon - 1$, $u'=e^\epsilon-1$, and $v = \frac{u}{u'}$.
Then the expected $l_1$ loss of the %$(\calXS,\epsilon)$-utility-optimized RR 
$(\calXS,\epsilon)$-\colorBB{uRR} 
mechanism is given by:
\begin{align}
\bbE \left[ l_1(\hbmp,\bmp) \right]
\approx& {\textstyle \sqrt{\!\frac{2}{n\pi}}} \biggl( \sum_{x \in \calXS} \sqrt{\bigl( \bmp(x) + 1/u' \bigr) \bigl( v - \bmp(x) - 1/u' \bigr)} \nonumber\\[-1.5ex]
&\hspace{7ex} +\!\sum_{x \in \calXN}\hspace{-1.5ex} \sqrt{\bmp(x) \bigl( v - \bmp(x) \bigr)} \biggr),
\label{eq:restRR-l1-loss:formal}
\end{align}
where $f(n)\approx g(n)$ 
represents 
$\lim_{n\rightarrow\infty} f(n)/g(n) = 1$.
\end{restatable}

Let $\bmpUN$ be the uniform distribution over $\calXN$; i.e., 
for any $x \in \calXS$, $\bmpUN(x) = 0$, 
and for any $x \in \calXN$, $\bmpUN(x) = \frac{1}{|\calXN|}$.
Symmetrically, let $\bmpUS$ be the uniform distribution over $\calXS$.

For $0 < \epsilon < \ln(|\calXN|+1)$, 
the $l_1$ loss is maximized by~$\bmpUN$: 

\begin{restatable}{prop}{propLoneRestRRsmall}
\label{prop:l1_restRR:eps=0}
For any $0 < \epsilon < \ln(|\calXN|+1)$ and $|\calXS| \le |\calXN|$, 
(\ref{eq:restRR-l1-loss:formal})
is maximized by $\bmpUN$:
\begin{align}
&\bbE \left[ l_1(\hbmp,\bmp) \right] 
\lesssim \bbE \left[ l_1(\hbmp,\bmpUN) \right] \nonumber\\
&=\!{\textstyle \sqrt{\!\frac{2}{n\pi}}} \biggl(\! {\textstyle \frac{|\calXS|\sqrt{ |\calXS| + e^{\epsilon} - 2 }}{e^{\epsilon} - 1} + \sqrt{ \frac{|\calXS||\calXN|}{e^{\epsilon} - 1} + |\calXN| - 1 } } \biggr), 
\label{eq:restRR-l1-loss:small_espilon}
\end{align}
where $f(n)\lesssim g(n)$ 
represents 
$\lim_{n\rightarrow\infty} f(n)/g(n) \leq 1$.
\end{restatable}

For $\epsilon \ge \ln(|\calXN|+1)$, 
the $l_1$ loss is maximized by a mixture distribution of $\bmpUN$ and $\bmpUS$:

\begin{restatable}{prop}{propLoneRestRRbig}
\label{prop:l1_restRR:eps=lnX}
Let $\bmp^*$ be a distribution over $\calX$ defined by:
\begin{align}
\bmp^*(x) =
\begin{cases}
\frac{1 - |\calXN|/(e^{\epsilon}-1)}{|\calXS| + |\calXN|} 
~~\mbox{(if $x\in\calXS$)}
\\[2ex]
\frac{1 + |\calXS|/(e^{\epsilon}-1)}{|\calXS| + |\calXN|}
~~\mbox{(otherwise)}
\end{cases}
\label{eq:restRR-l1-loss:p_ast}
\end{align}
Then for any $\epsilon \ge \ln(|\calXN|+1)$,
(\ref{eq:restRR-l1-loss:formal})
is maximized by $\bmp^*$:
\begin{align}
\bbE \left[ l_1(\hbmp,\bmp) \right] 
\lesssim \bbE \left[ l_1(\hbmp,\bmp^*) \right]
= {\textstyle \sqrt{\frac{2(|\calX| - 1)}{n\pi}} \cdot \frac{|\calXS|+e^{\epsilon}-1}{e^{\epsilon}-1} },
\label{eq:restRR-l1-loss:large_espilon}
\end{align}
where $f(n)\lesssim g(n)$ represents $\lim_{n\rightarrow\infty} f(n)/g(n) \leq 1$.
\end{restatable}

Next, we instantiate the $l_1$ loss in the high and low privacy regimes based on these propositions. 

\smallskip
\noindent{\textbf{\colorBB{uRR} 
in the high privacy regime.}}~~When $\epsilon$ is close to $0$, 
we have $e^\epsilon - 1 \approx \epsilon$. 
Thus, 
the right-hand side of (\ref{eq:restRR-l1-loss:small_espilon}) 
in Proposition~\ref{prop:l1_restRR:eps=0} 
can be simplified 
as follows: 
\begin{align}
\bbE \left[ l_1(\hbmp,\bmpUN) \right] %\nonumber\\[0.5ex]
&\approx\! {\textstyle \sqrt{\frac{2}{n\pi}} \cdot \frac{|\calXS|\sqrt{ |\calXS| - 1 }}{\epsilon}}.
\label{eq:restRR-better-than-RR_main}
\end{align}
It was shown in \cite{Kairouz_ICML16} that 
the expected $l_1$ loss of the $\epsilon$-RR is at most 
$\sqrt{\frac{2}{n\pi}} \frac{|\calX|\sqrt{ |\calX| - 1 }}{\epsilon}$ when $\epsilon \approx 0$. 
\colorBB{The right-hand side of (\ref{eq:restRR-better-than-RR_main}) is much smaller than this when $|\calXS| \ll |\calX|$. 
Although both of them are ``upper-bounds'' of the expected $l_1$ losses, we show that the total variation of the $(\calXS,\epsilon)$-\colorBB{uRR} 
is also much smaller than that of the $\epsilon$-RR when $|\calXS| \ll |\calX|$ in Section~\ref{sec:exp}.}

\smallskip
\noindent{\textbf{\colorBB{uRR} 
in the low privacy regime.}}~~When $\epsilon=\ln|\calX|$ 
and $|\calXS| \ll |\calX|$, 
the right-hand side of (\ref{eq:restRR-l1-loss:large_espilon}) 
in Proposition~\ref{prop:l1_restRR:eps=lnX} 
can be simplified by using $|\calXS| / |\calX| \approx 0$: 
\begin{align*}
\bbE \left[ l_1(\hbmp,\bmp^*) \right] 
&\approx {\textstyle \sqrt{\frac{2(|\calX| - 1)}{n\pi}}}. %\label{eq:l1_restRR_large_epsilon_np}
\end{align*}
It should be noted that the expected $l_1$ loss of the non-private mechanism, which does not obfuscate the personal data at all, is at most $\sqrt{\frac{2(|\calX| - 1)}{n\pi}}$ \cite{Kairouz_ICML16}. 
Thus, when $\epsilon=\ln|\calX|$ and $|\calXS| \ll |\calX|$,
the 
$(\calXS,\epsilon)$-\colorBB{uRR} 
achieves almost the same data utility as the non-private mechanism, 
whereas the expected $l_1$ loss of the $\epsilon$-RR is 
twice 
larger than that of the non-private mechanism \cite{Kairouz_ICML16}.

\smallskip
\noindent{\textbf{\colorBB{uRAP} 
in the general case.}}~~We then 
\colorBB{analyze} 
the 
\colorBB{uRAP:} 

\begin{restatable}[$l_1$ loss of the \colorBB{uRAP}]{prop}{propLoneRestRAP}
\label{prop:l1_restRAPPOR}
Let $\epsilon \in \nngreals$,
$u' = e^{\epsilon/2} - 1$, 
and $v_N = \frac{e^{\epsilon/2}}{e^{\epsilon/2}-1}$.
The expected $l_1$-loss of the $(\calXS,\epsilon)$-\colorBB{uRAP} 
mechanism is:
\begin{align}
\bbE \left[ l_1(\hbmp,\bmp) \right]
\approx& {\textstyle \sqrt{\frac{2}{n\pi}}} \biggl( \sum_{j=1}^{|\calXS|} \sqrt{\bigl( \bmp(x_j) + 1/u' \bigr) \bigl( v_N - \bmp(x_j) \bigr)} \nonumber\\[-1.5ex]
&\hspace{6.5ex} +\hspace{-1.5ex}\sum_{j=|\calXS|+1}^{|\calX|}\hspace{-2.5ex}\sqrt{\bmp(x_j) \bigl( v_N - \bmp(x_j) \bigr)} \biggr),
\label{eq:restRAPPOR-l1-loss:formal}
\end{align}
where $f(n)\approx g(n)$ represents $\lim_{n\rightarrow\infty} f(n)/g(n) = 1$.
\end{restatable}

When $0 < \epsilon < 2\ln(\frac{|\calXN|}{2}+1)$, the $l_1$ loss is maximized by the uniform distribution $\bmpUN$ over $\calXN$: 

\begin{restatable}{prop}{propLoneRestRAPapprox}
\label{prop:l1_restRAPPOR:approx}
For any $0 < \epsilon < 2\ln(\frac{|\calXN|}{2}+1)$ and $|\calXS|\le|\calXN|$, 
(\ref{eq:restRAPPOR-l1-loss:formal}) is maximized when $\bmp =~\bmpUN$:
\begin{align}
\bbE \left[ l_1(\hbmp,\bmp) \right]
&\lesssim \bbE \left[ l_1(\hbmp,\bmpUN) \right] \nonumber \\
&=\!{\textstyle \sqrt{\frac{2}{n\pi}} \biggl( \frac{e^{\epsilon/4}|\calXS|}{e^{\epsilon/2}-1} + \sqrt{ \frac{e^{\epsilon/2}|\calXN|}{e^{\epsilon/2}-1} - 1 } \biggr)}, %\nonumber
\label{eq:restRAPPOR-l1_upper_bound}
\end{align}
where $f(n)\lesssim g(n)$ represents $\lim_{n\rightarrow\infty} f(n)/g(n) \le 1$.
\end{restatable}

Note that this proposition covers a wide range of $\epsilon$. 
For example, when $|\calXS|\le|\calXN|$, 
it covers 
both the high privacy regime ($\epsilon \approx 0$) and low privacy regime ($\epsilon = \ln |\calX|$), 
since $\ln |\calX| < 2\ln(\frac{|\calXN|}{2}+1)$. 
Below we instantiate the $l_1$ loss in the high and low privacy regimes based on this proposition. 

\smallskip
\noindent{\textbf{\colorBB{uRAP} 
in the high privacy regime.}}~~If $\epsilon$ is close to $0$, we have 
$e^{\epsilon/2} - 1 \approx \epsilon/2$. 
Thus, the right-hand side of 
(\ref{eq:restRAPPOR-l1_upper_bound}) 
in Proposition~\ref{prop:l1_restRAPPOR:approx} 
can be simplified as follows: 
\begin{align}
\bbE \left[ l_1(\hbmp,\bmpUN) \right] 
\approx {\textstyle \sqrt{\frac{2}{n\pi}} \cdot\frac{2|\calXS|}{\epsilon}}.
\label{eq:restRAP-better-than-RAP_main}
\end{align}
It is shown in \cite{Kairouz_ICML16} that 
the expected $l_1$ loss of the $\epsilon$-RAPPOR is at most 
$\sqrt{\frac{2}{n\pi}} \cdot\frac{2|\calX|}{\epsilon}$ 
when $\epsilon \approx 0$. 
Thus, by (\ref{eq:restRAP-better-than-RAP_main}), the expected $l_1$ loss of the 
$(\calXS,\epsilon)$-\colorBB{uRAP} 
is much smaller than that of the $\epsilon$-RAPPOR when $|\calXS| \ll |\calX|$. 

Moreover, 
by (\ref{eq:restRAP-better-than-RAP_main}), 
the expected $l_1$ loss of the $(\calXS,\epsilon)$-\colorBB{uRAP} 
in the worst case is expressed as 
$\Theta(\frac{|\calXS|}{\sqrt{n\epsilon^2}})$ 
in the high privacy regime. 
\colorB{As described in Section~\ref{sub:ULDP_theoretical}, 
this is ``order'' optimal among all 
\colorBB{$(\calXS,\calYP,\epsilon)$-\ULDP{}} 
mechanisms 
(in \arxiv{Appendix~\ref{sec:proofs_utility_l2loss},}\conference{Appendix~\ref{sub:proofs_utility_l2loss},} 
we also show that the expected $l_2$ 
of the 
$(\calXS,\epsilon)$-\colorBB{uRAP} 
is expressed as 
$\Theta(\frac{|\calXS|}{n\epsilon^2})$).} 

\smallskip
\noindent{\textbf{\colorBB{uRAP} 
in the low privacy regime.}}~~If $\epsilon=\ln|\calX|$ 
and 
$|\calXS| \ll |\calX|^{\frac{3}{4}}$, 
the right-hand side of (\ref{eq:restRAPPOR-l1_upper_bound}) 
can be simplified, using 
$|\calXS| / |\calX|^{\frac{3}{4}} \approx 0$, 
as follows:
\begin{align*}
\bbE \left[ l_1(\hbmp,\bmpUN) \right]
\approx {\textstyle \sqrt{\frac{2(|\calX|-1)}{n\pi}}}. %\nonumber \\
\end{align*}
Thus, when $\epsilon=\ln|\calX|$ and 
$|\calXS| \ll |\calX|^{\frac{3}{4}}$,
the 
$(\calXS,\epsilon)$-\colorBB{uRAP} 
also 
achieves almost the same data 
utility as the non-private mechanism, 
whereas the expected $l_1$ loss of the $\epsilon$-RAPPOR is $\sqrt{|\calX|}$ times larger than that of the non-private mechanism \cite{Kairouz_ICML16}.

\smallskip
\noindent{\textbf{Summary.}}~~In summary, 
the \colorBB{uRR} 
and 
\colorBB{uRAP} 
provide much higher utility than 
the RR and RAPPOR 
when $|\calXS| \ll |\calX|$. 
Moreover, 
when $\epsilon = \ln |\calX|$ and $|\calXS| \ll |\calX|$ (resp.~$|\calXS| \ll |\calX|^{\frac{3}{4}}$), 
the 
\colorBB{uRR} 
(resp.~\colorBB{uRAP}) 
achieves 
almost the same 
utility as a non-private mechanism.

\section{Personalized ULDP Mechanisms}
\label{sec:per_rest}
We now consider the personalized-mechanism scenario (outlined in Section~\ref{sub:notations}), 
and 
propose a 
\emph{PUM (Personalized ULDP Mechanism)} 
to keep secret what is sensitive for each user 
while enabling the data collector to estimate a distribution. 

Sections~\ref{sub:semantic_PUM} describes the PUM. 
Section~\ref{sub:privacy_PUM} explains its privacy properties. 
Section~\ref{sub:per_rest_dist_est} proposes 
\colorB{a method} to estimate the distribution $\bmp$ from $\bmY$ obfuscated using the PUM. 
Section~\ref{sub:per_rest_util_anal} analyzes the data utility of 
the PUM.

\subsection{PUM with $\kappa$ Semantic Tags}
\label{sub:semantic_PUM}

Figure~\ref{fig:basic_PUM} shows the overview of the 
PUM $\bmQ\uid{i}$ for the $i$-th user ($i = 1, 2, \ldots, n$). 
It first 
deterministically 
maps personal data $x \in \calX$ to \emph{intermediate data} 
using a \emph{pre-processor} $f_{pre}\uid{i}$, 
and then 
maps the intermediate data to obfuscated data $y \in \calY$ 
using a utility-optimized mechanism $\bmQ_{cmn}$ common to all users. 
The pre-processor $f_{pre}\uid{i}$ 
maps 
user-specific sensitive data 
$x \in \calXS\uid{i}$ to 
one of $\kappa$ bots: $\bot_1, \bot_2, \cdots,$ or $\bot_\kappa$. 
The $\kappa$ bots 
represent 
user-specific sensitive data, 
and 
each of them 
is associated with a \emph{semantic tag} such as ``home'' or ``workplace''. 
\colorBB{The $\kappa$ semantic tags are the same for all users, and are} 
useful when the data collector has some background knowledge about $\bmp$ 
conditioned on each tag. 
For example, a distribution of POIs tagged as ``home'' or ``workplace'' can be easily obtained via the Fousquare venue API \cite{Yang_TIST16}. 
Although this is not a user distribution but a ``POI distribution'', it can be used to roughly approximate the distribution of users tagged as ``home'' or ``workplace'', as shown in Section~\ref{sec:exp}. 
We define a set $\calZ$ of intermediate data by 
$\calZ = \calX \cup \{\bot_1, \cdots, \bot_\kappa\}$, 
and a set $\calZS$ of sensitive intermediate data by 
$\calZS = \calXS \cup \{\bot_1, \cdots, \bot_\kappa\}$. 
\begin{figure}
\centering
\includegraphics[width=0.76\linewidth]{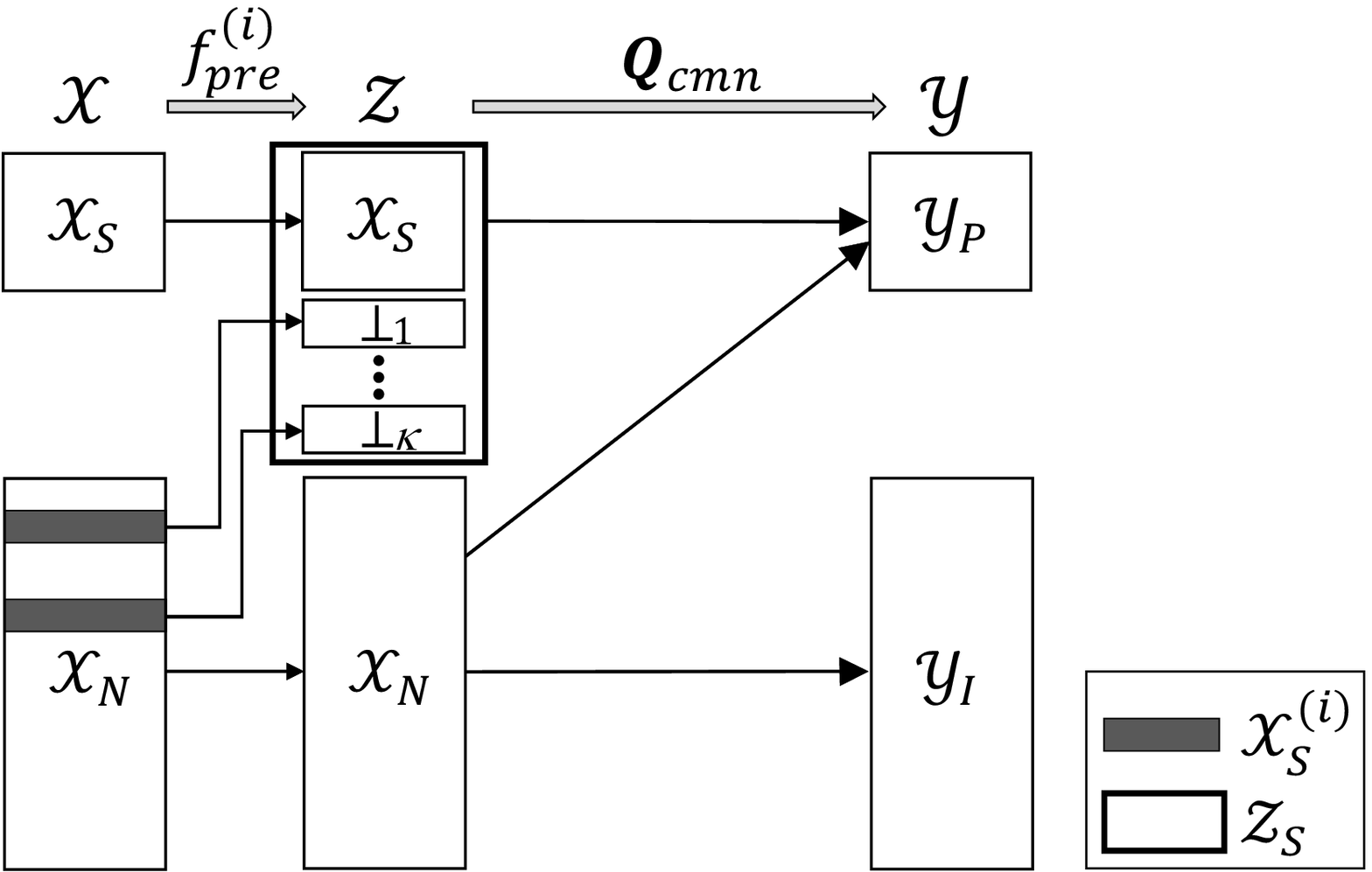}
\vspace{-2mm}
\caption{Overview of the PUM $\bmQ\uid{i}$ \colorBB{($= \bmQ_{cmn} \circ f_{pre}\uid{i}$)}. 
}
\label{fig:basic_PUM}
\end{figure}

Formally, 
the 
PUM $\bmQ\uid{i}$ first maps personal data $x \in \calX$ to intermediate data $z \in \calZ$ using a pre-processor $f_{pre}\uid{i}: \calX \rightarrow \calZ$ specific to each user. 
The pre-processor $f_{pre}\uid{i}$ maps sensitive data $x \in \calXS\uid{i}$ associated with the $k$-th tag ($1 \leq k \leq \kappa$) to the corresponding bot $\bot_k$, and maps other data to themselves. 
Let $\calX_{S,k}\uid{i}$ be a set of the $i$-th user's 
sensitive data associated with the $k$-th tag 
(e.g., set of regions including her primary home and second home). 
Then, $\calXS\uid{i}$ is expressed as $\calXS\uid{i} = \bigcup_{1 \leq k \leq \kappa} \calX_{S,k}\uid{i}$, and 
$f_{pre}\uid{i}$ is given by: 
\begin{align}
f_{pre}\uid{i}(x) = 
\begin{cases}
\bot_k & \text{(if $x \in \calX_{S,k}\uid{i}$)}\\
x & \text{(otherwise)}.\\
\end{cases} 
\label{eq:f_pre_semantic}
\end{align}

After mapping personal data $x \in \calX$ to intermediate data $z \in \calZ$, the $(\calZS,\calYP,\epsilon)$-utility-optimized mechanism $\bmQ_{cmn}$ maps $z$ to obfuscated data $y \in \calY$. 
Examples of $\bmQ_{cmn}$ include the $(\calZS,\epsilon)$-\colorBB{uRR} 
(in Definition~\ref{def:restRR}) and $(\calZS,\epsilon)$-\colorBB{uRAP} 
(in Definition~\ref{def:restRAPPOR}). 
As a whole, 
the 
PUM $\bmQ\uid{i}$ can be expressed as: $\bmQ\uid{i} = \bmQ_{cmn} \circ f_{pre}\uid{i}$. 
The $i$-th user stores $f_{pre}\uid{i}$ and $\bmQ_{cmn}$ in a device that obfuscates her personal data (e.g., mobile phone, personal computer). 
\colorBB{Note that if $f_{pre}\uid{i}$ is leaked, $x \in \calXN$ corresponding to each bot (e.g., home, workplace) is leaked. Thus, the user keeps $f_{pre}\uid{i}$ secret. To strongly prevent the leakage of $f_{pre}\uid{i}$, the user may deal with $f_{pre}\uid{i}$ using a tamper-resistant hardware/software. On the other hand,} 
the 
utility-optimized mechanism $\bmQ_{cmn}$, which is common to all users, is available to the data collector. 

The feature of the proposed PUM $\bmQ\uid{i}$ is two-fold: 
(i) the secrecy of the pre-processor $f_{pre}\uid{i}$ and (ii) the $\kappa$ semantic tags. 
By the first feature, the $i$-th user can keep $\calXS\uid{i}$ (i.e., what is sensitive for her) secret, as shown in Section~\ref{sub:privacy_PUM}. 
The second feature enables the data collector to estimate a distribution $\bmp$ with high accuracy. 
Specifically, 
she 
estimates 
$\bmp$ from obfuscated data 
$\bmY$ using $\bmQ_{cmn}$ and some background knowledge about $\bmp$ conditioned on each tag, 
as shown in Section~\ref{sub:per_rest_dist_est}. 

In practice, it may happen that a user has her specific sensitive data $x \in \calXS\uid{i}$ that is not associated with 
any semantic tags. 
For example, if 
we 
prepare only tags named ``home'' and ``workplace'', then 
sightseeing places, restaurants, 
and any other places are not associated with these tags. 
One way to deal with such data is to 
create another bot 
associated with a tag named ``others'' 
(e.g., if $\bot_1$ and $\bot_2$ are associated with
``home'' and ``workplace'', respectively, we create $\bot_3$ 
associated with ``others''), and 
map $x$ to this bot. 
It would be difficult for the data collector to obtain background knowledge about $\bmp$ conditioned on such a tag. 
In Section~\ref{sub:per_rest_dist_est}, 
we will explain how to estimate $\bmp$ in this case. 

\subsection{Privacy Properties}
\label{sub:privacy_PUM}

\colorB{We analyze the privacy properties of the PUM $\bmQ\uid{i}$. First, we show that it provides ULDP.} 

\begin{restatable}{prop}{propPUMULDP}
\label{prop:PUM_ULDP}
The PUM $\bmQ\uid{i}$ ($= \bmQ_{cmn} \circ f_{pre}\uid{i}$) provides $(\calXS \cup \calXS\uid{i},$ $\calYP,\epsilon)$-\ULDP{}.
\end{restatable}
\colorBB{We also show that 
our PUM 
provides DP in that 
an adversary who has observed $y \in \calYP$ cannot determine, for any $i,j \in [n]$, whether it is obfuscated using $\bmQ\uid{i}$ or $\bmQ\uid{j}$, which means that $y \in \calYP$ reveals almost no information about $\calXS\uid{i}$: 
\begin{restatable}{prop}{propPUMDPcalXS}
\label{prop:PUM_DPcalXS}
For any $i,j \in [n]$, any $x \in \calX$, and any $y \in \calYP$, 
\begin{align*}
\bmQ\uid{i} (y | x) \leq e^\epsilon \bmQ\uid{j} (y | x).
\end{align*}
\end{restatable}}

\colorB{We then analyze the secrecy of $\calXS\uid{i}$. The data collector, who knows the common-mechanism $\bmQ_{cmn}$, 
cannot obtain any information about $\calXS\uid{i}$ from $\bmQ_{cmn}$ \colorBB{and $y \in \calYP$}. 
Specifically, 
the data collector knows, for each $z \in \calZ$, whether $z \in \calZS$ or not by viewing $\bmQ_{cmn}$. 
However, she cannot obtain any information about $\calXS\uid{i}$ from $\calZS$, because 
she does not know the mapping between $\calXS\uid{i}$ and $\{\bot_1, \cdots, \bot_\kappa\}$ (i.e., $f_{pre}\uid{i}$). 
In addition, 
\colorBB{Propositions~\ref{prop:PUM_ULDP} and \ref{prop:PUM_DPcalXS} guarantee that $y \in \calYP$ reveals almost no information about both input data and $\calXS\uid{i}$.}
}

\colorB{For example, assume that the $i$-th user obfuscates her home 
$x \in \calXS \cup \calXS\uid{i}$ 
using the PUM $\bmQ\uid{i}$, and 
sends $y \in \calYP$ to the data collector. 
The data collector cannot infer 
either $x \in \calXS \cup \calXS\uid{i}$ or 
$z \in \calZS$ 
from 
$y \in \calYP$, since both $\bmQ_{cmn}$ and $\bmQ\uid{i}$ provide ULDP. 
This means that the data collector cannot infer \emph{the fact that she was at home} from $y$. 
Furthermore, the data collector cannot infer \emph{where her home is}, 
since $\calXS\uid{i}$ cannot be inferred from $\bmQ_{cmn}$ and $y \in \calYP$ as explained above.
}

\colorB{We need to take a little care when the $i$-th user obfuscates 
non-sensitive data $x \in \calXN \setminus \calXS\uid{i}$ 
using $\bmQ\uid{i}$ and sends $y \in \calYI$ to the data collector. 
In this case, the data collector learns 
$x$ from $y$, 
and therefore learns that $x$ is not sensitive (i.e., $x \notin \calXS\uid{i}$). 
Thus, the data collector, who knows that the user wants to hide her home, would reduce the number of possible candidates for her home from $\calX$ to  $\calX \setminus \{x\}$. 
However, if $|\calX|$ is large (e.g., $|\calX|$ = $625$ in our experiments using location data), the number $|\calX|-1$ of candidates is still large. 
Since the data collector cannot further reduce the number of candidates using $\bmQ_{cmn}$, her home is still kept 
strongly 
secret. 
In 
Section~\ref{sec:discussions_extension}, 
we also explain that the secrecy of $\calXS\uid{i}$ is achieved under reasonable assumptions even when she sends multiple 
data.
}

\subsection{Distribution Estimation}
\label{sub:per_rest_dist_est}
We now explain how to estimate a distribution $\bmp$ from data $\bmY$ obfuscated using the PUM. 
Let $\bmr^{(i)}$ be a distribution of intermediate data for the $i$-th user: 
\begin{align*}
\bmr^{(i)}(z) = 
\begin{cases}
\sum_{x \in \calX_{S,k}\uid{i}} \bmp(x) & \hspace{-2mm} \text{(if $z = \bot_{k}$ for some $k = 1, \ldots , \kappa$)}\\
0 & \hspace{-2mm} \text{(if $z \in \calXS\uid{i}$)}\\
\bmp(z) & \hspace{-2mm} \text{(otherwise)}.
\end{cases} 
\end{align*}
and $\bmr$ be the average of $\bmr^{(i)}$ over $n$ users; i.e., 
\arxiv{$\bmr(z) = \frac{1}{n} \allowbreak \sum_{i=1}^n \bmr^{(i)}(z)$}\conference{$\bmr(z) = \frac{1}{n} \sum_{i=1}^n \bmr^{(i)}(z)$} 
for any $z \in \calZ$. 
Note that 
$\sum_{x \in \calX} \bmp(x) = 1$ and 
$\sum_{z \in \calZ} \bmr(z) = 1$. 
Furthermore, 
let $\bmpi_k$ be a distribution of personal data $x \in \calX$ conditioned on $\bot_k$ defined by:
\begin{align}
\bmpi_k(x) &= \frac{\sum_{i=1}^n \bmp_k\uid{i}(x)}{\sum_{x' \in \calX} \sum_{i=1}^n \bmp_{k}\uid{i}(x')},
\label{eq:bmpi_kappa}
\\
\bmp_k\uid{i}(x) &=
\begin{cases}
\bmp(x) & \text{(if $f_{pre}\uid{i}(x) = \bot_k$)}\\
0 & \text{(otherwise)}.\nonumber
\end{cases} 
\end{align}
$\pi_k(x)$ in (\ref{eq:bmpi_kappa}) is a normalized sum of the probability $\bmp(x)$ of personal data $x$ whose corresponding intermediate data is $\bot_k$. 
Note that although $x \in \calX$ is deterministically mapped to $z \in \calZ$ for each user, we can consider the probability distribution $\bmpi_k$ 
for $n$ users. 
For example, if $\bot_k$ is tagged as ``home'', then $\bmpi_k$ is a distribution of users at home. 

\colorB{We propose a method to estimate a distribution $\bmp$ from obfuscated data $\bmY$ 
using some background knowledge about $\bmpi_k$ as an estimate $\hat{\bmpi}_k$ of $\bmpi_k$ 
(we explain the case where 
we have no background knowledge later). 
Our estimation method first estimates a distribution $\bmr$ of intermediate data 
from obfuscated data $\bmY$ using $\bmQ_{cmn}$. 
This can be performed in the same way as the common-mechanism scenario. 
Let $\hat{\bmr}$ be the estimate of $\bmr$. 
}

\colorB{After computing $\hbmr$, our method estimates $\bmp$ using 
the estimate $\hat{\bmpi}_k$ (i.e., background knowledge about $\bmpi_k$) 
as follows:}
\begin{align}
\hbmp(x) = \hat{\bmr}(x) + \sum_{k=1}^\kappa \hat{\bmr}(\bot_k) \hat{\bmpi}_k(x),~~\forall x \in \calX.
\label{eq:hbmp_PUM_2}
\end{align}
Note that $\hbmp$ in (\ref{eq:hbmp_PUM_2}) can be regarded as an empirical estimate of $\bmp$. 
Moreover, if both $\hbmr$ and $\hat{\bmpi}_k$ are in the probability simplex $\calC$, then $\hbmp$ in (\ref{eq:hbmp_PUM_2}) is always in $\calC$. 

\colorB{If 
we do 
not have estimates $\hat{\bmpi}_k$ 
for some bots (like the one tagged as ``others'' in Section~\ref{sub:semantic_PUM}), then 
we 
set $\hat{\bmpi}_k(x)$ 
in proportion to 
$\hat{\bmr}(x)$ 
over $x \in \calXN$ 
(i.e., 
$\hat{\bmpi}_k(x) = \frac{\hat{\bmr}(x)}{\sum_{x' \in \calXN} \hat{\bmr}(x')}$) 
for such bots. 
When 
we do 
not have any background knowledge $\hat{\bmpi}_1, \cdots, \hat{\bmpi}_\kappa$ for all bots, it amounts to simply discarding the estimates $\hat{\bmr}(\bot_1), \cdots, \hat{\bmr}(\bot_\kappa)$ for $\kappa$ bots and normalizing 
$\hat{\bmr}(x)$ 
over $x \in \calXN$ 
so that the sum is one.}

\subsection{Utility Analysis}
\label{sub:per_rest_util_anal}
We now theoretically analyze the data utility of our PUM. 
Recall that $\hbmp$, $\hbmr$, and $\hat{\bmpi}_k$ are the estimate of the distribution of personal data, intermediate data, and personal data conditioned on $\bot_k$, respectively. 
In the following, we show that the $l_1$ loss of $\hbmp$ can be upper-bounded as follows:

\begin{restatable}[$l_1$ loss of the PUM]{thm}{lonelossPUM}
\label{thm:l1_loss_decomposition}
\begin{align}
l_1(\hbmp, \bmp) \leq l_1(\hbmr,\bmr) + \sum_{k=1}^\kappa \hbmr(\bot_k) l_1(\hat{\bmpi}_k, \bmpi_k).
\label{eq:l1_loss_decomposition}
\end{align}
\end{restatable}

This means 
the upper-bound on the $l_1$ loss of $\hbmp$ can be decomposed into 
the $l_1$ loss of $\hbmr$ and of $\hat{\bmpi}_k$ weighted by $\hbmr(\bot_k)$. 

The first term in (\ref{eq:l1_loss_decomposition}) is the $l_1$ loss of $\hbmr$, which depends on 
$\bmQ_{cmn}$. 
For example, if we use 
the \colorBB{uRR} or \colorBB{uRAP} 
as $\bmQ_{cmn}$, the expectation of $l_1(\hbmr,\bmr)$ is given by Propositions~\ref{prop:l1_restRR} and \ref{prop:l1_restRAPPOR}, respectively. 
In Section~\ref{sec:exp}, we show they are very small. 

The second term in (\ref{eq:l1_loss_decomposition}) is 
the summation of the $l_1$ loss of $\hat{\bmpi}_k$ weighted by $\hbmr(\bot_k)$. 
If we accurately estimate $\bmpi_k$, 
the second term is very small. 
In other words, 
if we have enough background knowledge about $\bmpi_k$, 
we can accurately estimate $\bmp$ in the personalized-mechanism scenario. 

It should be noted that 
when the probability $\hbmr(\bot_k)$ is small, 
the second term in (\ref{eq:l1_loss_decomposition}) is small \emph{even if we have no background knowledge about $\bmpi_k$}. 
For example, when only a small number of users map $x \in \calXS\uid{i}$ to a tag named ``others'', 
they hardly affect the accuracy of $\hbmp$. 
Moreover, the second term in (\ref{eq:l1_loss_decomposition}) is upper-bounded by $2 \sum_{k=1}^\kappa \hbmr(\bot_k)$, since the $l_1$ loss is at most $2$. 
Thus, after computing $\hbmr$, 
the data collector can easily compute 
the worst-case value of the second term in (\ref{eq:l1_loss_decomposition}) to know the effect of the estimation error of $\hat{\bmpi}_k$ on the accuracy of $\hbmp$.

Last but not least, the second term in (\ref{eq:l1_loss_decomposition}) does not depend on $\epsilon$ (while the first term depends on $\epsilon$). 
Thus, the effect of the second term is relatively small when $\epsilon$ is small (i.e., high privacy regime), 
as shown in Section~\ref{sec:exp}. 

\smallskip
\noindent{\colorBB{\textbf{Remark.}}}~~\colorBB{Note that 
different privacy preferences 
might skew the distribution $\pi_k$. 
For example, doctors might not consider hospitals as sensitive as compared to patients. Consequently, the distribution $\pi_k$ conditioned on ``hospital'' might be a distribution of patients (not doctors) in hospitals. 
This kind of systematic bias can increase the estimation error of $\hat{\bmpi}_k$. 
Theorem~\ref{thm:l1_loss_decomposition} and the above discussions are also valid in this case.}

\section{Experimental Evaluation}
\label{sec:exp}

\subsection{Experimental Set-up}
\label{sub:set-up}
We conducted experiments using 
two large-scale datasets: 

\smallskip
\noindent{\textbf{Foursquare dataset.}}~~The Foursquare dataset (global-scale check-in dataset) \cite{Yang_TIST16} is one of the largest location datasets among publicly available datasets (e.g., see \cite{Gowalla}, \cite{CRAWDAD}, \cite{Yang_TSMC15}, \cite{Geolife}); it contains $33278683$ check-ins all over the world, each of which is associated with a POI ID and venue category (e.g., restaurant, shop, hotel, hospital, home, workplace). 

We used $359054$ 
\colorBB{check-ins} 
in Manhattan, assuming that each 
\colorBB{check-in} 
is from a different user. 
Then we divided Manhattan into $25 \times 25$ regions at regular intervals 
and used them as input alphabets; i.e., $|\calX| = 625$. 
The size of each region is about $400$m (horizontal) $\times$ $450$m (vertical). 
We assumed a region that includes a hospital visited by at least ten users as a sensitive region common to all users. 
The number of such regions was $|\calXS| = 15$. 
In addition, we assumed a region in $\calXN$ that includes a user's home or workplace 
as her user-specific sensitive region. 
The number of users at home and workplace was $5040$ and $19532$, respectively.

\smallskip
\noindent{\textbf{US Census dataset.}}~~The US Census (1990) dataset \cite{Lichman2013} was collected as part of the 1990 U.S. census. 
It contains responses from $2458285$ people (each person provides one response), each of which contains $68$ attributes. 

We used the responses from all people, 
and used age, income, marital status, and sex as attributes. 
Each attribute has $8$, $5$, $5$, and $2$ categories, respectively. 
(See \cite{Lichman2013} for details about the value of each category ID.)
We regarded a tuple of the category IDs as a total category ID, 
and used it as an input alphabet; i.e., $|\calX| = 400$ 
($= 8 \times 5 \times 5 \times 2$). 
We 
considered the fact that ``divorce'' and ``unemployment'' might be sensitive for many users \cite{Leahy_Psychology13}, and regarded such categories as sensitive for all users 
(to be on the safe side, as described in Section~\ref{sub:notations}). 
Note that people might be students until their twenties and might retire in their fifties or sixties. 
Children of age twelve and under cannot get married. 
We excluded such categories from sensitive ones. 
The number of sensitive categories was $|\calXS| = 76$. 

\smallskip
We used a frequency distribution of all people as a true distribution $\bmp$, and randomly chose a half of all people as users who provide their obfuscated data; i.e., $n=179527$ and $1229143$ in the Foursquare and US Census datasets, respectively. 
Here we did not use all people, because we would like 
to evaluate the non-private mechanism that does not obfuscate the personal data; i.e., 
the non-private mechanism has an estimation error in our experiments 
due to the random sampling from the population.

As 
utility, we evaluated the TV (Total Variation) by computing the sample mean over 
a hundred 
realizations of $\bmY$.

\subsection{Experimental Results}
\label{sub:res}

\noindent{\textbf{Common-mechanism scenario.}}~~We first focused on the common-mechanism scenario, and 
evaluated the RR, RAPPOR, 
\colorBB{uRR}, 
and 
\colorBB{uRAP}. 
As distribution estimation methods, 
we used empirical estimation, 
empirical estimation with the significance threshold, 
and EM reconstruction (denoted by ``emp'', ``emp+thr'', and ``EM'', respectively). 
In ``emp+thr'', we set the significance level $\alpha$ to be $\alpha = 0.05$, and uniformly assigned the remaining probability to each of the estimates below the significance threshold in the same way as \cite{Wang_USENIX17}.

\begin{figure}[t]
\centering
\includegraphics[width=0.97\linewidth]{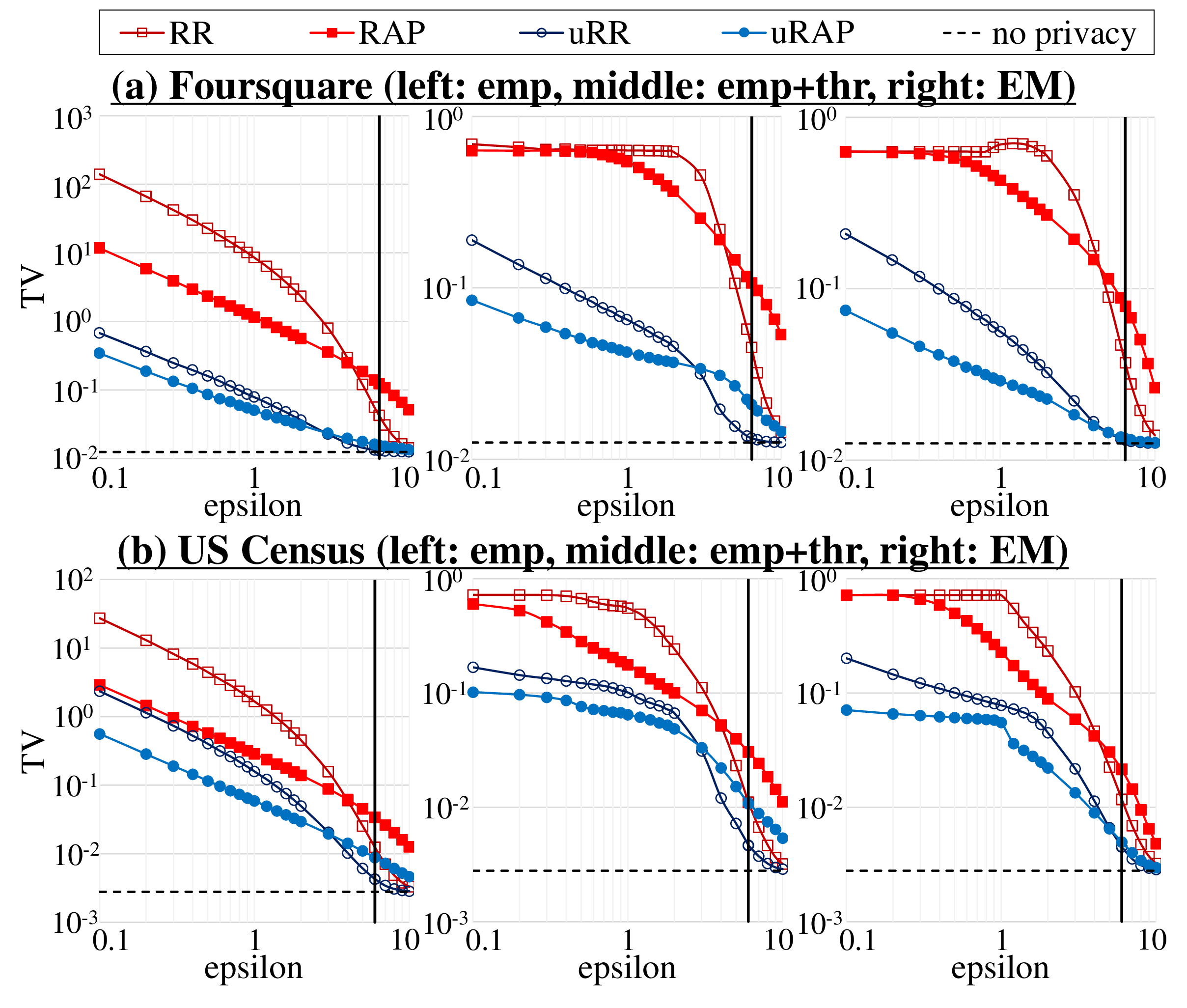}
\vspace{-5mm}
\caption{$\epsilon$ vs.~TV (common-mechanism). 
A bold 
line parallel to the $y$-axis represents 
$\epsilon = \ln |\calX|$.
}
\label{fig:res1_TV}
\end{figure}

Figure~\ref{fig:res1_TV} shows 
the results in the case where $\epsilon$ is changed from $0.1$ to $10$. 
\colorBB{``no privacy'' represents the non-private mechanism.} 
It can be seen that 
our mechanisms %(i.e., utility-optimized RR and RAPPOR) 
outperform the existing mechanisms 
by one or two orders of magnitude. 
Our mechanisms are effective especially 
in the Foursquare dataset, 
since the proportion of sensitive regions is very small ($15/625=0.024$). 
Moreover, 
the \colorBB{uRR} 
provides almost the same performance as the non-private mechanism when $\epsilon = \ln |\calX|$, as described in Section~\ref{sub:utility_analysis}. 
It can also be seen that 
``emp+thr'' and ``EM'' significantly outperform ``emp'', 
since the estimates in ``emp+thr'' and ``EM'' are always non-negative. 
Although ``EM'' outperforms ``emp+thr'' for the RAPPOR and 
\colorBB{uRAP} 
when $\epsilon$ was large, 
the two estimation methods provide very close performance as a whole.

We then evaluated the relationship between the number of sensitive regions/categories and the TV. 
To this end, 
we randomly chose $\calXS$ from $\calX$, and increased $|\calXS|$ from $1$ to $|\calX|$ 
(only in this experiment). 
We attempted 
one hundred
cases for randomly choosing $\calXS$ from $\calX$, and evaluated the TV by computing the sample mean over 
one hundred 
cases. 

\begin{figure}[t]
\centering
\includegraphics[width=0.97\linewidth]{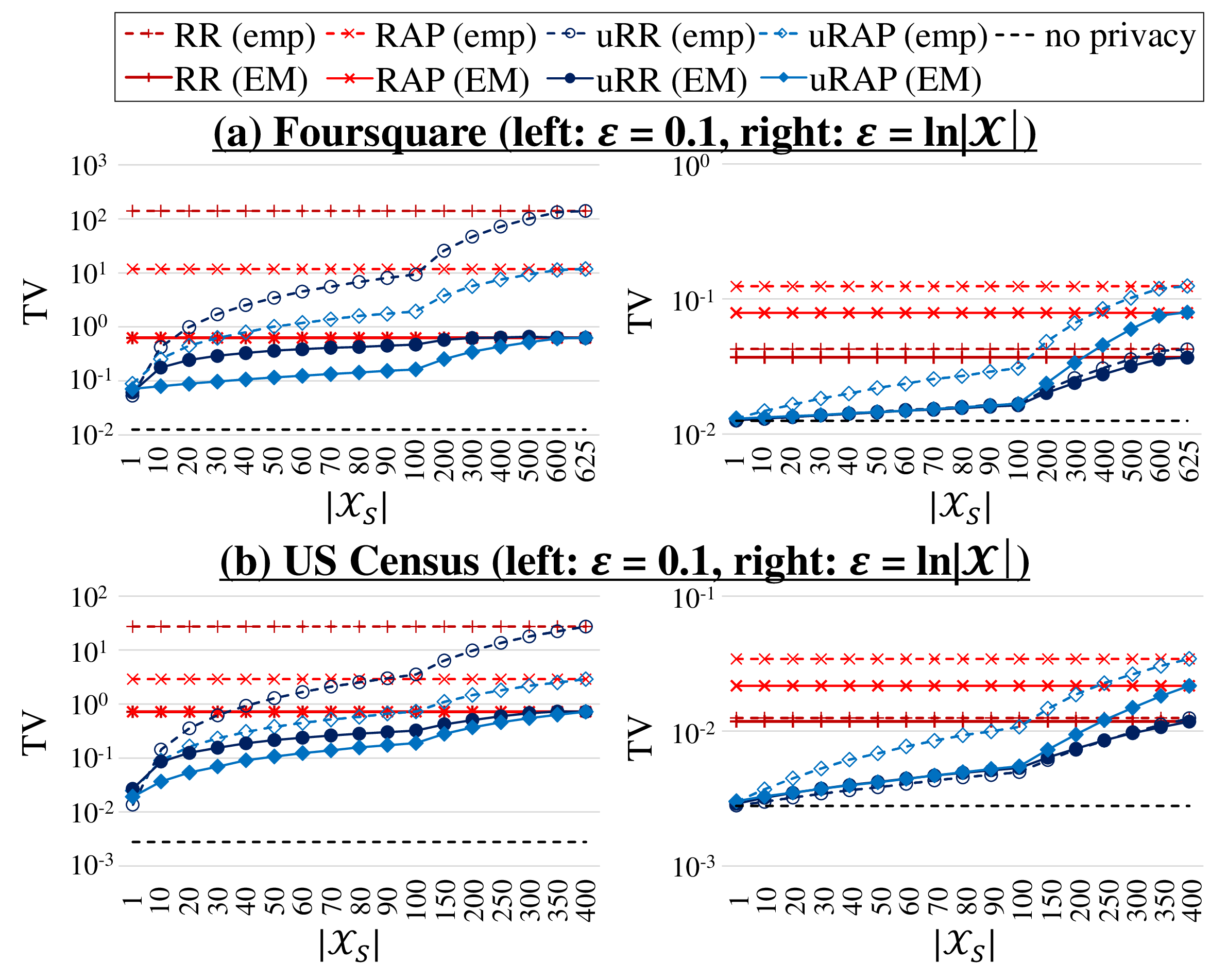}
\vspace{-5mm}
\caption{$|\calXS|$ vs.~TV when $\epsilon = 0.1$ or $\ln |\calX|$. %(left: Foursquare, right: US %\caption{\colorB{Relationship between $|\calXS|$ and the TV in the case where $\epsilon = 0.1$ or $\ln |\calX|$.} %(left: Foursquare, right: US Census, top: $\epsilon = 0.1$, bottom: $\epsilon = \ln |\calX|$). 
}
\label{fig:res2_TV}
\vspace{1.5mm}
\includegraphics[width=0.97\linewidth]{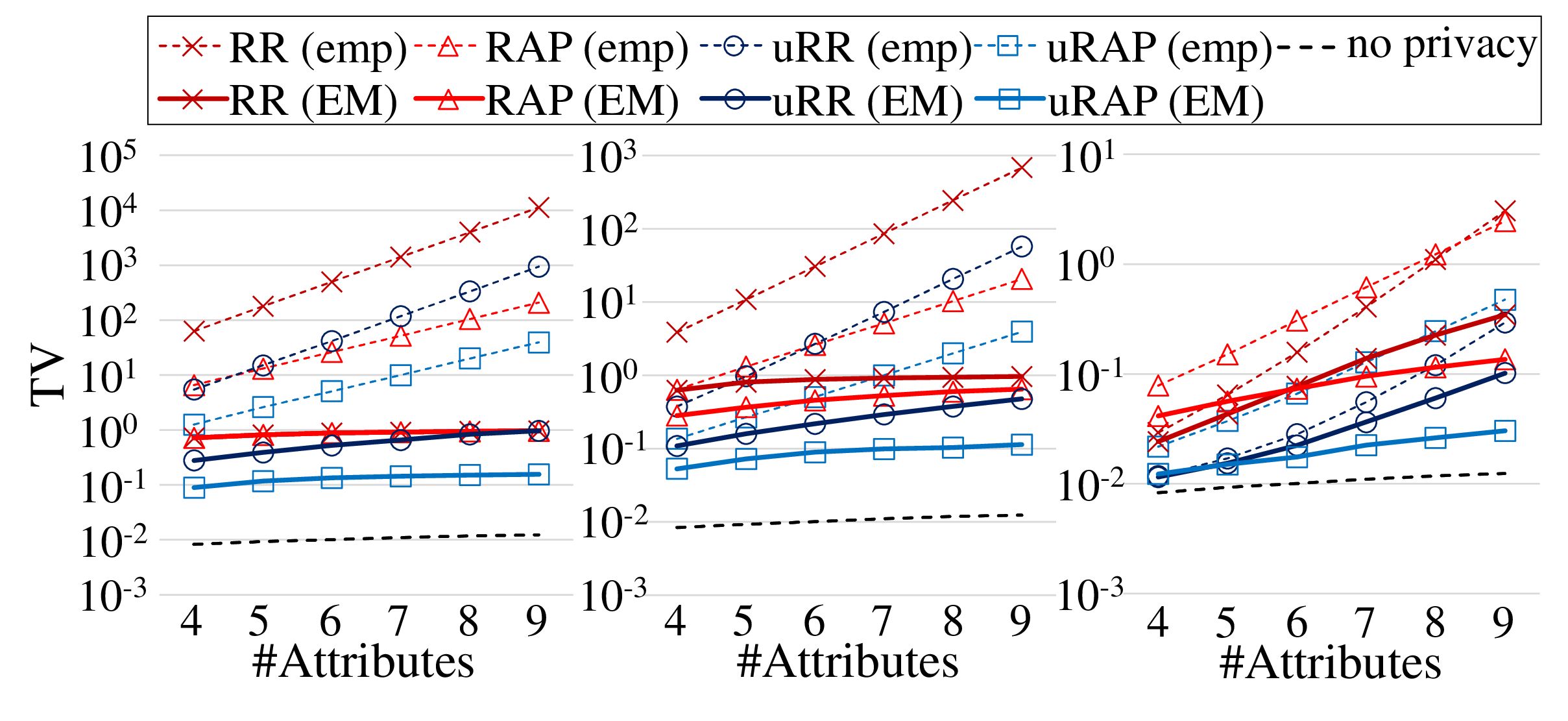}
\vspace{-5mm}
\caption{Number of attributes vs.~TV (US Census dataset; left: $\epsilon=0.1$, middle: $\epsilon=1.0$, right: $\epsilon=6.0$).
}
\label{fig:res4_TV}
\end{figure}

Figure~\ref{fig:res2_TV} shows the results 
\colorBB{for} 
$\epsilon = 0.1$ (high privacy regime) 
or $\ln |\calX|$ (low privacy regime). 
Here we omit the performance of ``emp+thr'', since it is very close to that of ``EM'' in the same way as in Figure~\ref{fig:res1_TV}.
The 
\colorBB{uRAP} 
and 
\colorBB{uRR} 
provide the best performance when $\epsilon = 0.1$ and $\ln |\calX|$, respectively. 
In addition, 
the \colorBB{uRR} 
provides the performance close to the non-private mechanism when $\epsilon = \ln |\calX|$ and 
the number $|\calXS|$ of sensitive regions/categories 
is less than $100$. 
The performance of 
the \colorBB{uRAP} 
is also close to that of the non-private mechanism when $|\calXS|$ is 
less than $20$ 
(note that $|\calX|^{\frac{3}{4}} = 125$ and $89$ in the Foursquare and US Census datasets, respectively). 
However, it rapidly increases with increase in $|\calXS|$. 
Overall, our theoretical results in Section~\ref{sub:utility_analysis} hold for the two real datasets. 

We also evaluated the performance when the number of attributes was increased from $4$ to $9$ in the US Census dataset. 
We added, one by one, five attributes as to whether or not a user has served in the military during five periods (``Sept80'', ``May75880'', ``Vietnam'', ``Feb55'', and ``Korean'' in \cite{Dua2017}; we added them in this order). 
We assumed that these attributes are non-sensitive. 
Since 
each of the five attributes had two categories (1: yes, 0: no), 
$|\calX|$ (resp.~$|\calXS|$) was changed from $400$ to $12800$ (resp.~from $76$ to $2432$). 
We randomly chose $n=240000$ people as users who provide obfuscated data, and 
evaluated the TV by computing the sample mean over ten realizations of $\bmY$ (only in this experiment).

Figure~\ref{fig:res4_TV} shows the results in the case where $\epsilon = 0.1$, $1.0$, or $6.0$ (=$\ln 400$). Here we omit the performance of ``emp+thr'' in the same way as Figure~\ref{fig:res2_TV}. 
Although the TV increases with an increase in the number of attributes, 
overall our utility-optimized mechanisms remain effective, compared to the existing mechanisms.
One exception is the case where $\epsilon=0.1$ and the number of attributes is $9$; 
the TV of the RR (EM), RAPPOR  (EM), and 
\colorBB{uRR} 
(EM) is almost $1$. 
Note that 
when we use the EM reconstruction method, 
the worst value of the TV is $1$. 
Thus, as with the RR and RAPPOR,
the \colorBB{uRR} 
fails to estimate a distribution in this case. 
On the other hand, the TV of 
the \colorBB{uRAP} 
(EM) is much smaller than $1$ even in this case, 
which is consistent with the fact that 
the \colorBB{uRAP} 
is order optimal in the high privacy regime. 
Overall, 
the \colorBB{uRAP} 
is robust to the increase of the attributes at the same value of $\epsilon$ 
(note that 
for large $|\calX|$, 
$\epsilon = 1.0$ or $6.0$ is a medium privacy regime where $0 \ll \epsilon \ll \ln |\calX|$).

We also measured 
the running time (i.e., time to estimate $\bmp$ from $\bmY$)
of 
``EM'' 
(which sets 
the estimate by ``emp+thr'' 
as an initial value of $\hbmp$) 
on an Intel Xeon CPU E5-2620 v3 (2.40 GHz, 6 cores, 12 logical processors) with 32 GB RAM. We found that the running time increases
roughly linearly with the number of attributes.
For example, 
when $\epsilon = 6.0$ and 
the number of attributes is $9$, 
the running time of ``EM'' 
required 
$3121$, $1258$, $5225$, and $1073$ 
seconds 
for 
``RR'', ``uRR'', ``RAP'', and ``uRAP'', respectively. 
We also measured the running time of 
`emp'' and ``emp+thr'', and found that 
they required less than one second even when the number of attributes is $9$.
Thus, if ``EM'' requires too much time for a large number of attributes, ``emp+thr'' would be a good alternative to ``EM''.

\begin{figure}[t]
\centering
\includegraphics[width=0.97\linewidth]{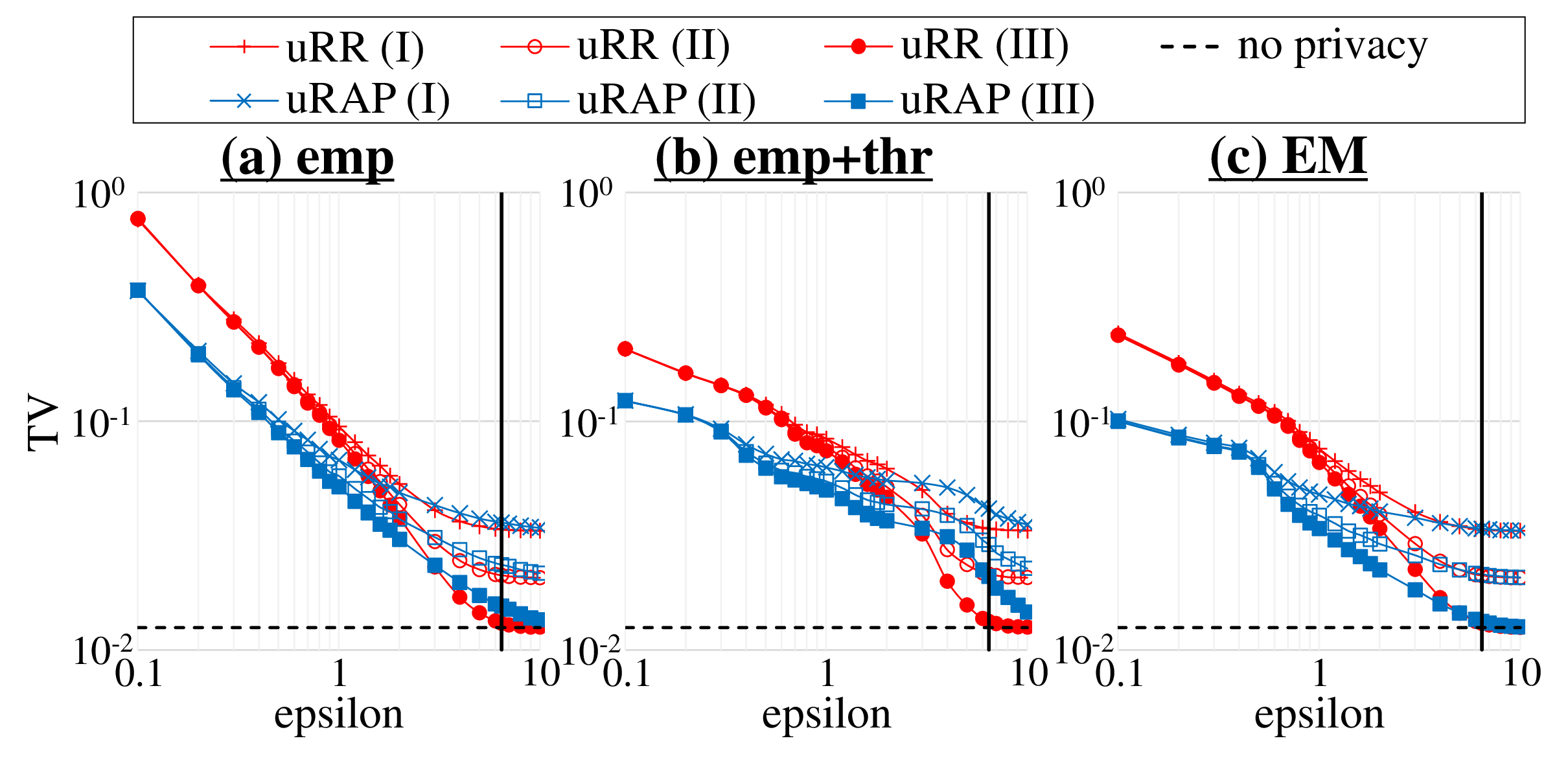}
\vspace{-5mm}
\caption{$\epsilon$ vs.~TV 
(personalized-mechanism) 
((I): w/o 
knowledge, (II): POI distribution, (III): true distribution).
}
\label{fig:res3_TV}
\includegraphics[width=1.0\linewidth]{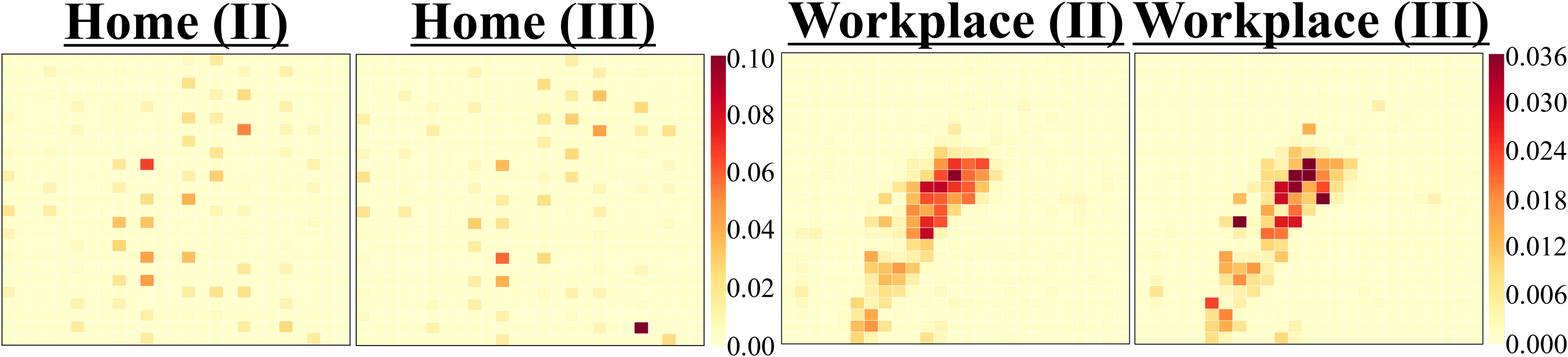}

\vspace{-6mm}
\caption{Visualization of the distributions 
((II): POI distribution, (III): true distribution).
}
\label{fig:res4_dist}
\end{figure}

\smallskip
\noindent{\textbf{Personalized-mechanism scenario.}}~~We then 
focused on 
the personalized-mechanism scenario, and 
evaluated 
our utility-optimized mechanisms 
using the Foursquare dataset. 
We used the PUM with $\kappa=2$ semantic tags 
(described in Section~\ref{sub:semantic_PUM}), which maps ``home'' and `workplace'' to bots $\bot_1$ and $\bot_2$, respectively. 
As the background knowledge about the bot distribution $\bmpi_k$ ($1 \leq k \leq 2)$, we considered three cases: 
(I) we do not have any background knowledge; 
(II) we use a distribution of POIs tagged as ``home'' (resp.~``workplace''), which is computed from the POI data in \cite{Yang_TIST16}, as an estimate of the bot probability $\hat{\bmpi}_1$ (resp.~$\hat{\bmpi}_2$);
(III) we use the true distributions (i.e., $\hat{\bmpi}_k = \bmpi_k$). 
Regarding (II), we emphasize again that it is not a user distribution but a ``POI distribution'', and can be easily obtained via the Foursquare venue API \cite{Yang_TIST16}.

Figure~\ref{fig:res3_TV} shows the results. 
We also show the POI and true distributions in Figure~\ref{fig:res4_dist}. 
It can be seen that 
the performance of (II) 
lies in 
between that of (I) and (III), 
which shows that the estimate $\hat{\bmpi}_k$ of the bot distribution affects utility. 
However, when $\epsilon$ is smaller than $1$, all of
(I), (II), and (III) 
provide almost the same performance, since the effect of the estimation error of $\hat{\bmpi}_k$ does not depend on $\epsilon$, 
as described in Section~\ref{sub:per_rest_util_anal}. 

\begin{table}[t]
\caption{$l_1$ loss $l_1(\hbmp, \bmp)$ and the first and second terms in the right-hand side of (\ref{eq:l1_loss_decomposition}) in the case where $\epsilon = \ln |\calX|$ and the EM reconstruction method is used.}
\centering
\hbox to\hsize{\hfil
\begin{tabular}{l|l|l|l}
\hline
Method	&	$l_1(\hbmp, \bmp)$ 	&	first term	&	second term\\
\hline
uRR (\textbf{I})	&	$6.73 \times 10^{-2}$	&	$2.70 \times 10^{-2}$	&	$7.34 \times 10^{-2}$\\
uRR (\textbf{II})	&	$4.24 \times 10^{-2}$	&	$2.70 \times 10^{-2}$	&	$2.96 \times 10^{-2}$\\
uRR (\textbf{III})	&	$2.62 \times 10^{-2}$	&	$2.70 \times 10^{-2}$	&	$0$\\
uRAP (\textbf{I})	&	$6.77 \times 10^{-2}$	&	$2.76 \times 10^{-2}$	&	$7.35 \times 10^{-2}$\\
uRAP (\textbf{II})	&	$4.28 \times 10^{-2}$	&	$2.76 \times 10^{-2}$	&	$2.96 \times 10^{-2}$\\
uRAP (\textbf{III})	&	$2.67 \times 10^{-2}$	&	$2.76 \times 10^{-2}$	&	$0$\\
\hline
\end{tabular}
\hfil}
\label{tab:first_second_terms}
\end{table}

We 
also 
computed the $l_1$ loss $l_1(\hbmp, \bmp)$ and 
the first and second terms in the right-hand side of (\ref{eq:l1_loss_decomposition}) to investigate whether Theorem~\ref{thm:l1_loss_decomposition} holds. 
Table~\ref{tab:first_second_terms} shows the results 
(we averaged the values over 
one hundred
realizations of $\bmY$). 
It can be seen that $l_1(\hbmp, \bmp)$ is smaller than the summation of the first and second terms in all of the methods, which shows that Theorem~\ref{thm:l1_loss_decomposition} holds in our experiments.

\smallskip
From these experimental results, we conclude that 
our proposed methods are very effective 
in both the common-mechanism and personalized-mechanism scenarios. 
In Appendix~\ref{sec:results_MSE}, 
we  
show the MSE has similar results to the TV.

\section{Discussions}
\label{sec:discussions_extension}
\noindent{\textbf{On the case of multiple 
data 
per user.}}~~We have so far assumed that each user sends only 
a single 
\colorBB{datum}. 
Now we discuss
the case where 
each user sends multiple 
data 
based on the 
compositionality of \ULDP{} described in 
Section~\ref{sub:ULDP_theoretical}. 
Specifically, 
when 
a user sends $t$ $(>1)$ 
data, 
we obtain $(\calXS, (\calYP)^t, \epsilon)$-\ULDP{} in total by obfuscating each 
data 
using the $(\calXS,\calYP, \epsilon/t)$-utility-optimized mechanism. 
Note, however, that 
the amount of noise added to each data increases with 
increase in $t$. 
Consequently, 
for $\epsilon \in [0,t]$, 
the lower bound on the $l_1$ (resp.~$l_2$) loss 
(described in Section~\ref{sub:ULDP_theoretical}) can be expressed as 
$\Theta(\frac{\sqrt{t}|\calXS|}{\sqrt{n \epsilon^2}})$ (resp.~$\Theta(\frac{t|\calXS|}{n \epsilon^2})$), which 
increases with increase in $t$. 
Thus, $t$ cannot be large for distribution estimation in practice. 
This is also common to all LDP mechanisms.

Next we discuss 
the secrecy of $\calXS\uid{i}$. %as follows.
Assume that the $i$-th user 
\colorBB{obfuscates $t$ data using different seeds, and} 
sends $t_P$ 
protected 
data 
in $\calYP$ and $t_I$ 
invertible 
data 
in $\calYI$, 
where $t = t_P + t_I > 1$ 
\colorBB{(she can also use the same seed for the same data to reduce $t_I$ as in \cite{Erlingsson_CCS14})}.
If all the $t_I$ 
data 
in $\calYI$ are different from each other, 
the data collector learns $t_I$ original 
data 
in $\calXN$.
However, 
$t_I$ ($\leq t$) 
cannot be large in practice, 
as explained above.
In addition, 
in many 
applications, a user's personal data is highly non-uniform and sparse. 
In locations data, for example, a user often visits only a small number of regions in the whole map $\calX$.
Let $\mathcal{T}\uid{i} \subseteq \calXN$ be a set of possible input values 
for the $i$-th user 
in $\calXN$. 
Then, 
even if $t_I$ is large, 
the data collector 
cannot learn more than $|\mathcal{T}\uid{i}|$ 
data 
in~$\calXN$. 

Moreover, 
the $t_P$ 
data 
in $\calYP$ reveal almost no information about $\calXS\uid{i}$, 
since $\bmQ\uid{i}$ provides $(\calXS, (\calYP)^t, \epsilon)$-\ULDP{}. 
$\bmQ_{cmn}$ provides no information about $\calXS\uid{i}$, since $f_{pre}\uid{i}$ is kept secret. 
Thus, 
the data collector, 
who knows that the user wants to hide her home, 
cannot 
reduce the number of candidates for her home from 
$\max\{|\calX| - t_I, |\calX| - |\mathcal{T}\uid{i}|\}$ using the $t_P$ 
data 
and $\bmQ_{cmn}$. 
If either $t_I$ or $|\mathcal{T}\uid{i}|$ is much smaller than $|\calX|$, 
her home is kept strongly secret. 

\colorBB{Note that $\bmp$ can be estimated 
even if $\calXS\uid{i}$ changes over time. $\calXS\uid{i}$ is also kept strongly secret if $t_I$ or $|\mathcal{T}\uid{i}|$ is small.}

\smallskip
\noindent{\textbf{On the correlation between $\calXS$ and $\calXN$.}}~~It should also be noted that 
there might be a correlation between sensitive data $\calXS$ and non-sensitive data $\calXN$. 
For example, 
if a user discloses a non-sensitive region close to a sensitive region including her home, the adversary might  infer approximate information about the original location (e.g., the fact that the user lives in Paris). 
However, we emphasize that if the size of each region is large, the adversary cannot infer the exact location such as the exact home address. 
Similar approaches can be seen in 
a state-of-the-art location privacy 
measure 
called \emph{geo-indistinguishability} \cite{Andres_CCS13,Bordenabe_CCS14,Oya_CCS17,Shokri_PoPETs15}. 
Andr\'{e}s \textit{et al.}~\cite{Andres_CCS13} considered privacy protection within a radius of $200$m from the original location, whereas the size of each region in our experiments was about $400$m $\times$ $450$m (as described in Section~\ref{sub:set-up}). 
We can protect the exact location by setting the size of each region to be large enough, or setting all regions close to a user's sensitive location to be sensitive. 

There might also be a correlation 
between two attributes (e.g., income and marital status) 
in the US Census dataset. 
However, 
we combined the four category IDs into a total category ID for each user as described in Section~\ref{sub:set-up}. 
Thus, 
there is only ``one'' category ID for each user. 
Assuming that each user's 
data 
is independent, 
there is no correlation between 
data. 
Therefore, we conclude that 
the sensitive data are 
strongly protected 
in both the Foursquare and US Census datasets 
in our experiments. 

It should be noted, however, that 
the number of total category IDs increases exponentially with the number of attributes. 
Thus, 
when 
there are many attributes 
as in Figure~\ref{fig:res4_TV}, 
the estimation accuracy might be increased by obfuscating each attribute independently 
(rather than obfuscating a total ID) 
while considering the correlation among attributes. 
We also need to consider a correlation among ``users'' 
for some types of personal data (e.g., flu status). 
For 
rigorously protecting 
such 
correlated data, 
we should 
incorporate 
\colorBB{Pufferfish privacy \cite{Kifer_TODS14,Song_SIGMOD17} into ULDP, as described in Section~\ref{sec:intro}.}

\section{Conclusion}
\label{sec:conc}
In this paper, we introduced the notion of ULDP that guarantees privacy equivalent to LDP for only sensitive data. 
We proposed 
ULDP mechanisms 
in both the 
common and personalized mechanism scenarios. 
We evaluated the utility of our mechanisms theoretically 
and 
demonstrated the effectiveness of our mechanisms through experiments.

\arxiv{\balance}

\bibliographystyle{abbrv}
\bibliography{sigproc_abrv}

\appendix

\section{Notations}
\label{sec:notations}

We show the basic notations used throughout this paper in Table~\ref{tab:notations}.

\section{Properties of ULDP}
\label{sec:basic-properties}

In this section we present basic properties of ULDP: 
adaptive sequential composition, 
post-processing, 
and  the compatibility with LDP. 
We also prove that 
the utility-optimized RR and 
the utility-optimized RAPPOR provide ULDP.

\subsection{Sequential Composition}
\label{sub:sequential-composition}

Below we prove that ULDP provides the compositionality.

\SequentialComposition*

\begin{proof}
Let $\calY_{0I} = \calY \setminus \calY_{0P}$ and $\calY_{1I} = \calY \setminus \calY_{1P}$.
Let $\bmQ$ be the sequential composition of $\bmQ_0$ and $\bmQ_1$; 
i.e., 
\begin{align*}
\bmQ((y_0,y_1) | x) 
= \bmQ_0(y_0 | x) \bmQ_1(y_1 | (y_0, x)).
\end{align*}

We first show that $\bmQ$ satisfies the first condition (\ref{eq:xs_ys_0}) in Definition~\ref{def:rest}.
Let $(\calY_0\times\calY_1)_I = \calY_0\times\calY_1 \setminus (\calY_0\times\calY_1)_P$, which can be expressed as follows:
\begin{align*}
 (\calY_0\times\calY_1)_I = \left\{ (y_0,y_1)\in\calY_0\times\calY_1
 \mid y_0\in\calY_{0I} \,\mbox{ or }\, y_1\in\calY_{1I}
 \right\}.
\end{align*}
If either $y_0\in\calY_{0I}$ or $y_1\in\calY_{1I}$, then it reveals the corresponding input $x \in \calXN$; 
i.e., 
if $y_0\in\calY_{0I}$, 
then 
there exists an $x \in \calXN$ such that 
$\bmQ_0(y_0|x) > 0$ and 
$\bmQ_0(y_0|x') = 0$ for any $x' \neq x$; 
if $y_0\in\calY_{1I}$, 
then 
there exists an $x \in \calXN$ such that 
$\bmQ_1(y_1|(y_0,x)) > 0$ and 
$\bmQ_1(y_1|(y_0,x')) = 0$ for any $x' \neq x$. 
Thus, for any $(y_0,y_1) \in (\calY_0\times\calY_1)_I$, there exists an $x \in \calXN$ such that 
$\bmQ((y_0,y_1)|x) > 0$ and 
$\bmQ((y_0,y_1)|x') = 0$ for any $x' \neq x$.

Next we show that $\bmQ$ satisfies the second condition (\ref{eq:epsilon_OSLDP_whole}).
Let $x,x'\in\calX$ and $(y_0,y_1)\in(\calY_0\times\calY_1)_P$.
Then $y_0\in\calY_{P0}$ and $y_1\in\calY_{P1}$.
Hence we obtain:
\begin{align*}
&\bmQ((y_0,y_1) | x) \nonumber\\
&= \bmQ_0(y_0 | x) \bmQ_1(y_1 | (y_0, x)) \\
&\le e^{\epsilon_0} \bmQ_0(y_0 | x') \bmQ_1(y_1 | (y_0, x))
~~~~~\text{(by $y_0\in\calY_{0P}$)} \\
&\le e^{\epsilon_0} \bmQ_0(y_0 | x') e^{\epsilon_1} \bmQ_1(y_1 | (y_0, x'))
~~\text{(by $y_1\in\calY_{1P}$)} \\
&= e^{\epsilon_0+\epsilon_1}\bmQ((y_0,y_1) | x')
.
\end{align*}
\end{proof}

\subsection{Post-processing}
\label{sub:post-process}

\colorBB{
We first define a class of post-processing randomized algorithms that 
preserve data types: 
\begin{definition}[Preservation of data types]\label{def:sense-preserving}\rm
Let $\calYP$ and $\calZP$ 
be 
sets of protected data, 
and 
$\calYI$ and $\calZI$ be 
sets of invertible data.
Given a randomized algorithm $\bmQ_1$ from $\calYP\cup\calYI$ to $\calZP\cup\calZI$, 
we say that $\bmQ_1$ 
\emph{preserves data types}
if it satisfies:
\begin{itemize}
\item for any $z\in\calZP$ and any $y\in\calYI$, $\bmQ_1(z | y) = 0$, and
\item for any $z\in\calZI$, there exists a $y\in\calYI$ such that 
$\bmQ_1(z | y) > 0$ and $\bmQ_1(z | y') = 0$ for any $y' \neq y$.
\end{itemize}
\end{definition}}

\colorBB{
Then we show that \ULDP{} is immune to the post-processing by this class of randomized algorithms.
\begin{restatable}[Post-processing]{prop}{propPostProcess}
\label{prop:PostProcess}
Let $\varepsilon \ge 0$.
Let $\calZP$ and $\calZI$ be sets of protected and invertible data respectively, and $\calZ = \calZP\cup\calZI$.
Let $\bmQ_1$ be a randomized algorithm from $\calY$ to $\calZ$ that 
preserves 
data types. 
If an obfuscation mechanism $\bmQ_0$ from $\calX$ to $\calY$ provides $(\calXS,\calYP,\varepsilon)$-\ULDP{} then the composite function $\bmQ_1\circ\bmQ_0$ provides $(\calXS,\calZP,\varepsilon)$-\ULDP{}.
\end{restatable}}

\begin{proof}
We first show that $\bmQ$ satisfies the first condition (\ref{eq:xs_ys_0}) in Definition~\ref{def:rest}.
Let 
$z\in\calZI$. 
Since $\bmQ_1$ preserves data types, there exists a $y\in\calYI$ such that 
$\bmQ_1(z | y) > 0$ and $\bmQ_1(z | y') = 0$ for any $y' \neq y$.
In addition, 
since $\bmQ_0$ provides $(\calXS,\calYP,\varepsilon)$-\ULDP{}, 
there exists an $x\in\calXN$ such that 
$\bmQ_0(y | x) > 0$ and $\bmQ_0(y | x') = 0$ for any $x' \neq x$. 
Hence we obtain:
\begin{align*}
&(\bmQ_1\circ\bmQ_0)(z|x) 
= \bmQ_0(y|x) \bmQ_1(z|y) > 0 
\end{align*}
and for any $x' \neq x$,
\begin{align*}
(\bmQ_1\circ\bmQ_0)(z|x') &=
\bmQ_0(y|x') \bmQ_1(z|y) +
\sum_{y'\neq y} \bmQ_0(y'|x') \bmQ_1(z|y')
\\ &= 0.
\end{align*}

Next we show that $\bmQ$ satisfies the second condition (\ref{eq:epsilon_OSLDP_whole}).
Let $x,x'\in\calX$ and $z\in\calZP$.
Since $\bmQ_1$ preserves 
data types, 
$\bmQ_1(z | y) = 0$ holds for all $y\in\calYI$.
Then we obtain:
\begin{align*}
&\,
(\bmQ_1\circ\bmQ_0)(z|x) \\
&=
\sum_{y\in\calYP}\! \bmQ_0(y|x) \bmQ_1(z|y) \\
&\le
\sum_{y\in\calYP} e^\varepsilon \bmQ_0(y|x') \bmQ_1(z|y)
~~~\text{(by $\bmQ_0$'s $(\calXS,\calYP,\varepsilon)$-\ULDP{})} \\
&= 
e^\varepsilon (\bmQ_1\circ\bmQ_0)(z|x')
.
\end{align*}
Therefore $\bmQ_1\circ\bmQ_0$ provides $(\calXS, \calZP, \varepsilon)$-\ULDP{}.
\end{proof}

\begin{table}[t]
\caption{Basic notations used in this paper.} 
\centering
\hbox to\hsize{\hfil
\begin{tabular}{l|l}
\hline
Symbol													&	Description\\
\hline
$n$														&		Number of users.\\
$\calX$													&		Set of personal data.\\
$\calY$													&		Set of obfuscated data.\\
$\calXS$												&		Set of sensitive data common to all users.\\
$\calXN$												&		Set of the remaining personal data ($= \calX \setminus \calXS$).\\
$\calXS\uid{i}$											&		Set of sensitive data specific to the $i$-th user.\\
$X\uid{i}$										& 		Personal data of the $i$-th user.\\
$Y\uid{i}$	 									& 		Obfuscated data of the $i$-th user.\\
$\bmX$ 						& 		Tuple of all personal data.\\ 
$\bmY$						& 		Tuple of all obfuscated data.\\ 
$\bmQ\uid{i}$												&		Obfuscation mechanism of the $i$-th user.\\
$\bmp$													&		Distribution of the personal data.\\
$\hbmp$													&		Estimate of $\bmp$.\\
$\calC$													&		Probability simplex.\\
\hline
\end{tabular}
\hfil}
\label{tab:notations}
\end{table}

\colorBB{
For example, ULDP is immune to data cleaning operations 
(e.g., transforming values, merging disparate values) \cite{Krishnan_SIGMOD16} 
as long as they are represented as $\bmQ_1$ explained above.}

\colorBB{Note that 
$\bmQ_1$ needs to preserve data types for utility (i.e., to make all $y \in \calYI$ invertible, 
as in Definition~\ref{def:rest}, 
after post-processing), 
and the DP guarantee for $y \in \calYP$ is preserved by any post-processing algorithm. Specifically, 
by (\ref{eq:epsilon_OSLDP_whole}), 
for any randomized post-processing algorithm $\bmQ_1^*$, any obfuscated data $z \in \calZ$ obtained from $y \in \calYP$ via $\bmQ_1^*$, and any $x, x' \in \calX$, we have: $\Pr(z|x) \leq e^\epsilon \Pr(z|x')$.
}

\subsection{\colorB{Compatibility with LDP}}
\label{sub:compatibility}
\colorBB{
Assume that data collectors A and B adopt 
a mechanism $\bmQ_A$ providing $(\calXS,\calYP,\epsilon_A)$-\ULDP{} 
and 
a mechanism $\bmQ_B$ providing $\epsilon_B$-\LDP{}, respectively. 
In this case, 
all protected 
data in the data collector A 
can be 
combined with 
all obfuscated data in the data collector B (i.e., data integration) 
to perform data analysis under LDP. 
More specifically, assume that 
Alice transforms 
her sensitive personal data 
in $\calXS$ 
into 
$y_A \in \calYP$ (resp.~$y_B \in \calY$) using $\bmQ_A$ (resp.~$\bmQ_B$), and sends $y_A$ (resp.~$y_B$) to the data collector A (resp.~B) 
to request two different services 
(e.g., location check-in for A and point-of-interest search for B). 
Then, 
the composition $(\bmQ_A, \bmQ_B)$ in parallel 
has 
the following property:
\begin{restatable}[Compatibility with LDP]{prop}{DataIntegration}
\label{prop:data_integration}
If $\bmQ_A$ and $\bmQ_B$ respectively provide $(\calXS,\calYP,\epsilon_A)$-\ULDP{} and $\epsilon_B$-\LDP{},
then 
for any $x, x' \in \calX$, $y_A \in \calYP$, and $y_B \in \calY$, we have:
\begin{align*}
(\bmQ_A, \bmQ_B)(y_A, y_B | x) \allowbreak \leq e^{\epsilon_A + \epsilon_B} (\bmQ_A, \bmQ_B)(y_A, y_B | x').
\end{align*}
\end{restatable}}

\begin{proof}
By (\ref{eq:LDP}) and (\ref{eq:epsilon_OSLDP_whole}), we have: 
\begin{align*}
(\bmQ_A, \bmQ_B)(y_A, y_B | x) 
&= \bmQ_A(y_A | x) \bmQ_B(y_B | x) \\
&\leq e^{\epsilon_A} \bmQ_A(y_A | x') e^{\epsilon_B} \bmQ_B(y_B | x')\\
&= e^{\epsilon_A + \epsilon_B} (\bmQ_A, \bmQ_B)(y_A, y_B | x').
\end{align*}
\end{proof}

\colorBB{
Proposition~\ref{prop:data_integration} %means 
implies 
that 
Alice's sensitive personal data 
in $\calXS$ 
is protected by $(\epsilon_A + \epsilon_B)$-\LDP{} 
after the data integration.}

\subsection{ULDP of the utility-optimized RR}
\label{sec:proofs_restRR}

Below we prove that the utility-optimized RR provides ULDP. 

\ULDPrestRR*

\begin{proof}
It follows from (\ref{eq:restRR_XS}) 
that 
(\ref{eq:xs_ys_0}) holds. 
Since $c_1 / c_2 = e^{\epsilon}$, 
the inequality 
(\ref{eq:epsilon_OSLDP_whole}) 
also 
holds 
\colorB{(note that 
$c_3$ is 
uniquely determined 
from $c_2$ so that the sum of probabilities from $x \in \calXN$ is $1$; i.e., 
$c_3 = 1 - |\calXS|c_2$)}.
\end{proof}

\subsection{\colorB{ULDP of the utility-optimized RAPPOR}}
\label{sec:proofs_restRAPPOR}
Below we prove that the utility-optimized RAPPOR provides ULDP. 

\ULDPrestRAP*

\begin{proof}
Let $i, i' \in \{1, 2, \ldots, |\calX|\}$.

By (\ref{eq:rRAP_YS}), 
if $y \in \calYI$, then only one of $y_{|\calXS|+1}, \cdots, y_{|\calX|}$ is $1$. 
In addition, 
it follows from (\ref{eq:restRAPPOR_calXN}) that 
for any $j \in \{|\calXS|+1, \cdots, |\calX|\}$, 
$y$ with $y_j = 1$ always comes from $x_j$. 
Therefore, the $(\calXS,\theta,\epsilon)$-utility-optimized RAPPOR satisfies~(\ref{eq:xs_ys_0}).

To show that the $(\calXS,\theta,\epsilon)$-utility-optimized RAPPOR satisfies (\ref{eq:epsilon_OSLDP_whole}), we first prove a few claims as follows.

Let $y \in \calYP$.
If $i = i'$ then $\bmQrestRAPPOR(y | x_i) = \bmQrestRAPPOR(y | x_{i'})$ obviously.
Thus, we assume $i \neq i'$ hereafter.

Then we obtain the following claim:
for any $j\neq i, i'$,
\begin{align}
\frac{\Pr(y_j | x_i)}{\Pr(y_j | x_{i'})} = 1
{.}
\label{eq:ratio_is_one}
\end{align}
This claim is proven as follows:
If $1 \leq j \leq |\calXS|$ and $y_j = 1$, then $\frac{\Pr(y_j | x_i)}{\Pr(y_j | x_{i'})} = \frac{d_1}{d_1} = 1$.
If $1 \leq j \leq |\calXS|$ and $y_j = 0$, then $\frac{\Pr(y_j | x_i)}{\Pr(y_j | x_{i'})} = \frac{1-d_1}{1-d_1} = 1$.
Otherwise, since $|\calXS|+1 \le j \le |\calX|$ 
and $y \in \calYP$, we have $y_j = 0$, hence
$\frac{\Pr(y_j | x_i)}{\Pr(y_j | x_{i'})} = \frac{1}{1} = 1$.

Now we show that the $(\calXS,\theta,\epsilon)$-utility-optimized RAPPOR satisfies (\ref{eq:epsilon_OSLDP_whole}) as follows.

If $x_i, x_{i'} \in \calXS$, 
then it follows from (\ref{eq:restRAPPOR_calXS}) that 
the $(\calXS,\theta,\epsilon)$-utility-optimized RAPPOR is equivalent to the ($\theta,\epsilon$)-RAPPOR in $\calXS$ and $\calYP$,
and thus satisfies (\ref{eq:epsilon_OSLDP_whole}).

Next we consider the case in which $x_i \in \calXS$ and $x_{i'} \in \calXN$ 
(i.e., $1 \leq i \leq |\calXS|$ and $|\calXS|+1 \le i' \le |\calX|$). 
By $x_{i'} \in \calXN$ and $y \in \calYP$, we have $y_{i'} = 0$.
If $y_i = 1$ then we have:
\begin{align}
&
\frac{\bmQrestRAPPOR(y | x_i)}{\bmQrestRAPPOR(y | x_{i'})}
\nonumber \\
&= \frac{\Pr(y_i | x_i)}{\Pr(y_i | x_{i'})} \cdot \frac{\Pr(y_{i'} | x_i)}{\Pr(y_{i'} | x_{i'})} \cdot\hspace{-2ex} 
\prod_{\substack{j \neq i \\ 1 \le j \le |\calXS|}}\hspace{-2ex} \frac{\Pr(y_j | x_i)}{\Pr(y_j | x_{i'})}\cdot\hspace{-4.5ex}
\prod_{\substack{j \neq i' \\~~~ |\calXS|+1 \le j \le |\calX|}\hspace{-1.5ex}}\hspace{-3.5ex} \frac{\Pr(y_j | x_i)}{\Pr(y_j | x_{i'})}
\nonumber \\
&= 
\frac{\theta}{d_1} \cdot \frac{1}{d_2} \cdot 1 \cdot 1
\hspace{20ex}\text{(by \eqref{eq:restRAPPOR_calXS}, \eqref{eq:restRAPPOR_calXN}, and \eqref{eq:ratio_is_one})}
\nonumber \\
&= 
 \theta \cdot \frac{(1 - \theta)e^\epsilon+\theta}{\theta} \cdot \frac{e^\epsilon}{(1 - \theta)e^\epsilon+\theta}
\nonumber\\
&= 
 e^\epsilon
{,}
\nonumber
\end{align}
hence \eqref{eq:epsilon_OSLDP_whole} is satisfied.
If $y_i = 0$ then we obtain:
\begin{align}
&
\frac{\bmQrestRAPPOR(y | x_i)}{\bmQrestRAPPOR(y | x_{i'})}
\nonumber \\
&= \frac{\Pr(y_i | x_i)}{\Pr(y_i | x_{i'})} \cdot \frac{\Pr(y_{i'} | x_i)}{\Pr(y_{i'} | x_{i'})} \cdot\hspace{-2ex} 
\prod_{\substack{j \neq i \\ 1 \le j \le |\calXS|}}\hspace{-2ex} \frac{\Pr(y_j | x_i)}{\Pr(y_j | x_{i'})}\cdot\hspace{-4.5ex}
\prod_{\substack{j \neq i' \\~~~ |\calXS|+1 \le j \le |\calX|}\hspace{-1.5ex}}\hspace{-3.5ex} \frac{\Pr(y_j | x_i)}{\Pr(y_j | x_{i'})}
\nonumber \\
&=
 \frac{1 - \theta}{1 - d_1} \cdot \frac{1}{d_2} \cdot 1 \cdot 1
\hspace{15.5ex}\text{(by \eqref{eq:restRAPPOR_calXS}, \eqref{eq:restRAPPOR_calXN}, and \eqref{eq:ratio_is_one})}
\nonumber \\
&=
 (1-\theta) \cdot \frac{(1 - \theta)e^\epsilon+\theta}{(1-\theta)e^\epsilon} \cdot \frac{e^\epsilon}{(1 - \theta)e^\epsilon+\theta}
\nonumber\\
&= 
 1,
\nonumber
\end{align}
which also imply \eqref{eq:epsilon_OSLDP_whole}.

Finally we consider the case where $x_i, x_{i'} \in \calXN$ 
(i.e., $|\calXS|+1 \le i, i' \le |\calX|$).
By $x_i, x_{i'} \in \calXN$ and $y \in \calYP$, we have $y_i = y_{i'} = 0$.
Then:
\begin{align}
&
\frac{\bmQrestRAPPOR(y | x_i)}{\bmQrestRAPPOR(y | x_{i'})}
\nonumber \\
&= \frac{\Pr(y_i | x_i)}{\Pr(y_i | x_{i'})} \cdot \frac{\Pr(y_{i'} | x_i)}{\Pr(y_{i'} | x_{i'})} \cdot\hspace{-1ex} 
\prod_{\substack{1 \le j \le |\calXS|}}\hspace{-0ex} \frac{\Pr(y_j | x_i)}{\Pr(y_j | x_{i'})}\cdot\hspace{-3.5ex}
\prod_{\substack{j \neq i,i' \\~~~ |\calXS|+1 \le j \le |\calX|}\hspace{-1.5ex}}\hspace{-2.5ex} \frac{\Pr(y_j | x_i)}{\Pr(y_j | x_{i'})}
\nonumber \\
&=
 \frac{d_2}{1} \cdot \frac{1}{d_2} \cdot 1 \cdot 1
\hspace{19.5ex}\text{(by \eqref{eq:restRAPPOR_calXS}, \eqref{eq:restRAPPOR_calXN}, and \eqref{eq:ratio_is_one})}
\nonumber \\
&= 
 1.
\nonumber
\end{align}

Therefore, 
the $(\calXS,\theta,\epsilon)$-utility-optimized RAPPOR provides $(\calXS,\allowbreak\calYP, \epsilon)$-ULDP.
\end{proof}

\section{Relationship between LDP, ULDP and OSLDP}
\label{sec:ULDP_OSLDP}

Our main contributions lie in the proposal of local obfuscation mechanisms (i.e., uRR, uRAP, PUM) and ULDP is introduced to characterize the main features of these mechanisms, i.e., LDP for sensitive data and high utility in distribution estimation. 
Nonetheless, it is worth making clearer the reasons for using ULDP as a privacy measure.
To this end, 
we also
introduce the notion of
OSLDP (One-sided LDP), a local model version of OSDP (One-sided DP) proposed in a preprint \cite{Doudalis_arXiv17}:
\colorB{\begin{definition} [$(\calXS,\epsilon)$-\OSLDP{}] \label{def:OSLDP} 
Given $\calXS \subseteq \calX$ and $\epsilon \in \nngreals$, 
an obfuscation mechanism $\bmQ$ from $\calX$ to $\calY$ provides 
$(\calXS,\epsilon)$-\OSLDP{} 
if for any 
$x \in \calXS$, 
any $x' \in \calX$ and any $y \in \calY$, we have 
\begin{align}
\bmQxy \leq e^\epsilon \bmQxdy.
\label{eq:OSLDP}
\end{align}
\end{definition}}
OSLDP 
is a special case of OSDP \cite{Doudalis_arXiv17} 
that takes as input personal data of a single user. 
Unlike ULDP, 
OSLDP allows the transition probability $\bmQxdy$ from non-sensitive data $x' \in \calXN$ to be very large for any $y \in \calY$, 
and hence does not provide $\epsilon$-LDP for $\calY$ (whereas ULDP provides $\epsilon$-LDP for 
$\calYP$). 
Thus, OSLDP can be regarded as a ``relaxation'' of ULDP. 
In fact, the following proposition holds:
\begin{restatable}{prop}{propODLP}
\label{prop:rest_OSLDP} 
If an obfuscation mechanism $\bmQ$ provides $(\calXS,\calYP,\epsilon)$-\ULDP{}, then it also provides $(\calXS,\epsilon)$-\OSLDP{}.
\end{restatable}
\begin{proof}
It is easy to check by 
(\ref{eq:xs_ys_0}) and (\ref{eq:epsilon_OSLDP_whole}) that 
$\bmQ$ provides 
(\ref{eq:OSLDP}) 
for any $x \in \calXS$, any $x' \in \calX$, and any $y \in \calY$.
\end{proof}
It should be noted that if an obfuscation mechanism provides $\epsilon$-LDP, then it obviously provides $(\calXS,\calYP,\epsilon)$-\ULDP{}, where $\calYP = \calY$. 
Therefore, 
$(\calXS,\calYP,\epsilon)$-\ULDP{} 
is a privacy measure that lies between $\epsilon$-LDP and $(\calXS,\epsilon)$-\OSLDP{}. 

The advantage of ULDP over LDP is that it provides much higher utility than LDP when $|\calXS| \ll |\calX|$. 
As described in Section~\ref{sub:ULDP_theoretical}, for $\epsilon \in [0,1]$, the lower bound on the $l_1$ and $l_2$ losses of any $\epsilon$-\LDP{} mechanism can be expressed as $\Theta(\frac{|\calX|}{\sqrt{n \epsilon^2}})$ and $\Theta(\frac{|\calX|}{n \epsilon^2})$, respectively. 
On the other hand, 
the lower bound on the $l_1$ and $l_2$ losses of any $(\calXS,\calYP,\epsilon)$-\ULDP{} mechanism can be expressed as $\Theta(\frac{|\calXS|}{\sqrt{n \epsilon^2}})$ and $\Theta(\frac{|\calXS|}{n \epsilon^2})$, respectively, both of which are achieved by the utility-optimized RAPPOR. 
In addition, 
the utility-optimized RR and the utility-optimized RAPPOR can even achieve almost the same utility as non-private mechanisms when $\epsilon = \ln |\calX|$, as described in Section~\ref{sub:utility_analysis}.

We use ULDP instead of OSLDP for 
the following 
two reasons. 
The first reason is that 
ULDP is compatible with LDP, and makes it possible to perform data integration and data analysis under LDP 
(Proposition~\ref{prop:data_integration}). 
OSLDP does not have this property in general, since 
it allows the transition probability $\bmQxdy$ from non-sensitive data $x' \in \calXN$ to be very large for any $y \in \calY$, as explained above.

The second reason, which is more important, is 
that 
\textit{the utility of OSLDP is not better than that of ULDP}. 
Intuitively, it can be explained as follows. 
First, although $\calYP$ is not explicitly defined in OSLDP, we can define $\calYP$ in OSLDP as the \textit{image of $\calXS$}, 
and $\calYI$ as 
$\calYI = \calY \setminus \calYP$,
analogously to ULDP. 
Then, OSLDP differs from ULDP in the following two points: 
(i) it allows the transition probability $\bmQxdy$ from $x' \in \calXN$ to $y \in \calYP$ to be very large (i.e., (\ref{eq:epsilon_OSLDP_whole}) may not satisfied); 
(ii) it allows $y \in \calYI$ to be non-invertible. 
(i.e., (\ref{eq:xs_ys_0}) may not satisfied). 
Regarding (i), 
it is important to note that 
the transition probability from $x' \in \calXN$ to $\calYI$ 
decreases with increase in 
the transition probability from $x'$ to $\calYP$. 
Thus, (i) and (ii) only allow us to mix non-sensitive data with sensitive data or other non-sensitive data, and 
reduce the amount of output data $y \in \calYI$ that can be inverted to $x \in \calXN$. 

Then, each OSLDP mechanism can be decomposed into 
a ULDP mechanism and a randomized post-processing that mixes non-sensitive data with sensitive data or other non-sensitive data. 
Note that this post-processing does not preserve data types (in Definition~\ref{def:sense-preserving}), 
and hence OSLDP 
does not have a compatibility with LDP as explained above. 
In addition, although the post-processing might improve privacy for non-sensitive data, we would like to protect sensitive data in this paper and ULDP is sufficient for this purpose; i.e., it guarantees $\epsilon$-LDP for sensitive data. 

Since the 
information is 
generally lost (never gained) 
by mixing 
data via the randomized post-processing, 
the utility of OSLDP is not better than that of ULDP 
(this holds for the information-theoretic utility such as mutual information and $f$-divergences \cite{Kairouz_JMLR16} because of the data processing inequality \cite{elements,Chen_JMLR16}; we also show this for the expected $l_1$ and $l_2$ losses at the end of Appendix~\ref{sec:ULDP_OSLDP}). 
Thus, 
it suffices to consider ULDP for our goal of 
designing obfuscation mechanisms that achieve high utility while providing LDP for sensitive data 
(as tdescribed in Section~\ref{sec:intro}).

We now formalize our claim as follows:

\begin{restatable}{prop}{ULDPOSLDPutility}
\label{prop:ULDP_OSLDP_utility} 
Let $\calM_O$ be the class of all mechanisms from $\calX$ to $\calY$ providing $(\calXS,\epsilon)$-\OSLDP{}. 
For any $\bmQ_O \in \calM_O$,
there exist two sets $\calZ$ and $\calZP$, 
a $(\calXS,\calZP,\epsilon)$-\ULDP{} mechanism $\bmQ_U$ from $\calX$ to $\calZ$, 
and 
a randomized algorithm $\bmQ_R$ from $\calZ$ to $\calY$ 
such that: 
\begin{align}
\bmQ_O = \bmQ_R \circ \bmQ_U.
\label{eq:Q_O_Q_R_Q_U}
\end{align}
\end{restatable}

\begin{proof}
Let $\bmQ_O \in \calM_O$ and 
$\calYP$ be the image of $\calXS$ in $\bmQ_O$.

If $\bmQ_O$ provides $(\calXS,\calYP,\epsilon)$-\ULDP{},
then 
(\ref{eq:Q_O_Q_R_Q_U}) holds, where $\bmQ_O = \bmQ_U$ and $\bmQ_R$ is the identity transform. 
In addition, if $\calXS = \calX$ (i.e., $\calXN = \emptyset$), then 
all of 
$(\calXS,\epsilon)$-\OSLDP{}, 
$(\calXS,\calYP,\epsilon)$-\ULDP{}, and 
$\epsilon$-\LDP{} are equivalent, and hence 
(\ref{eq:Q_O_Q_R_Q_U}) holds. 

Assume that 
$\bmQ_O$ does not provide $(\calXS,\calYP,\epsilon)$-\ULDP{},
and that $\calXN \neq \emptyset$. 
Below we construct a ULDP mechanism $\bmQ_U$ by modifying $\bmQ_O$ so that the conditions (\ref{eq:xs_ys_0}) and (\ref{eq:epsilon_OSLDP_whole}) are satisfied.
Let $\calW = \calYP \cup \calXN$.
First, from $\bmQ_O$, 
we construct a mechanism $\bmQ_O^\dagger$ from $\calX$ to 
$\calW$  
such that: 
\begin{align}
\bmQ_O^\dagger(w|x) = 
\begin{cases}
\bmQ_O(w|x) & \hspace{-2mm} \text{(if $w \in \calYP$)}\\
\sum_{y' \in \calYI} \bmQ_O(y'|x) & \hspace{-2mm} \text{(if $w \in \calXN$ and $w = x$)}\\
0 & \hspace{-2mm} \text{(if $w \in \calXN$ and $w \neq x$)}.
\end{cases}
\label{eq:Q_0_dagger}
\end{align}
For any $w\in\calYP$, we define 
$\bmQ_{max}^\dagger(w)$ by: 
\begin{align*}
\bmQ_{max}^\dagger(w) = \max_{x_0 \in \calXS} \bmQ_O^\dagger(w|x_0). 
\end{align*}
For any $x\in\calXN$, we define $\alpha(x)$ by:
\begin{align*}
\alpha(x) = 1 - \sum_{w \in \calYP} \min\{\bmQ_O^\dagger(w|x), \bmQ_{max}^\dagger(w)\}.
\end{align*}
Note that $\alpha(x) \geq 0$. 
Let $\calXN' = \{x \in \calXN \mid \alpha(x) > 0 \}$, 
$\calZ = \calYP \cup \calXN'$, and $\calZP = \calYP$.
Then, 
from $\bmQ_O^\dagger$, 
we construct a mechanism 
$\bmQ_U$ 
from $\calX$ to 
$\calZ$ ($= \calYP \cup \calXN'$)
such that: 
\begin{align}
&\bmQ_U(z|x) \nonumber\\
&= 
\begin{cases}
\min\{\bmQ_O^\dagger(z|x), \bmQ_{max}^\dagger(z)\} & \hspace{-2mm} \text{(if $z \in \calYP$)}\\
\alpha(x) & \hspace{-2mm} \text{(if $z \in \calXN'$ and $z = x$)}\\
0 & \hspace{-2mm} \text{(if $z \in \calXN'$ and $z \neq x$)}.
\end{cases} 
\label{eq:Q_0_ddagger}
\end{align}

Below we show that 
$\bmQ_U$ provides $(\calXS,\calZP,\epsilon)$-\ULDP{}, where $\calZ = \calYP \cup \calXN'$ and $\calZP = \calYP$. 
By (\ref{eq:Q_0_ddagger}), 
$\bmQ_U$ satisfies the first condition (\ref{eq:xs_ys_0}) in Definition~\ref{def:rest}. 
By (\ref{eq:OSLDP}), 
it satisfies the second condition (\ref{eq:epsilon_OSLDP_whole}) 
for any $x, x' \in \calXS$ and any $z \in \calZP$. In addition, 
by (\ref{eq:OSLDP}) and (\ref{eq:Q_0_ddagger}), 
for any $x\in\calXN$, any $x'\in\calXS$, and any $z \in \calZP$, 
we obtain:
\begin{align}
\bmQ_U(z|x) 
&\leq \bmQ_{max}^\dagger(z) & \text{(by (\ref{eq:Q_0_ddagger}))} \nonumber\\
&\leq e^\epsilon \bmQ_O^\dagger(z|x') & \text{(by (\ref{eq:OSLDP}))} \nonumber\\
&= e^\epsilon \bmQ_U(z|x'). & \text{(by (\ref{eq:Q_0_ddagger}))}
\label{eq:Q_y_x_upper_bound}
\end{align}
By (\ref{eq:OSLDP}) and (\ref{eq:Q_y_x_upper_bound}), $\bmQ_U$ satisfies the second condition (\ref{eq:epsilon_OSLDP_whole}) 
for any $x\in\calXN$, any $x'\in\calXS$, and any $z \in \calYP$.
Furthermore, by (\ref{eq:OSLDP}) and (\ref{eq:Q_0_ddagger}), for any $x\in\calXN$ and any $z \in \calYP$, 
we obtain:
\begin{align}
e^{-\epsilon} \bmQ_{max}^\dagger(z) \leq \bmQ_U(z|x) \leq \bmQ_{max}^\dagger(z).
\label{eq:Q_y_x_two_sided}
\end{align}
Thus, $\bmQ_U$ satisfies the second condition (\ref{eq:epsilon_OSLDP_whole}) 
for any $x,x'\in\calXN$ and any $z \in \calYP$.
Therefore, $\bmQ_U$ provides $(\calXS,\calZP,\epsilon)$-\ULDP{}.

Finally, we show that there 
exists a randomized algorithm $\bmQ_R$ 
from $\calZ$ to $\calY$ 
such that $\bmQ_O = \bmQ_R \circ \bmQ_U$.
Let $\calYI = \calY \setminus \calYP$. 
First, we define a randomized algorithm $\bmQ_{R_1}$ 
from $\calW$ ($= \calYP \cup \calXN$) to $\calY$ by:
\begin{align}
\bmQ_{R_1}(y|w) = 
\begin{cases}
1 & \hspace{-2mm} \text{(if $w \in \calYP$ and $y=w$)}\\
0 & \hspace{-2mm} \text{(if $w \in \calYP$ and $y \neq w$)}\\
0 & \hspace{-2mm} \text{(if $w \in \calXN$ and $y \in \calYP$)}.\\
\frac{\bmQ_O(y|w)}{\sum_{y' \in \calYI} \bmQ_O(y'|w)} & \hspace{-2mm} \text{(if $w \in \calXN$ and $y \in \calYI$)}.
\end{cases}
\label{eq:rand_alg_R_1}
\end{align}
Note that $\sum_{y \in \calY} \bmQ_{R_1}(y|w) = 1$ for any $w \in \calW$.
$\bmQ_{R_1}$ mixes non-sensitive data with other non-sensitive data. 
By (\ref{eq:Q_0_dagger}) and (\ref{eq:rand_alg_R_1}), we obtain:
\begin{align}
\bmQ_O = \bmQ_{R_1} \circ \bmQ_O^\dagger
\label{eq:bmQ_O_R_1_dagger}
\end{align}
(note that in (\ref{eq:Q_0_dagger}), if $w=x$, then $\bmQ_O(y'|x) = \bmQ_O(y'|w)$). 

Next, 
for any $z \in \calXN'$ and any $w \in \calYP$, we define $\beta(z,w)$ by: 
\begin{align*}
\beta(z,w) = \frac{\bmQ_O^\dagger(w|z) - \min\{\bmQ_O^\dagger(w|z), \bmQ_{max}^\dagger(w)\}}{\alpha(z)},
\end{align*}
where $\alpha(z) > 0$ since $z\in\calXN'$. 
We also define a randomized algorithm $\bmQ_{R_2}$ 
from $\calZ$ ($=\calYP \cup \calXN'$) to $\calW$ ($=\calYP \cup \calXN$) by:
\begin{align}
&\bmQ_{R_2}(w|z) \nonumber\\
&= 
\begin{cases}
1 & \hspace{-2mm} \text{(if $z \in \calYP$ and $w=z$)}\\
0 & \hspace{-2mm} \text{(if $z \in \calYP$ and $w \neq z$)}\\
\beta(z,w) & \hspace{-2mm} \text{(if $z \in \calXN'$ and $w \in \calYP$)}\\
1 - \sum_{w'\in\calYP} \beta(z,w') & \hspace{-2mm} \text{(if $z \in \calXN'$, $w \in \calXN$, and $w = z$)}\\
0 & \hspace{-2mm} \text{(if $z \in \calXN'$, $w \in \calXN$, and $w \neq z$)}.
\end{cases}
\label{eq:rand_alg_R_2}
\end{align}
Note that 
$\beta(z,w) \geq 0$, since 
$\bmQ_O^\dagger(w|z) \geq \min\{\bmQ_O^\dagger(w|z),$ $\bmQ_{max}^\dagger(w)\}$.
$\sum_{w'\in\calYP} \beta(z,w') \leq 1$, since
$\sum_{w'\in\calYP} \bmQ_O^\dagger(w'|z) - \sum_{w'\in\calYP} \min\{\bmQ_O^\dagger(w'|z), \bmQ_{max}^\dagger(w')\} \leq \alpha(z)$. 
Furthermore, 
$\sum_{w \in \calW} \bmQ_{R_2}(w|z) = 1$ for any $z \in \calZ$. 
$\bmQ_{R_2}$ mixes non-sensitive data with sensitive data. 
By (\ref{eq:Q_0_ddagger}) and (\ref{eq:rand_alg_R_2}), we obtain:
\begin{align}
\bmQ_O^\dagger = \bmQ_{R_2} \circ \bmQ_U
\label{eq:bmQ_dagger_R_2_U}
\end{align}
(note that in (\ref{eq:Q_0_ddagger}), if $z=x$, then $\alpha(x) = \alpha(z)$). 
Let $\bmQ_R = \bmQ_{R_1} \circ \bmQ_{R_2}$.
Then by 
(\ref{eq:bmQ_O_R_1_dagger}) and (\ref{eq:bmQ_dagger_R_2_U}), 
we obtain
$\bmQ_O = \bmQ_R \circ \bmQ_U$. 
\end{proof}

From Proposition~\ref{prop:ULDP_OSLDP_utility}, we show that 
the expected $l_1$ and $l_2$ losses 
of OSLDP are not better than those of ULDP as follows. 
For any OSLDP mechanism $\bmQ_O \in \calM_O$ and any estimation method 
$\lambda_O$ from data in $\calY$, 
we can construct a ULDP mechanism $\bmQ_O$ by (\ref{eq:Q_0_ddagger}) and an estimation method $\lambda_U$ 
that perturbs data in $\calZ$ via $\bmQ_R$ and then estimates a distribution from data in $\calY$ via $\lambda_O$. 
$\bmQ_U$ and $\lambda_U$ provide the same expected $l_1$ and $l_2$ losses as $\bmQ_O$ and $\lambda_O$, 
and 
there might also exist ULDP mechanisms and estimation methods 
from data in $\calZ$ 
that provide 
smaller expected $l_1$ and $l_2$ losses. 
Thus, the expected $l_1$ and $l_2$ losses of OSLDP are not better than those of ULDP.

\begin{table*}[t]
\caption{$l_1$ loss of each obfuscation mechanism in the worst case (RR: randomized response, RAP: RAPPOR, uRR: utility-optimized RR, uRAP: utility-optimized RAPPOR, no privacy: non-private mechanism, *1: approximation in the case where $|\calXS| \ll |\calX|$).} 
\centering
\renewcommand{\arraystretch}{1.2}
\hbox to\hsize{\hfil
\begin{tabular}{l|c|c}
\hline
Mechanism			&	$\epsilon \approx 0$					&	$\epsilon = \ln |\calX|$\\
\hline
RR					&	$\sqrt{\frac{2}{n\pi}} \frac{|\calX|\sqrt{ |\calX| - 1 }}{\epsilon}$	&	$\sqrt{\frac{8(|\calX| - 1)}{n\pi}}$  \\
\hline
RAP				&	$\sqrt{\frac{2}{n\pi}} \cdot\frac{2|\calX|}{\epsilon}$			& $\sqrt{\frac{2\sqrt{\calX}(|\calX|-1)}{n\pi}}$\\
\hline
uRR		&	$\sqrt{\frac{2}{n\pi}} \cdot \frac{|\calXS|\sqrt{ |\calXS| - 1 }}{\epsilon}$ (see Appendix~\ref{sub:l1_restRR:high_privacy})	&	$\sqrt{\frac{2(|\calX| - 1)}{n\pi}}$ $^{(*1)}$ (see Appendix~\ref{sub:l1_restRR:low_privacy})  \\
\hline
uRAP	&	$\sqrt{\frac{2}{n\pi}} \cdot\frac{2|\calXS|}{\epsilon}$ (see Appendix~\ref{sub:l1_restRAPPOR:high_privacy})	&	
$\sqrt{\frac{2(|\calX|-1)}{n\pi}} \biggl( 1 + \frac{|\calXS|}{|\calX|^\frac{3}{4}} \biggr)$ $^{(*1)}$ (see Appendix~\ref{sub:l1_restRAPPOR:low_privacy}) \\
\hline
no privacy	&	\multicolumn{2}{|c}{$\sqrt{\frac{2(|\calX| - 1)}{n\pi}}$} \\
\hline
\end{tabular}
\hfil}
\label{tab:l_1_loss}
\end{table*}

\section{L1 loss of the utility-optimized Mechanisms}
\label{sec:proofs_utility}

In this section we show the detailed analyses on the $l_1$ loss of the utility-optimized RR and the utility-optimized RAPPOR. 
Table~\ref{tab:l_1_loss} summarizes the $l_1$ loss of each obfuscation mechanism. 

\subsection{$l_1$ loss of the utility-optimized RR}
\label{sec:proof_prop_l1_restRR}

We first present the $l_1$ loss of the $(\calXS,\epsilon)$-utility-optimized RR. 
In the theoretical analysis of utility, we use the empirical estimation method described in Section~\ref{sub:distribution_estimation}.
Then it follows from (\ref{eq:restRR_XS}) 
that 
the distribution $\bmm$ of the obfuscated data can be written as follows: 

\begin{align}
\bmm(x) = 
\begin{cases}
\frac{e^\epsilon - 1}{|\calXS| + e^\epsilon - 1} \bmp(x) + \frac{1}{|\calXS| + e^\epsilon - 1} & \text{(if $x \in \calXS$)}\\
\frac{e^\epsilon - 1}{|\calXS| + e^\epsilon - 1} \bmp(x) & \text{(if $x \in \calXN$)}.\\
\end{cases} 
\label{eq:bmm_restRR}
\end{align}
The empirical estimate of $\bmp$ is given by: 
\begin{align}
\hbmp(x) = 
\begin{cases}
\frac{|\calXS| + e^\epsilon - 1}{e^\epsilon - 1} \hbmm(x) - \frac{1}{e^\epsilon - 1} & \text{(if $x \in \calXS$)}\\
\frac{|\calXS| + e^\epsilon - 1}{e^\epsilon - 1} \hbmm(x) & \text{(if $x \in \calXN$)}.\\
\end{cases} 
\label{eq:hbmp_restRR}
\end{align}
The following proposition is derived from (\ref{eq:bmm_restRR}) and (\ref{eq:hbmp_restRR}):

\propLoneRestRR*

\begin{proof}
Let $\bmt$ be a frequency distribution of the obfuscated data with sample size $n$; i.e., $\bmt(x) = \hbmm(x)n$. 
By $\epsilon>0$, we have $u>0$ and $v>0$.
By (\ref{eq:bmm_restRR}) and (\ref{eq:hbmp_restRR}),
the $l_1$ loss of $\hbmp$ can be written as follows:
\begin{align*}
\expectym \left[ l_1(\hbmp,\bmp) \right]
&= \bbE \left[ \sum_{x \in \calX} |\hbmp(x) - \bmp(x)| \right] \\
&= \bbE \left[ \sum_{x \in \calX} v \cdot |\hbmm(x) - \bmm(x)| \right] \\
&= \sum_{x \in \calX} v \cdot \bbE \left[\, |\hbmm(x) - \bmm(x)| \,\right] \\
&= \sum_{x \in \calX} v \cdot \bbE \left[\, \left|\frac{\bmt(x)}{n} - \bbE \left[ \frac{\bmt(x)}{n} \right] \right| \,\right] \\
&= \sum_{x \in \calX} \frac{v}{\sqrt{n}} \cdot \bbE \left[\, \left|\frac{\bmt(x) - \bbE \bmt(x)}{\sqrt{n}} \right| \,\right].
\end{align*}
It follows from the central limit theorem that 
$\frac{\bmt(x) - \bbE \bmt(x)}{\sqrt{n}}$ 
converges to the normal distribution $\normal{0, \bmm(x) (1 - \bmm(x))}$ as $n\rightarrow\infty$.
Here we use the fact that the absolute moment of a random variable $G\sim\normal{\mu, \sigma}$ is given by:
\begin{align*}
\bbE[\,|G|\,] = \sqrt{{\textstyle\frac{2}{\pi}}} \cdot\sigma \cdot\Phi\bigl({\textstyle-\frac{1}{2},\frac{1}{2};-\frac{\mu^2}{2\sigma^2}}\bigr)
\end{align*}
where $\Phi$ is Kummer's confluent hypergeometric function.
(See~\cite{Winkelbauer:12:arXiv} for details.)
Hence we obtain:
\begin{align}
\lim_{n\rightarrow\infty} \expectym \left[\, \left|\frac{\bmt(x) - \bbE \bmt(x)}{\sqrt{n}} \right| \,\right]
= \sqrt{\frac{2}{\pi} \bmm(x) (1 - \bmm(x))}.
\label{eq:restRR-l1-loss:cen_lim_the}
\end{align}
Then we have:
\begin{align}
\bbE \left[ l_1(\hbmp,\bmp) \right]
\approx \sqrt{\frac{2}{n\pi}} \cdot \sum_{x \in \calX} v \sqrt{\bmm(x) (1 - \bmm(x))}.
\label{eq:restRR-l1-loss:outdist}
\end{align}
Recall that $u = |\calXS| + e^\epsilon - 1$, $u'=e^\epsilon-1$, and $v = \frac{u}{u'}$.
It follows from (\ref{eq:bmm_restRR}) that for $x\in\calXS$, $\bmm(x) = \bmp(x)/v + 1/u$, and for $x\in\calXN$, $\bmm(x) = \bmp(x)/v$.
Therefore, we obtain:
\begin{align}
&\bbE \left[ l_1(\hbmp,\bmp) \right] \nonumber\\
&\approx \sqrt{\frac{2}{n\pi}} \biggl( \sum_{x \in \calXS} v \sqrt{\bigl( \bmp(x)/v + 1/u \bigr) \bigl( 1 - \bmp(x)/v - 1/u \bigr)} \nonumber\\
&\hspace{10ex} +\sum_{x \in \calXN} v \sqrt{\bmp(x)/v \bigl( 1 - \bmp(x)/v \bigr)} \biggr) \nonumber\\
&= \sqrt{\frac{2}{n\pi}} \biggl( \sum_{x \in \calXS} \sqrt{\bigl( \bmp(x) + 1/u' \bigr) \bigl( v - \bmp(x) - 1/u' \bigr)} \nonumber\\
&\hspace{10ex} +\sum_{x \in \calXN} \sqrt{\bmp(x) \bigl( v - \bmp(x) \bigr)} \biggr). \nonumber
\end{align}
\end{proof}

\subsubsection{Maximum of the $l_1$ loss}
Next we show that 
when $0 < \epsilon < \ln(|\calXN|+1)$, 
the $l_1$ loss 
is maximized by the uniform distribution 
$\bmpUN$ over $\calXN$.

\propLoneRestRRsmall*

Here we do not present the general result for $|\calXS| > |\calXN|$, because in this section later, we are interested in using this proposition to analyze the utility for $|\calX| \gg |\calXS|$, where the utility-optimized mechanism is useful.

To prove these propositions, we first show the lemma below.

\begin{lemma}
\label{lem:l1_restRR:decreasing}
Assume $1 \le |\calXS| \le |\calXN|$.
Let $u'=e^\epsilon-1$ and $v = \frac{|\calXS| + e^\epsilon-1}{e^\epsilon-1}$.
For $w\in[0,1]$, we define $A(w)$, $B(w)$, and $F(w)$ by:
\begin{align*}
A(w) &= \bigl( w + {\textstyle\frac{|\calXS|}{u'}} \bigr) \bigl( v|\calXS| - w - {\textstyle\frac{|\calXS|}{u'}} \bigr) \\
B(w) &= \bigl(1-w\bigr) \bigl( v|\calXN| - 1 + w \bigr) \\
F(w) &= \sqrt{ A(w) } + \sqrt{ B(w) }
{.}
\end{align*}
Then for any $0 < \epsilon < \ln(|\calXN|+1)$,
$F(w)$ is decreasing in $w$.
\end{lemma}

\begin{proof}
Let $0 < \epsilon < \ln(|\calXN|+1)$ and $w\in[0, 1)$.
By $\epsilon>0$, $|\calXS| \ge 1$, $0 \leq w \leq 1$, $v = \frac{|\calXS| + e^\epsilon-1}{e^\epsilon-1} \geq 1$, and $v - \frac{1}{u'} = \frac{|\calXS|-1}{e^\epsilon-1} + 1 \geq 1$, we have $A(w) > 0$ and $B(w) > 0$.
Then:
\begin{align*}
\frac{d \sqrt{A(w)}}{dw}
&= \frac{1}{2\sqrt{A(w)}} \frac{dA(w)}{dw} \\
&= \frac{1}{2\sqrt{A(w)}} \Bigl(
- 2 w + |\calXS| \bigl( v - {\textstyle\frac{2}{u'}} \bigr)
\Bigr) \\
\frac{d \sqrt{B(w)}}{dw}
&= \frac{1}{2\sqrt{B(w)}} \frac{dB(w)}{dw} \\
&= \frac{1}{2\sqrt{B(w)}} \Bigl(
- 2 w + \bigl( - v |\calXN| + 2 \bigr)
\Bigr).
\end{align*}
Let $C_w = \min\left( 2\sqrt{A(w)}, 2\sqrt{B(w)} \right)$.
By $C_w>0$, we obtain:
\begin{align*}
\frac{dF(w)}{dw}
&= \frac{d\sqrt{A(w)}}{dw} + \frac{d\sqrt{B(w)}}{dw} \\
&\le \frac{1}{C_w} \Bigl({\textstyle
- 2 w + |\calXS| \bigl( v - \frac{2}{u'} \bigr)
- 2 w + \bigl( - v|\calXN| + 2 \bigr) }
\Bigr) \\
&= \frac{1}{C_w} \Bigl(
{\textstyle - 4 w + |\calXS| \bigl( v - \frac{2}{u'} \bigr) - |\calXN| \bigl( v - \frac{2}{|\calXN|} \bigr)}
\Bigr) \nonumber \\
&\le \frac{1}{C_w} \Bigl(
{\textstyle - 4 w + |\calXN| \bigl( v - \frac{2}{u'} \bigr) - |\calXN| \bigl( v - \frac{2}{|\calXN|} \bigr)}
\Bigr) \nonumber \\
&\hspace{35ex}\text{(by $|\calXS| \le |\calXN|$)} \nonumber \\
&= \frac{1}{C_w} \Bigl(
{\textstyle - 4 w + 2|\calXN| \bigl( \frac{1}{|\calXN|} - \frac{1}{u'} \bigr)}
\Bigr). \nonumber
\end{align*}
Recall that $u'=e^\epsilon-1$.
By $e^\epsilon < |\calXN|+1$, we have:
\begin{align*}
\frac{1}{|\calXN|} - \frac{1}{u'}
= \frac{1}{|\calXN|} - \frac{1}{e^{\epsilon}-1}
< \frac{1}{|\calXN|} - \frac{1}{|\calXN|}
= 0.
\end{align*}
Hence $\frac{dF(w)}{dw} < 0$.
Therefore for $w\in[0,1)$, $F(w)$ is decreasing in $w$.
\end{proof}

Now we prove Proposition~\ref{prop:l1_restRR:eps=0} as follows.
\begin{proof}
Let $0 < \epsilon < \ln(|\calXN|+1)$. 
Let $\calC_{SN}$ be the set of all distributions $\bmp^*$ over $\calX$ that satisfy:
\begin{itemize}
\item for any $x\in\calXS$,\,
 $\bmp^*(x) = \frac{\bmp^*(\calXS)}{|\calXS|}$, and
\item for any $x\in\calXN$,\,
 $\bmp^*(x) = \frac{\bmp^*(\calXN)}{|\calXN|} = \frac{1 - \bmp^*(\calXS)}{|\calXN|}$.
\end{itemize}
Note that 
$\calC_{SN}$ is the set of mixture distributions of 
the uniform distribution $\bmpUS$ over $\calXS$ and 
the uniform distribution $\bmpUN$ over $\calXN$. 

By the Cauchy-Schwarz inequality, for any $a_1,a_2,\ldots,a_L\allowbreak \ge~0$, $\sum_{i=1}^{L} \sqrt{a_i} \le \sqrt{L \sum_{i=1}^{L} a_i}$, where the equality holds iff $a_1=a_2=\ldots=a_L$.
Hence by (\ref{eq:restRR-l1-loss:formal}) we obtain:
\begin{align}
&\bbE \left[ l_1(\hbmp,\bmp) \right] \nonumber\\
&\lesssim \sqrt{\frac{2}{n\pi}} \biggl( \sqrt{|\calXS| \sum_{x \in \calXS} \bigl( \bmp(x) + 1/u' \bigr) \bigl( v - \bmp(x) - 1/u' \bigr)} \nonumber\\
&\hspace{10ex} + \sqrt{|\calXN| \sum_{x \in \calXN} \bmp(x) \bigl( v - \bmp(x) \bigr)} \biggr), 
\nonumber
\end{align}
where $\approx$ holds iff $\bmp\in\calC_{SN}$.

Therefore we obtain:
\begin{align}
&~ \bbE \left[ l_1(\hbmp,\bmp) \right] \nonumber \\
&\lesssim \sqrt{\frac{2}{n\pi}} \!\biggl(\!\sqrt{|\calXS|^2 \Bigl( {\textstyle\frac{\bmp(\calXS)}{|\calXS|}} + 1/u' \Bigr) \Bigl( v - {\textstyle\frac{\bmp(\calXS)}{|\calXS|}} - 1/u' \Bigr)} \label{eq:l1-loss-RR-with-pxs} \\
&\hspace{9ex} + \sqrt{|\calXN|^2{\textstyle\frac{1-\bmp(\calXS)}{|\calXN|}} \Bigl( v - {\textstyle\frac{1-\bmp(\calXS)}{|\calXN|}} \Bigr)} \biggr) \nonumber\\
&= \sqrt{\frac{2}{n\pi}} \!\biggl(\!\sqrt{ \Bigl( \bmp(\calXS) + |\calXS|/u' \Bigr) \Bigl( v|\calXS| - \bmp(\calXS) - |\calXS|/u' \Bigr)} \nonumber\\
&\hspace{9ex} + \sqrt{\bigl(1-\bmp(\calXS)\bigr) \Bigl( v|\calXN| - 1 +\bmp(\calXS)} \biggr) \nonumber\\
&= \sqrt{\frac{2}{n\pi}}\, F(\bmp(\calXS)), %\nonumber
\label{eq:l1_RR_FpXs}
\end{align}
where $F$ is defined in Lemma~\ref{lem:l1_restRR:decreasing}.
Note that in (\ref{eq:l1-loss-RR-with-pxs}), $\approx$ holds iff $\bmp\in\calC_{SN}$.

By Lemma~\ref{lem:l1_restRR:decreasing} and
$0 < \epsilon < \ln(|\calXN|+1)$, 
$F(\bmp(\calXS))$ is maximized when $\bmp(\calXS)=0$.
Hence the right-hand side of (\ref{eq:restRR-l1-loss:formal}) is maximized when $\bmp(\calXS) = 0$ and $\bmp\in\calC_{SN}$, i.e., when $\bmp$ is the uniform distribution $\bmpUN$ over~$\calXN$.

Therefore we obtain:
\begin{align*}
&\bbE \left[ l_1(\hbmp,\bmp) \right] \nonumber\\
&\lesssim\!\bbE \left[ l_1(\hbmp,\bmpUN) \right] \\
&=\!\sqrt{\frac{2}{n\pi}}\,F(0) \\
&=\!\sqrt{\frac{2}{n\pi}} \!\biggl(\!\sqrt{ {\textstyle\frac{|\calXS|^2}{u'}} \bigl( v - {\textstyle\frac{1}{u'}} \bigr)} + \sqrt{ v|\calXN| - 1 } \biggr) \nonumber \\
&=\!\sqrt{\frac{2}{n\pi}} \!\Bigl(\! {\textstyle \frac{|\calXS|\!\sqrt{ |\calXS| + e^{\epsilon} - 2 }}{e^{\epsilon} - 1} +\!\sqrt{\! \frac{|\calXS||\calXN|}{e^{\epsilon} - 1} + |\calXN| - 1 } } \Bigr). \nonumber
\end{align*}
\end{proof}

Finally, we show that 
when $\epsilon \ge \ln(|\calXN|+1)$, 
the $l_1$ loss is maximized by a mixture of 
the uniform distribution $\bmpUS$ over $\calXS$ and 
the uniform distribution $\bmpUN$ over $\calXN$.

\propLoneRestRRbig*

\begin{proof}
We show that for any $\epsilon \ge \ln(|\calXN|+1)$, 
the right-hand side of (\ref{eq:restRR-l1-loss:formal})
is maximized when $\bmp = \bmp^*$.
(Note that by $\epsilon \ge \ln(|\calXN|+1)$, $\bmp^*(x)\ge 0$ holds for all $x\in\calX$.)

To show this, we prove that if $\bmp = \bmp^*$ then $\bmm$ is the uniform distribution over $\calY$ as follows.
If $x\in\calXS$ then we have:
\begin{align*}
\bmm(x)
&= \frac{e^\epsilon - 1}{|\calXS| + e^\epsilon - 1} \bmp^*(x) + \frac{1}{|\calXS| + e^\epsilon - 1} \\
&= \frac{e^\epsilon - 1}{|\calXS| + e^\epsilon - 1} \frac{1 - \frac{|\calXN|}{e^{\epsilon}-1}}{|\calXS| + |\calXN|} + \frac{1}{|\calXS| + e^\epsilon - 1} \\
&= \frac{e^\epsilon - 1 - |\calXN| + (|\calXS| + |\calXN|)}{(|\calXS| + e^\epsilon - 1)(|\calXS| + |\calXN|)} \\
&= \frac{1}{|\calX|}
.
\end{align*}
On the other hand, if $x\in\calXN$ then we obtain:
\begin{align*}
\bmm(x)
&= \frac{e^\epsilon - 1}{|\calXS| + e^\epsilon - 1} \bmp^*(x) \\
&= \frac{e^\epsilon - 1}{|\calXS| + e^\epsilon - 1} \frac{1 + \frac{|\calXS|}{e^{\epsilon}-1}}{|\calXS| + |\calXN|} \\
&= \frac{1}{|\calX|}
.
\end{align*}
Hence $\bmm$ is the uniform distribution over $\calY$.

By (\ref{eq:restRR-l1-loss:outdist}) and the Cauchy-Schwarz inequality, we obtain:
\begin{align*}
\bbE \left[ l_1(\hbmp,\bmp) \right]
&\lesssim \sqrt{\frac{2}{n\pi}} \cdot v \sqrt{|\calX| \sum_{x \in \calX} \bmm(x) (1 - \bmm(x))}
,
\end{align*}
where $\approx$ holds iff $\bmm$ is the uniform distribution over $\calY$, or equivalently $\bmp = \bmp^*$.
Hence:
\begin{align*}
\bbE \left[ l_1(\hbmp,\bmp) \right]
&\lesssim \bbE \left[ l_1(\hbmp,\bmp^*) \right] \\
&= \sqrt{\frac{2}{n\pi}} \cdot v \sqrt{|\calX| \sum_{x \in \calX} \frac{1}{|\calX|} \Bigl(1 - \frac{1}{|\calX|}\Bigr)} \\
&= \sqrt{\frac{2(|\calX| - 1)}{n\pi}} \cdot \frac{|\calXS|+e^{\epsilon}-1}{e^{\epsilon}-1}
.
\end{align*}
\end{proof}

\subsubsection{$l_1$ loss in the high privacy regime}
\label{sub:l1_restRR:high_privacy}
Consider the high privacy regime where $\epsilon\approx0$. 
In this case, $e^\epsilon - 1 \approx \epsilon$. 
By using this approximation, we simplify the $l_1$ loss of the utility-optimized RR for both the cases where $|\calXS|\le|\calXN|$ and where $|\calXS|>|\calXN|$.\\

\noindent{\textbf{Case 1: $|\calXS|\le|\calXN|$.}}~~By Proposition~\ref{prop:l1_restRR:eps=0},
the expected $l_1$ loss of the $(\calXS,\epsilon)$-utility-optimized RR mechanism is maximized by $\bmpUN$ and given by:
\begin{align*}
&\bbE \left[ l_1(\hbmp,\bmp) \right] \nonumber\\
&\lesssim\bbE \left[ l_1(\hbmp,\bmpUN) \right] \nonumber\\
&\approx\!{\textstyle \sqrt{\frac{2}{n\pi}}} \!\biggl(\! {\textstyle \frac{|\calXS|\sqrt{ |\calXS| +\epsilon - 1 }}{\epsilon} +\!\sqrt{ \frac{|\calXS|\cdot|\calXN|}{\epsilon} + |\calXN| - 1 } } \biggr) \nonumber \\
&\hspace{35ex} (\text{by $e^\epsilon-1\approx\epsilon$}) \nonumber \\
&\approx\!\sqrt{\frac{2}{n\pi}} \cdot \frac{|\calXS|\sqrt{ |\calXS| - 1 }}{\epsilon} \hspace{8ex}(\text{by $\epsilon\approx0$}). 
\end{align*}\\

\noindent{\textbf{Case 2: $|\calXS|>|\calXN|$.}}~~Let 
$F$ be the function defined in Lemma~\ref{lem:l1_restRR:decreasing}, 
$w^* = \argmax_{w\in[0,1]} F(w)$, 
and
$\bmp^{\!*}$ be the prior distribution over $\calX$ defined by:
\begin{align*}
\bmp^{\!*}(x) = 
\begin{cases}
\frac{w^*}{|\calXS|}
& (\text{if $x \in \calXS$}) \\
\frac{1 - w^*}{|\calXN|}
& (\text{otherwise}).
\end{cases}
\end{align*}
Then, by (\ref{eq:l1_RR_FpXs}), 
$\bbE \left[ l_1(\hbmp,\bmp) \right]$ is maximized by $\bmp^{\!*}$. Thus, 
for $\epsilon\approx0$, the expected $l_1$ loss of the $(\calXS,\epsilon)$-utility-optimized RR mechanism is given by:
\begin{align}
\bbE \left[ l_1(\hbmp,\bmp) \right] %\nonumber\\
&\lesssim \left[ l_1(\hbmp,\bmp^{\!*}) \right] \nonumber \\
&=\!\sqrt{\frac{2}{n\pi}}\,F(w^*) \nonumber \\
&\approx\!\sqrt{\frac{2}{n\pi}} \!\biggl(\! {\textstyle 
\sqrt{ \bigl( w^* + {\textstyle\frac{|\calXS|}{\epsilon}} \bigr) \bigl( \frac{|\calXS|^2}{\epsilon} - w^* - {\textstyle\frac{|\calXS|}{\epsilon}} \bigr) } } \nonumber \\
&\hspace{9ex} + {\textstyle \sqrt{ \bigl(1-w^*\bigr) \bigl( \frac{|\calXS||\calXN|}{\epsilon} - 1 + w^* \bigr) } }\biggr) \nonumber \\
&\hspace{13ex} (\text{by $u' = e^\epsilon-1\approx\epsilon$ and $v\approx {\textstyle\frac{|\calXS|}{\epsilon}}$}) \nonumber \\
&=\!\sqrt{\frac{2}{n\pi}} \!\biggl(\! {\textstyle 
\sqrt{ \frac{|\calXS|^2 (|\calXS| - 1)}{\epsilon^2} + O(\epsilon^{-1}) } } + {\textstyle \sqrt{ O(\epsilon^{-1}) } }\biggr) \nonumber \\
&\hspace{31ex} (\text{by $0\le w^* \le 1$}) \nonumber \\
&\approx\!\sqrt{\frac{2}{n\pi}} \cdot \frac{|\calXS| \sqrt{ |\calXS| - 1 } }{\epsilon}.
\hspace{11ex}(\text{by $\epsilon\approx0$})
\nonumber
\end{align}\\

In summary, the expected $l_1$ loss of the utility-optimized RR 
is at most 
$\sqrt{\frac{2}{n\pi}} \cdot \frac{|\calXS| \sqrt{ |\calXS| - 1 } }{\epsilon}$ 
in the high privacy regime, 
irrespective of whether $|\calXS|\le|\calXN|$ or not.
It is shown in \cite{Kairouz_ICML16} that 
the expected $l_1$ loss of the $\epsilon$-RR is at most 
$\sqrt{\frac{2}{n\pi}} \frac{|\calX|\sqrt{ |\calX| - 1 }}{\epsilon}$ when $\epsilon \approx 0$. 
Thus, 
the expected $l_1$ loss of the $(\calXS,\epsilon)$-utility-optimized RR is much smaller than that of the $\epsilon$-RR when $|\calXS| \ll |\calX|$.

\subsubsection{$l_1$ loss in the low privacy regime}
\label{sub:l1_restRR:low_privacy}

Consider the low privacy regime where $\epsilon=\ln|\calX|$ and $|\calXS| \ll |\calX|$. 
By Proposition~\ref{prop:l1_restRR:eps=lnX},
the expected $l_1$ loss of the $(\calXS,\epsilon)$-utility-optimized RR is given by:
\begin{align*}
\bbE \left[ l_1(\hbmp,\bmp) \right]
&\lesssim \bbE \left[ l_1(\hbmp,\bmp^*) \right] \\
&= \sqrt{\frac{2(|\calX| - 1)}{n\pi}} \cdot \frac{|\calXS|+|\calX|-1}{|\calX|-1} \\
&= \sqrt{\frac{2(|\calX| - 1)}{n\pi}} \Bigl( 1 + \frac{|\calXS|}{|\calX|-1} \Bigr) \\
&\approx \sqrt{\frac{2(|\calX| - 1)}{n\pi}}. \hspace{8ex}(\text{by $|\calXS| / |\calX| \approx 0$}) %\label{eq:l1_restRR_large_epsilon_np}
\end{align*}
It should be noted that the expected $l_1$ loss of the non-private mechanism, which does not obfuscate the personal data, is 
at most 
$\sqrt{\frac{2(|\calX| - 1)}{n\pi}}$ \cite{Kairouz_ICML16}. 
Thus, when $\epsilon=\ln|\calX|$ and $|\calXS| \ll |\calX|$,
the $(\calXS,\epsilon)$-utility-optimized RR achieves almost the same data utility as the non-private mechanism, 
whereas the expected $l_1$ loss of the $\epsilon$-RR is 
twice 
larger than that of the non-private mechanism \cite{Kairouz_ICML16}.

\subsection{$l_1$ loss of the utility-optimized RAPPOR}
\label{sub:proof_prop_l1_restRAPPOR}

We first present the $l_1$ loss of the $(\calXS,\epsilon)$-utility-optimized RAPPOR.
Recall that 
$\bmm_j$ (resp.~$\hbmm_j$) is the true probability (resp.~empirical probability) that the $j$-th coordinate in obfuscated data is $1$. 
It follows from (\ref{eq:restRAPPOR}), (\ref{eq:restRAPPOR_calXS}) and (\ref{eq:restRAPPOR_calXN}) that 
$\bmm_j$ can be written as follows: 
\begin{align}
\bmm_j = 
\begin{cases}
\frac{e^{\epsilon/2} - 1}{e^{\epsilon/2} + 1} \bmp(x_j) + \frac{1}{e^{\epsilon/2} + 1} & \text{(if $1 \leq j \leq |\calXS|$)}\\
\frac{e^{\epsilon/2} - 1}{e^{\epsilon/2}} \bmp(x_j) & \text{(otherwise)}.\\
\end{cases} 
\label{eq:bmmj_bmpj_restRAPPOR}
\end{align}
Then, the empirical estimate $\hbmp$ is given by: 
\begin{align}
\hbmp(x_j) = 
\begin{cases}
\frac{e^{\epsilon/2} + 1}{e^{\epsilon/2} - 1} \hbmm_j - \frac{1}{e^{\epsilon/2} - 1} & \text{(if $1 \leq j \leq |\calXS|$)}\\
\frac{e^{\epsilon/2}}{e^{\epsilon/2} - 1} \hbmm_j & \text{(otherwise)}.\\
\end{cases} 
\label{eq:bmpj_bmmj_restRAPPOR}
\end{align}
The following proposition is derived from 
(\ref{eq:bmmj_bmpj_restRAPPOR}) and (\ref{eq:bmpj_bmmj_restRAPPOR}): 

\propLoneRestRAP*

\begin{proof}
Let $v_S = \frac{e^{\epsilon/2}+1}{e^{\epsilon/2}-1}$.
By $\epsilon>0$, we have $v_S>0$ and $v_N>0$.

Analogously to the derivation of (\ref{eq:restRR-l1-loss:outdist}) in Appendix~\ref{sec:proof_prop_l1_restRR}, 
it follows from (\ref{eq:bmmj_bmpj_restRAPPOR}) and (\ref{eq:bmpj_bmmj_restRAPPOR}) that:
\begin{align*}
\bbE \left[ l_1(\hbmp,\bmp) \right]
\approx& \sqrt{\frac{2}{n\pi}} \sum_{j=1}^{|\calXS|} v_S \cdot \sqrt{\bmm_j (1 - \bmm_j)} \\
&+  \sqrt{\frac{2}{n\pi}} \sum_{j=|\calXS|+1}^{|\calX|} v_N \cdot \sqrt{\bmm_j (1 - \bmm_j)}.
\end{align*}
Let $u = e^{\epsilon/2} + 1$.
Then $v_S = \frac{u}{u'}$ and $v_N = \frac{u-1}{u'}$.
It follows from (\ref{eq:bmmj_bmpj_restRAPPOR}) that for $1\le j \le |\calXS|$, $\bmm_j = \bmp(x_j)/v_S + 1/u$, and for $|\calXS|+1 \le j \le |\calX|$, $\bmm_j = \bmp(x_j)/v_N$.
Therefore, we obtain:
\begin{align}
&\bbE \left[ l_1(\hbmp,\bmp) \right] \nonumber\\
&\approx \sqrt{\frac{2}{n\pi}} \biggl( \sum_{j=1}^{|\calXS|} v_S \sqrt{ \bigl( \bmp(x_j)/v_S + 1/u \bigr) \bigl( 1 - \bmp(x_j)/v_S - 1/u \bigr)} \nonumber\\
&\hspace{10ex} +\sum_{j=|\calXS|+1}^{|\calX|} v_N \sqrt{ \bmp(x_j)/v_N \bigl( 1 - \bmp(x_j)/v_N \bigr)} \biggr) \nonumber\\
&= \sqrt{\frac{2}{n\pi}} \biggl( \sum_{j=1}^{|\calXS|} \sqrt{\bigl( \bmp(x_j) + 1/u' \bigr) \bigl( v_S - \bmp(x_j) - 1/u' \bigr)} \nonumber\\
&\hspace{10ex} +\sum_{j=|\calXS|+1}^{|\calX|} \sqrt{\bmp(x_j) \bigl( v_N - \bmp(x_j) \bigr)} \biggr) \nonumber\\
&= \sqrt{\frac{2}{n\pi}} \biggl( \sum_{j=1}^{|\calXS|} \sqrt{\bigl( \bmp(x_j) + 1/u' \bigr) \bigl( v_N - \bmp(x_j) \bigr)} \nonumber\\
&\hspace{10ex} +\sum_{j=|\calXS|+1}^{|\calX|} \sqrt{\bmp(x_j) \bigl( v_N - \bmp(x_j) \bigr)} \biggr). \nonumber
\end{align}
\end{proof}

\subsubsection{Maximum of the $l_1$ loss}
Next we show that 
when $0 < \epsilon < 2\ln(\frac{|\calXN|}{2}+1)$, 
the $l_1$ loss is maximized by the uniform distribution $\bmpUN$ over $\calXN$.

\propLoneRestRAPapprox*

To prove this proposition, we first show the lemma below.

\begin{lemma}
\label{lem:l1_restRAPPOR:decreasing}
Assume $1 \le |\calXS| \le |\calXN|$.
Let $u'=e^{\epsilon/2}-1$ and 
$v_n = \frac{e^{\epsilon/2}}{e^{\epsilon/2}-1}$. 
For $w\in[0,1]$, we define $A(w)$, $B(w)$, and $F(w)$~by:
\begin{align*}
A(w) &= \bigl( w + {\textstyle\frac{|\calXS|}{u'}} \bigr) \bigl( v_N|\calXS| - w \bigr) \\
B(w) &= \bigl(1-w\bigr) \bigl( v_N|\calXN| - 1 + w \bigr) \\
F(w) &= \sqrt{ A(w) } + \sqrt{ B(w) }
{.}
\end{align*}
Then for any $0 < \epsilon < 2\ln (\frac{|\calXN|}{2}+1)$,
$F(w)$ is decreasing in~$w$.
\end{lemma}

\begin{proof}
Let $0 < \epsilon < 2\ln (\frac{|\calXN|}{2}+1)$ and $w\in[0, 1)$.
By $\epsilon>0$, $|\calXS| \ge 1$, 
and $v_n = \frac{e^{\epsilon/2}}{e^{\epsilon/2}-1} > 1$,
we have $A(w) > 0$ and $B(w) > 0$.
Then:
\begin{align*}
\frac{d \sqrt{A(w)}}{dw}
&= \frac{1}{2\sqrt{A(w)}} \frac{dA(w)}{dw} \\
&= \frac{1}{2\sqrt{A(w)}} \Bigl(
- 2 w + |\calXS|\bigl( v_N - {\textstyle\frac{1}{u'}} \bigr)
\Bigr) \\
\frac{d \sqrt{B(w)}}{dw}
&= \frac{1}{2\sqrt{B(w)}} \frac{dB(w)}{dw} \\
&= \frac{1}{2\sqrt{B(w)}} \Bigl(
- 2 w + \bigl( - v_N|\calXN| + 2 \bigr)
\Bigr).
\end{align*}
Let $C_w = \min\left( 2\sqrt{A(w)}, 2\sqrt{B(w)} \right)$.
By $C_w>0$, we obtain:
\begin{align*}
&\frac{dF(w)}{dw} \nonumber\\
&= \frac{d\sqrt{A(w)}}{dw} + \frac{d\sqrt{B(w)}}{dw} \\
&\le \frac{1}{C_w} \Bigl({\textstyle
- 2 w + |\calXS|\bigl( v_N - \frac{1}{u'} \bigr)
- 2 w - |\calXN|\bigl( v_N - \frac{2}{|\calXN|}} \bigr)
\Bigr) \\
&\le \frac{1}{C_w} \Bigl({\textstyle
- 4 w + |\calXN|\bigl( v_N - \frac{1}{u'} \bigr)
- |\calXN|\bigl( v_N - \frac{2}{|\calXN|}} \bigr)
\Bigr) \\
&\hspace{38ex} (\text{by $|\calXS|\le|\calXN|$}) \\
&\le \frac{1}{C_w} \Bigl({\textstyle
- 4 w + |\calXN|\bigl( \frac{2}{|\calXN|} - \frac{1}{u'} \bigr)}
\Bigr)
.
\end{align*}
Recall that $u'=e^{\epsilon/2}-1$.
By $\epsilon < 2\ln \bigl(\frac{|\calXN|}{2}+1\bigr)$, we have $e^{\epsilon/2} < \frac{|\calXN|}{2}+1$, hence:
\begin{align*}
\frac{2}{|\calXN|} - \frac{1}{u'}
= \frac{2}{|\calXN|} - \frac{1}{e^{\epsilon/2}-1}
< \frac{2}{|\calXN|} - \frac{2}{|\calXN|}
= 0.
\end{align*}
Hence $\frac{dF(w)}{dw} < 0$.
Therefore for $w\in[0,1)$, $F(w)$ is decreasing in $w$.
\end{proof}

Now we prove Proposition~\ref{prop:l1_restRAPPOR:approx} as follows.

\begin{proof}
Let $0 < \epsilon < 2\ln\bigl(\frac{|\calXN|}{2}+1\bigr)$. 
As with the proof for Proposition~\ref{prop:l1_restRR:eps=0},
let $\calC_{SN}$ be the set of all distributions $\bmp^*$ over $\calX$ that satisfy:
\begin{itemize}
\item for any $1 \le j \le |\calXS|$,\,
 $\bmp^*(x_j) = \frac{\bmp^*(\calXS)}{|\calXS|}$, and
\item for any $|\calXS|+1 \le j \le |\calX|$,\,
 $\bmp^*(x_j) = \frac{\bmp^*(\calXN)}{|\calXN|} = \frac{1 - \bmp^*(\calXS)}{|\calXN|}$.
\end{itemize}
By the Cauchy-Schwarz inequality, for any $a_1,a_2,\ldots,a_L\ge~0$, $\sum_{i=1}^{L} \sqrt{a_i} \le \sqrt{L \sum_{i=1}^{L} a_i}$, where the equality holds iff $a_1=a_2=\ldots=a_L$.
Hence by (\ref{eq:restRAPPOR-l1-loss:formal}) we obtain:
\begin{align}
\bbE \left[ l_1(\hbmp,\bmp) \right]
\lesssim& \sqrt{\frac{2}{n\pi}} \Biggl(\!\sqrt{|\calXS| \sum_{j=1}^{|\calXS|} \bigl( \bmp(x_j) + 1/u' \bigr) \bigl( v_N - \bmp(x_j)  \bigr)} \nonumber\\
&\hspace{6ex} + \!\sqrt{|\calXN|\hspace{-2ex} \sum_{j=|\calXS|+1}^{|\calX|}\hspace{-2.3ex} \bmp(x_j) \bigl( v_N - \bmp(x_j) \bigr)} \Biggr)
,
\nonumber
\end{align}
where the equality holds iff $\bmp\in\calC_{SN}$.

Therefore we obtain:
\begin{align}
&~ \bbE \left[ l_1(\hbmp,\bmp) \right] \nonumber \\
&\lesssim \sqrt{\frac{2}{n\pi}} \!\biggl(\!\sqrt{|\calXS|^2 \Bigl( {\textstyle\frac{\bmp(\calXS)}{|\calXS|}} + 1/u' \Bigr) \Bigl( v_N - {\textstyle\frac{\bmp(\calXS)}{|\calXS|}} \Bigr)} \label{eq:l1-loss-RAP-with-pxs}\\
&\hspace{9ex} + \sqrt{|\calXN|^2{\textstyle\frac{1-\bmp(\calXS)}{|\calXN|}} \Bigl( v_N - {\textstyle\frac{1-\bmp(\calXS)}{|\calXN|}} \Bigr)} \biggr) \nonumber\\
&= \sqrt{\frac{2}{n\pi}} \!\biggl(\!\sqrt{ \Bigl( \bmp(\calXS) + |\calXS|/u' \Bigr) \Bigl( v_N|\calXS| - \bmp(\calXS) \Bigr)} \nonumber\\
&\hspace{9ex} + \sqrt{\bigl(1-\bmp(\calXS)\bigr) \Bigl( v_N|\calXN| - 1 + \bmp(\calXS) \Bigr)} \biggr) \nonumber\\
&= \sqrt{\frac{2}{n\pi}}\, F(\bmp(\calXS)) %\nonumber
\label{eq:l1_RAPPOR_FpXs}
\end{align}
where $F$ is defined in Lemma~\ref{lem:l1_restRAPPOR:decreasing}.
Note that in (\ref{eq:l1-loss-RAP-with-pxs}), $\approx$ holds iff $\bmp\in\calC_{SN}$.

By Lemma~\ref{lem:l1_restRAPPOR:decreasing} and
$0 < \epsilon < 2\ln(\frac{|\calXN|}{2}+1)$, 
$F(\bmp(\calXS))$ is maximized when $\bmp(\calXS)=0$.
Hence the right-hand side of (\ref{eq:restRAPPOR-l1-loss:formal}) is maximized when $\bmp(\calXS) = 0$ and $\bmp\in\calC_{SN}$, i.e., when $\bmp$ is the uniform distribution $\bmpUN$ over~$\calXN$.

Therefore we obtain:
\begin{align*}
\bbE \left[ l_1(\hbmp,\bmp) \right] 
&\lesssim\!\bbE \left[ l_1(\hbmp,\bmpUN) \right] \\
&=\!\sqrt{\frac{2}{n\pi}}\,F(0) \\
&=\!\sqrt{\frac{2}{n\pi}} \!\biggl(\!\sqrt{ {\textstyle\frac{|\calXS|^2}{u'}} v_N} + \sqrt{ v_N|\calXN| - 1 } \biggr) \\
&=\!\sqrt{\frac{2}{n\pi}} \!\biggl(\!\frac{e^{\epsilon/4}|\calXS|}{e^{\epsilon/2}-1} + \sqrt{ \frac{e^{\epsilon/2}|\calXN|}{e^{\epsilon/2}-1} - 1 } \biggr)
.
\end{align*}
\end{proof}

\subsubsection{$l_1$ loss in the high privacy regime}
\label{sub:l1_restRAPPOR:high_privacy}
Consider the high privacy regime where $\epsilon\approx0$. In this case, $e^{\epsilon/2} - 1 \approx \epsilon/2$. 
By using this approximation, we simplify the $l_1$ loss of the utility-optimized RAPPOR for both the cases where $|\calXS|\le|\calXN|$ and where $|\calXS|>|\calXN|$.\\

\noindent{\textbf{Case 1: $|\calXS|\le|\calXN|$.}}~~By Proposition~\ref{prop:l1_restRAPPOR:approx},
the expected $l_1$ loss of the $(\calXS,\epsilon)$-utility-optimized RAPPOR mechanism is given by:
\begin{align*}
\bbE \left[ l_1(\hbmp,\bmp) \right]
&\lesssim \bbE \left[ l_1(\hbmp,\bmpUN) \right] \nonumber \\
&=\!\sqrt{\frac{2}{n\pi}} \!\biggl(\!\sqrt{ {\textstyle\frac{|\calXS|^2}{u'}} v_N} + \sqrt{ v_N|\calXN| - 1 } \biggr) \nonumber \\
&\approx\!\sqrt{\frac{2}{n\pi}} \!\biggl(\!\sqrt{ {\textstyle\frac{|\calXS|^2\frac{2}{\epsilon}}{ {\textstyle\frac{\epsilon}{2}} }}} + \sqrt{ {\textstyle\frac{2|\calXN|}{\epsilon}-1} } \biggr) \nonumber \\
&\hspace{31ex} (\text{by $e^{\epsilon/2}-1\approx\epsilon/2$}) \nonumber \\
&=\!\sqrt{\frac{2}{n\pi}} \!\biggl(\!\frac{2|\calXS|}{\epsilon} + \sqrt{ \frac{2|\calXN|-\epsilon}{\epsilon} } \biggr) \nonumber \\
&\approx\!\sqrt{\frac{2}{n\pi}} \cdot\frac{2|\calXS|}{\epsilon} \hspace{23ex}(\text{by $\epsilon\approx0$}). %\nonumber \\
\end{align*}\\

\noindent{\textbf{Case 2: $|\calXS|>|\calXN|$.}}~~Let 
$F$ be the function defined in Lemma~\ref{lem:l1_restRAPPOR:decreasing}, 
$w^* = \argmax_{w\in[0,1]} F(w)$, 
and
$\bmp^{\!*}$ be the prior distribution over $\calX$ defined by:
\begin{align*}
\bmp^{\!*}(x_j) = 
\begin{cases}
\frac{w^*}{|\calXS|}
~~~(\text{if $1 \le j \le |\calXS|$}) \\
\frac{1 - w^*}{|\calXN|}
~~~(\text{if $|\calXS|+1 \le j \le |\calX|$})
\end{cases}
\end{align*}
Then, by (\ref{eq:l1_RAPPOR_FpXs}), 
$\bbE \left[ l_1(\hbmp,\bmp) \right]$ is maximized by $\bmp^{\!*}$. Thus, 
for $\epsilon\approx0$, the expected $l_1$ loss of the $(\calXS,\epsilon)$-utility-optimized RAPPOR mechanism is given by:
\begin{align}
\bbE \left[ l_1(\hbmp,\bmp) \right] %\nonumber\\[0.5ex]
&\lesssim \left[ l_1(\hbmp,\bmp^{\!*}) \right] \nonumber \\
&=\!\sqrt{\frac{2}{n\pi}}\,F(w^*) \nonumber \\
&\approx\!\sqrt{\frac{2}{n\pi}} \!\biggl(\! {\textstyle 
\sqrt{ \bigl( w^* + {\textstyle\frac{|\calXS|}{\epsilon/2}} \bigr) \bigl( \frac{2|\calXS|}{\epsilon} - w^* \bigr) } } \nonumber \\
&\hspace{9ex} + {\textstyle \sqrt{ \bigl(1-w^*\bigr) \bigl( \frac{2|\calXN|}{\epsilon} - 1 + w^* \bigr) } }\biggr) \nonumber \\
&\hspace{14ex} (\text{by $u' = e^\epsilon-1\approx{\textstyle\frac{\epsilon}{2}}$ and $v_N\approx {\textstyle\frac{2}{\epsilon}}$}) \nonumber \\
&=\!\sqrt{\frac{2}{n\pi}} \!\biggl(\! {\textstyle 
\sqrt{ \frac{4|\calXS|^2}{\epsilon^2} + O(\epsilon^{-1}) } } + {\textstyle \sqrt{ O(\epsilon^{-1}) } }\biggr) \nonumber \\
&\hspace{32ex} (\text{by $0\le w^* \le 1$}) \nonumber \\
&\approx\!\sqrt{\frac{2}{n\pi}} \cdot \frac{2|\calXS|}{\epsilon}.
\hspace{14.5ex}(\text{by $\epsilon\approx0$})
\nonumber
\end{align}\\

In summary, the expected $l_1$ loss of the utility-optimized RAPPOR 
is at most 
$\sqrt{\frac{2}{n\pi}} \cdot \frac{2|\calXS|}{\epsilon}$ 
in the high privacy regime, 
irrespective of whether $|\calXS|\le|\calXN|$ or not. 
It is shown in \cite{Kairouz_ICML16} that 
the expected $l_1$ loss of the $\epsilon$-RAPPOR is at most 
$\sqrt{\frac{2}{n\pi}} \cdot\frac{2|\calX|}{\epsilon}$ 
when $\epsilon \approx 0$. 
Thus, 
the expected $l_1$ loss of the $(\calXS,\epsilon)$-utility-optimized RAPPOR is much smaller than that of the $\epsilon$-RAPPOR when $|\calXS| \ll |\calX|$. 

Note that the expected $l_1$ loss of the utility-optimized RAPPOR in the worst case can also be expressed as 
$\Theta(\frac{|\calXS|}{\sqrt{n \epsilon^2}})$ 
in this case. 
As described in Section~\ref{sub:ULDP_theoretical}, 
this is ``order'' optimal among all ULDP mechanisms.

\subsubsection{$l_1$ loss in the low privacy regime}
\label{sub:l1_restRAPPOR:low_privacy}%\subsubsection{$l_1$ loss for large $\epsilon$}
Consider the low privacy regime where $\epsilon=\ln|\calX|$ and 
$|\calXS| \ll |\calX|$. 
In this case, 
we have $\epsilon=\ln|\calX| < 2\ln(\frac{|\calXN|}{2}+1)$
(since $|\calX| < (\frac{|\calXN|}{2}+1)^2$).
Then by Proposition~\ref{prop:l1_restRAPPOR:approx},
the expected $l_1$ loss of the $(\calXS,\epsilon)$-utility-optimized RAPPOR mechanism is given by:
\begin{align}
&\bbE \left[ l_1(\hbmp,\bmp) \right] \nonumber\\
&\lesssim \bbE \left[ l_1(\hbmp,\bmpUN) \right] \nonumber \\
&=\!\sqrt{\frac{2}{n\pi}} \!\biggl(\!\sqrt{ {\textstyle\frac{|\calXS|^2}{u'}} v_N} + \sqrt{ v_N|\calXN| - 1 } \biggr) \nonumber \\
&\approx\!\sqrt{\frac{2}{n\pi}} \!\biggl(\!{\textstyle\frac{|\calXS|}{\sqrt{e^{\epsilon/2}-1}}} + \!\sqrt{ |\calXN| - 1 } \biggr) \nonumber \\
&\hspace{20ex} \Bigl(\text{by $\textstyle v_N = \frac{e^{\epsilon/2}}{e^{\epsilon/2}-1} = \frac{\sqrt{|\calX|}}{\sqrt{|\calX|}-1} \approx 1$}\Bigr) \nonumber \\
&\approx\!\sqrt{\frac{2}{n\pi}} \!\biggl(\!{\textstyle\frac{|\calXS|}{\sqrt{\sqrt{|\calX|}-1}}} + \sqrt{ |\calX| - 1 } \biggr) \hspace{4ex}(\text{by $|\calXN|\approx|\calX|$}) \nonumber \\
&\approx\!\sqrt{\frac{2(|\calX|-1)}{n\pi}} \biggl( 1 + \frac{|\calXS|}{|\calX|^\frac{3}{4}} \biggr). \hspace{8.2ex}(\text{by $|\calX| \gg 1$}) 
\label{eq:l1_restRAPPOR_low_XS_ll_X}
\end{align}
When $|\calXS| \ll |\calX|^{\frac{3}{4}}$, 
the right-hand side of (\ref{eq:l1_restRAPPOR_low_XS_ll_X}) can be simplified as follows: 
\begin{align}
\sqrt{\frac{2(|\calX|-1)}{n\pi}} \biggl( 1 + \frac{|\calXS|}{|\calX|^\frac{3}{4}} \biggr)
\approx\!\sqrt{\frac{2(|\calX|-1)}{n\pi}}. 
\nonumber
\end{align}
Note that the expected $l_1$ loss of the non-private mechanism is 
at most 
$\sqrt{\frac{2(|\calX| - 1)}{n\pi}}$ \cite{Kairouz_ICML16}. 
Thus, when $\epsilon=\ln|\calX|$ and 
$|\calXS| \ll |\calX|^{\frac{3}{4}}$, 
the $(\calXS,\epsilon)$-utility-optimized RAPPOR achieves almost the same data utility as the non-private mechanism, 
whereas the expected $l_1$ loss of the $\epsilon$-RAPPOR is $\sqrt{|\calX|}$ times larger than that of the non-private mechanism \cite{Kairouz_ICML16}.

\section{L2 loss of the utility-optimized Mechanisms}
\label{sec:proofs_utility_l2loss}

In this section we 
theoretically analyze 
the $l_2$ loss of the utility-optimized RR and the utility-optimized RAPPOR. 
Table~\ref{tab:l_2_loss} summarizes the $l_2$ loss of each obfuscation mechanism. 
We also show the results of the MSE in our experiments. 

\begin{table*}[t]
\caption{$l_2$ loss of each obfuscation mechanism in the worst case (RR: randomized response, RAP: RAPPOR, uRR: utility-optimized RR, uRAP: utility-optimized RAPPOR, no privacy: non-private mechanism, *1: approximation in the case where $|\calXS| \ll |\calX|$).}
\centering
\renewcommand{\arraystretch}{1.2}
\hbox to\hsize{\hfil
\begin{tabular}{l|c|c}
\hline
Mechanism			&	$\epsilon \approx 0$					&	$\epsilon = \ln |\calX|$\\
\hline
RR					&	$\frac{|\calX|(|\calX| - 1)}{n \epsilon^2}$	&	$\frac{4}{n} \Bigl( 1 - \frac{1}{|\calX|} \Bigr)$  \\
\hline
RAP				&	$\frac{4|\calX|}{n\epsilon^2} \Bigl( 1 - \frac{1}{|\calX|} \Bigr)$			& $\frac{\sqrt{|\calX|}}{n} \Bigl( 1 - \frac{1}{|\calX|} \Bigr)$\\
\hline
uRR		&	$\frac{|\calXS|(|\calXS| - 1)}{n \epsilon^2}$ (see Appendix~\ref{sub:l2_restRR:high_privacy})	&	$\frac{1}{n}$ $^{(*1)}$ (see Appendix~\ref{sub:l2_restRR:low_privacy}) \\
\hline
uRAP	&	$\frac{4|\calXS|}{n\epsilon^2}$ (see Appendix~\ref{sub:l2_restRAPPOR:high_privacy})	&	$\frac{1}{n} \Bigl( {\textstyle 1 + \frac{|\calXS|+1}{\sqrt{|\calX|}}} \Bigr)$ $^{(*1)}$ (see Appendix~\ref{sub:l2_restRAPPOR:low_privacy}) \\
\hline
no privacy	&	\multicolumn{2}{|c}{$\frac{1}{n} \Bigl( 1 - \frac{1}{|\calX|} \Bigr)$} \\
\hline
\end{tabular}
\hfil}
\label{tab:l_2_loss}
\end{table*}

\subsection{$l_2$ loss of the utility-optimized RR}
\label{sub:proof_prop_l2_restRR}

We first present the $l_2$-loss of the $(\calXS,\epsilon)$-utility-optimized RR.

\begin{restatable}[$l_2$ loss of the \colorBB{uRR}]{prop}{propLtwoRestRR}
\label{prop:l2_restRR:new} 
The expected $l_2$-loss of the $(\calXS,\epsilon)$-\colorBB{uRR} 
mechanism is given by:
\begin{align}
\bbE[l_2^2(\hbmp,\bmp)] 
&= \frac{2(e^{\epsilon}-1)(|\calXS| - \bmp(\calXS)) + |\calXS|(|\calXS|-1)}{n (e^{\epsilon}-1)^2} \nonumber \\
&~~~ + \frac{1}{n} \bigl( 1 -\sum_{x \in \calX} \bmp(x)^2 \bigr).
\label{eq:restRR-l2-loss:formal}
\end{align}
\end{restatable}

\begin{proof}
Let $u = |\calXS| + e^\epsilon - 1$, $u'=e^\epsilon-1$, and $v = \frac{u}{u'}$.
By $\epsilon>0$, we have $u>0$ and $v>0$.

Let $\bmt$ be a frequency distribution of the obfuscated data $\bmY$; i.e., $\bmt(x) = \hbmm(x)n$. 
Since $\bmt(x)$ follows the binomial distribution with parameters $n$ and $\bmm(x)$, 
the mean is given by $\bbE[\bmt(x)] = n\bmm(x)$, and the variance of $\bmt(x)$ is given by $\var(\bmt(x)) = n\bmm(x)(1 - \bmm(x))$. 

Then, 
by (\ref{eq:bmm_restRR}) and (\ref{eq:hbmp_restRR}), 
the $l_2$-loss of $\hbmp$ can be written as follows:
\begin{align}
\expectym \left[ l_2^2(\hbmp,\bmp) \right] \nonumber
&= \bbE \left[ \sum_{x \in \calX} (\hbmp(x) - \bmp(x))^2 \right] \nonumber \\
&= \bbE \left[ \sum_{x \in \calX} v^2 \cdot (\hbmm(x) - \bmm(x))^2 \right]  \nonumber \\
&= \sum_{x \in \calX} v^2 \cdot \bbE \left[\, (\hbmm(x) - \bmm(x))^2 \,\right] \nonumber \\
&= \sum_{x \in \calX} v^2 \cdot \bbE \left[\, \left(\frac{\bmt(x)}{n} - \bbE \left[ \frac{\bmt(x)}{n} \right] \right)^2 \,\right] \nonumber \\
&= \sum_{x \in \calX} v^2 \cdot \frac{\var(\bmt(x))}{n^2} \nonumber \\
&= v^2 \sum_{x \in \calX} \frac{\bmm(x)(1 - \bmm(x))}{n} \nonumber \\
&= \frac{v^2}{n} \Bigl( 1 - \sum_{x \in \calX} \bmm(x)^2 \Bigr)
.
\label{eq:restRR-l2-loss:outdist}
\end{align}

It follows from (\ref{eq:bmm_restRR}) that for $x\in\calXS$, $\bmm(x) = \bmp(x)/v + 1/u$, and for $x\in\calXN$, $\bmm(x) = \bmp(x)/v$.
Therefore, we obtain:
\begin{align}
&~~~\bbE \left[ l_2^2(\hbmp,\bmp) \right] \nonumber\\
&= \frac{v^2}{n} \Bigl( 1 - \sum_{x \in \calXS} \bigl( \bmp(x)/v + 1/u \bigr)^2 - \sum_{x \in \calXN} \bigl( \bmp(x)/v \bigr)^2 \Bigr) \nonumber \\
&= \frac{1}{n} \Bigl( v^2 - \sum_{x \in \calXS} \bigl( \bmp(x) + 1/u' \bigr)^2 - \sum_{x \in \calXN} \bmp(x)^2 \Bigr) \nonumber \\
&= \frac{1}{n} \Bigl( v^2 - \sum_{x \in \calX} \bmp(x)^2 - \sum_{x \in \calXS} 2\bmp(x)/u' - \sum_{x \in \calXS} {\textstyle\frac{1}{u'^2}} \bigr) \Bigr) \nonumber \\
&= \frac{1}{n} \Bigl( v^2 - \sum_{x \in \calX} \bmp(x)^2 - {\textstyle \frac{2\bmp(\calXS)}{u'}} - {\textstyle \frac{|\calXS|}{u'^2}} \bigr) \Bigr) \nonumber \\
&= \frac{u^2 - 2 u'\bmp(\calXS) - |\calXS|}{n u'^2} - \frac{1}{n} \sum_{x \in \calX} \bmp(x)^2 \nonumber \\
&= {\textstyle\frac{(|\calXS| + e^\epsilon - 1)^2 - 2(e^{\epsilon}-1)\bmp(\calXS) - |\calXS|}{n (e^{\epsilon}-1)^2}} - {\textstyle\frac{1}{n}} \sum_{x \in \calX} \bmp(x)^2 \nonumber \\
&= {\textstyle\frac{2(e^{\epsilon}-1)(|\calXS| - \bmp(\calXS)) + |\calXS|(|\calXS|-1)}{n (e^{\epsilon}-1)^2}} + {\textstyle\frac{1}{n}} \Bigl( 1 -\sum_{x \in \calX} \bmp(x)^2 \Bigr) \nonumber
\end{align}
\end{proof}

\subsubsection{Maximum of the $l_2$ loss}

Next we show that when $0 < \epsilon < \ln(|\calXN|+1)$, the $l_2$ loss is maximized by the uniform distribution $\bmpUN$ over $\calXN$.

\begin{restatable}{prop}{propLtwoRestRRsmall}
\label{prop:l2_restRR:eps=0}
For any $0 < \epsilon < \ln(|\calXN|+1)$,
(\ref{eq:restRR-l2-loss:formal})
is maximized by the uniform distribution $\bmpUN$ over $\calXN$:
\begin{align*}
\bbE \left[ l_2^2(\hbmp,\bmp) \right] 
&\le \bbE \left[ l_2^2(\hbmp,\bmpUN) \right] \nonumber\\
&={\textstyle\frac{|\calXS|(|\calXS| + 2e^{\epsilon} - 3)}{n (e^{\epsilon}-1)^2} + \frac{1}{n} \bigl( 1 - \frac{1}{|\calXN|} \bigr)}. \nonumber
\end{align*}
\end{restatable}

To show this proposition, we first show the following lemma.

\begin{lemma}
\label{lem:l2_restRR:decreasing}
For $w\in[0,1]$, we define $F(w)$ by:
\begin{align*}
F(w) &= 
- \frac{2}{n (e^{\epsilon}-1)} w
+ \frac{1}{n} \Bigl( 1 - \frac{w^2}{|\calXS|} - \frac{(1-w)^2}{|\calXN|} \Bigr)
{.}
\end{align*}
Then for any $0 < \epsilon < \ln(|\calXN|+1)$,
$F(w)$ is decreasing in $w$.
\end{lemma}

\begin{proof}
Let $0 < \epsilon < \ln(|\calXN|+1)$ and $w\in[0, 1]$.
Then $e^{\epsilon}-1 < |\calXN|$.
\begin{align*}
\frac{d F(w)}{dw}
&= 
- \frac{2}{n (e^{\epsilon}-1)}
+ \frac{1}{n} \Bigl( - \frac{2w}{|\calXS|} - \frac{2w-2}{|\calXN|} \Bigr) 
%\\
\end{align*}
\begin{align*}
&< 
- \frac{2}{n |\calXN|}
+ \frac{1}{n} \Bigl( - \frac{2w}{|\calXS|} - \frac{2w}{|\calXN|} + \frac{2}{|\calXN|} \Bigr) \\
&\le 0. \hspace{27ex}(\text{by $w\ge0$})
\end{align*}
Therefore, $F(w)$ is decreasing in $w$.
\end{proof}

Now we prove Proposition~\ref{prop:l2_restRR:eps=0} as follows.

\begin{proof}
Let 
$M = {\textstyle\frac{|\calXS|(|\calXS| + 2e^{\epsilon} - 3)}{n (e^{\epsilon}-1)^2}}$.
By (\ref{eq:restRR-l2-loss:formal}), we have:
\begin{align}
\bbE \left[ l_2^2(\hbmp,\bmp) \right]
= M 
- {\textstyle\frac{2}{n (e^{\epsilon}-1)}\bmp(\calXS)}
+ {\textstyle\frac{1}{n}} \Bigl( 1 -\sum_{x \in \calX}\!\bmp(x)^2 \Bigr).\nonumber
\end{align}

Let $\calC_{SN}$ be the set of 
distributions $\bmp^*$ over $\calX$ such that:
\begin{itemize}
\item for any $x\in\calXS$,\,
 $\bmp^*(x) = \frac{\bmp^*(\calXS)}{|\calXS|}$, and
\item for any $x\in\calXN$,\,
 $\bmp^*(x) = \frac{\bmp^*(\calXN)}{|\calXN|} = \frac{1 - \bmp^*(\calXS)}{|\calXN|}$.
\end{itemize}
By the inequality of arithmetic and geometric means,
we obtain:
\begin{align*}
\sum_{x\in\calXS} \bmp(x)^2
\ge \sqrt[|\calXS|]{\textstyle \prod_{x\in\calXS} \bmp(x)^2} \cdot |\calXS|,
\end{align*}
where the equality holds iff for all $x\in\calXS$, $\bmp(x)=\frac{\bmp(\calXS)}{|\calXS|}$.
An analogous inequality holds for $\calXN$.
Therefore we obtain:
\begin{align*}
\sum_{x\in\calX} \bmp(x)^2
&\ge\! \sqrt[|\calXS|]{\prod_{x\in\calXS} \bmp(x)^2} \cdot |\calXS|
+\!\sqrt[|\calXN|]{\prod_{x\in\calXN} \bmp(x)^2} \cdot |\calXN| \\
&={\textstyle\frac{\bmp(\calXS)^2}{|\calXS|} + \frac{\bmp(\calXN))^2}{|\calXN|}} \\
&={\textstyle\frac{\bmp(\calXS)^2}{|\calXS|} + \frac{(1 - \bmp(\calXS))^2}{|\calXN|}},
\end{align*}
where the equality holds iff $\bmp\in\calC_{SN}$.

Hence we obtain:
\begin{align}
&~~~~ \bbE \left[ l_2^2(\hbmp,\bmp) \right] \nonumber\\
&= M 
- {\textstyle\frac{2}{n (e^{\epsilon}-1)}\bmp(\calXS)}
+ {\textstyle\frac{1}{n}} \Bigl( 1 -\sum_{x \in \calX}\!\bmp(x)^2 \Bigr) \nonumber \\
&\le {\textstyle M 
- \frac{2}{n (e^{\epsilon}-1)}\bmp(\calXS)
+ \frac{1}{n} \Bigl( 1 - \frac{\bmp(\calXS)^2}{|\calXS|} - \frac{(1 - \bmp(\calXS))^2}{|\calXN|} \Bigr) } \nonumber \\
&= {\textstyle M + F(\bmp(\calXS)), } \nonumber %\\
\end{align}
where $F$ is defined in Lemma~\ref{lem:l2_restRR:decreasing}, and the equality holds iff $\bmp\in\calC_{SN}$.

By Lemma~\ref{lem:l2_restRR:decreasing} and $0 < \epsilon < \ln(|\calXN|+1)$, $F(\bmp(\calXS))$ is maximized when $\bmp(\calXS) = 0$.
Therefore $\bbE \left[ l_2^2(\hbmp,\bmp) \right]$ is maximized when $\bmp(\calXS) = 0$ and $\bmp\in\calC_{SN}$, i.e., when $\bmp$ is the uniform distribution $\bmpUN$ over $\calXN$.

Therefore we obtain:
\begin{align*}
\bbE \left[ l_2^2(\hbmp,\bmp) \right]
&\le \bbE \left[ l_2^2(\hbmp,\bmpUN) \right] \\
&= M + F(0) \\
&= {\textstyle\frac{|\calXS|(|\calXS| + 2e^{\epsilon} - 3)}{n (e^{\epsilon}-1)^2} + \frac{1}{n} \bigl( 1 - \frac{1}{|\calXN|} \bigr) }
\end{align*}
\end{proof}

Next, we show that when $\epsilon \ge \ln(|\calXN|+1)$, the $l_2$ loss is maximized by a mixture of the uniform distribution $\bmpUS$ over $\calXS$ and the uniform distribution $\bmpUN$ over $\calXN$.

\begin{restatable}{prop}{propLtwoRestRRbig}
\label{prop:l2_restRR:eps=lnX}
Let $\bmp^*$ be a distribution over $\calX$ defined~by:
\begin{align*}
\bmp^*(x) =
\begin{cases}
\frac{1 - |\calXN|/(e^{\epsilon}-1)}{|\calXS| + |\calXN|} 
~~\mbox{(if $x\in\calXS$)}
\\[2ex]
\frac{1 + |\calXS|/(e^{\epsilon}-1)}{|\calXS| + |\calXN|}
~~\mbox{(otherwise)}
\end{cases}
\end{align*}
Then for any $\epsilon \ge \ln(|\calXN|+1)$,
the right-hand side of (\ref{eq:restRR-l2-loss:formal})
is maximized by $\bmp^*$:
\begin{align*}
\bbE \left[ l_2^2(\hbmp,\bmp) \right] 
\le \bbE \left[ l_2^2(\hbmp,\bmp^*) \right]
= \frac{(|\calXS|+e^\epsilon-1)^2}{n(e^\epsilon-1)^2} \Bigl( 1 - \frac{1}{|\calX|} \Bigr).
\end{align*}
\end{restatable}

\begin{proof}
We show that for any $\epsilon \ge \ln(|\calXN|+1)$, 
the right-hand side of (\ref{eq:restRR-l2-loss:formal})
is maximized when $\bmp = \bmp^*$.
(Note that by $\epsilon \ge \ln(|\calXN|+1)$, $\bmp^*(x)\ge 0$ holds for all $x\in\calX$.)
To show this, we recall that if $\bmp = \bmp^*$ then $\bmm$ is the uniform distribution over $\calY$, as shown 
in the proof for Proposition\ref{prop:l1_restRR:eps=lnX}.

Let $v = \frac{|\calXS|+e^{\epsilon}-1}{e^{\epsilon}-1}$.
By (\ref{eq:restRR-l2-loss:outdist}) and the inequality of arithmetic and geometric means, we obtain:
\begin{align*}
\bbE \left[ l_2^2(\hbmp,\bmp) \right]
&= \frac{v^2}{n} \Bigl( 1 - \sum_{x \in \calX} \bmm(x)^2 \Bigr) \\
&\le \frac{v^2}{n} \Bigl( 1 - \sqrt[|\calX|]{\textstyle \prod_{x\in\calX} \bmm(x)^2} \cdot |\calX| \Bigr),
\end{align*}
where the equality holds iff $\bmm$ is the uniform distribution over $\calX$.
Hence:
\begin{align*}
\bbE \left[ l_2^2(\hbmp,\bmp) \right]
&\le \bbE \left[ l_2^2(\hbmp,\bmp^*) \right] \\
&= \frac{v^2}{n} \Bigl( 1 - \sqrt[|\calX|]{\textstyle \prod_{x\in\calX} \bigl(\frac{1}{|\calX|}\bigr)^2} \cdot |\calX| \Bigr) \\
&= \frac{v^2}{n} \Bigl( 1 - \frac{1}{|\calX|} \Bigr) \\
&= \frac{(|\calXS|+e^\epsilon-1)^2}{n(e^\epsilon-1)^2} \Bigl( 1 - \frac{1}{|\calX|} \Bigr).
\end{align*}
\end{proof}

\subsubsection{$l_2$ loss in the high privacy regime}
\label{sub:l2_restRR:high_privacy}
Consider the high privacy regime where $\epsilon\approx0$. 
In this case, $e^\epsilon - 1 \approx \epsilon$. 
By using this approximation, we simplify the $l_2$ loss of the 
\colorBB{uRR}. 

By Proposition~\ref{prop:l2_restRR:eps=0}, the expected $l_2$ loss of the $(\calXS,\epsilon)$-\colorBB{uRR} 
mechanism is maximized by $\bmpUN$:
\begin{align*}
\bbE \left[ l_2^2(\hbmp,\bmp) \right]
&\le \bbE \left[ l_2^2(\hbmp,\bmpUN) \right] \nonumber \\
&={\textstyle\frac{|\calXS|(|\calXS| + 2e^{\epsilon} - 3)}{n (e^{\epsilon}-1)^2} + \frac{1}{n} \bigl( 1 - \frac{1}{|\calXN|} \bigr)} \nonumber \\
&\approx{\textstyle\frac{|\calXS|(|\calXS| + 2\epsilon - 1) + \epsilon^2\bigl( 1 - \frac{1}{|\calXN|} \bigr)}{n \epsilon^2} } 
\hspace{3ex} (\text{by $e^\epsilon-1\approx\epsilon$}) \nonumber \\
&\approx{\textstyle\frac{|\calXS|(|\calXS| - 1)}{n \epsilon^2}.} 
\hspace{10ex} (\text{by $e^\epsilon-1\approx\epsilon$})
\end{align*}
It is shown in \cite{Kairouz_ICML16} that 
the expected $l_2$ loss of the $\epsilon$-RR is at most 
$\textstyle\frac{|\calX|(|\calX| - 1)}{n \epsilon^2}$ 
when $\epsilon \approx 0$. 
Thus, 
the expected $l_2$ loss of the $(\calXS,\epsilon)$-\colorBB{uRR} 
is much smaller than that of the $\epsilon$-RR when $|\calXS| \ll |\calX|$.\\

\subsubsection{$l_2$ loss in the low privacy regime}
\label{sub:l2_restRR:low_privacy}

Consider the low privacy regime where $\epsilon=\ln|\calX|$ and $|\calXS| \ll |\calX|$. 
By Proposition~\ref{prop:l2_restRR:eps=lnX},
the expected $l_2^2$ loss of the $(\calXS,\epsilon)$-\colorBB{uRR} 
is given by:
\begin{align*}
\bbE \left[ l_2^2(\hbmp,\bmp) \right]
&\le \bbE \left[ l_2^2(\hbmp,\bmp^*) \right] \\
&= \frac{(|\calXS|+|\calX|-1)^2}{n(|\calX|-1)^2} \Bigl( 1 - \frac{1}{|\calX|} \Bigr) \\
&= \frac{\bigl(1+\frac{|\calXS|-1}{|\calX|}\bigr)^2}{n\bigl(1-\frac{1}{|\calX|}\bigr)^2} \Bigl( 1 - \frac{1}{|\calX|} \Bigr) \\
&\approx \frac{1}{n}. \hspace{5ex}(\text{by $1 / |\calX| \approx 0$ and $|\calXS| / |\calX| \approx 0$}) 
\end{align*}
It should be noted that the expected $l_2$ loss of the non-private mechanism 
is 
at most $\frac{1}{n} ( 1 - \frac{1}{|\calX|} )$ \cite{Kairouz_ICML16},
and that 
$\frac{1}{n} ( 1 - \frac{1}{|\calX|} ) \approx \frac{1}{n}$ when $|\calX| \gg 1$. 
Thus, when $\epsilon=\ln|\calX|$ and $|\calXS| \ll |\calX|$,
the 
$(\calXS,\epsilon)$-\colorBB{uRR} 
achieves almost the same data utility as the non-private mechanism, 
whereas the expected $l_1$ loss of the $\epsilon$-RR is 
four times 
larger than that of the non-private mechanism \cite{Kairouz_ICML16}.

\subsection{$l_2$ loss of the utility-optimized RAPPOR}
\label{sub:proof_prop_l2_restRAPPOR}

We first present the $l_2$ loss of the $(\calXS,\epsilon)$-\colorBB{uRAP}.
Recall that 
$\bmm_j$ (resp.~$\hbmm_j$) is the true probability (resp.~empirical probability) that the $j$-th coordinate in obfuscated data is $1$.

\begin{restatable}[$l_2$ loss of the \colorBB{uRAP}]{prop}{propLoneRestRAP_l2}
\label{prop:l2_restRAPPOR}
Then the expected $l_2$-loss of the $(\calXS,\epsilon)$-\colorBB{uRAP} 
mechanism is given by:
\begin{align}
&\bbE \left[ l_2^2(\hbmp,\bmp) \right] \nonumber\\
&= 
\frac{1}{n} \Bigl( 1 + {\textstyle\frac{(|\calXS|+1)e^{\epsilon/2} - 1}{(e^{\epsilon/2}-1)^2}} - {\textstyle \frac{1}{e^{\epsilon/2}-1}} \bmp(\calXS) - \!\sum_{j = 1}^{|\calX|} \bmp(x_j)^2 \Bigr).
\label{eq:restRAPPOR-l2-loss:formal}
\end{align}
\end{restatable}

\begin{proof}
Let $v_S = \frac{e^{\epsilon/2}+1}{e^{\epsilon/2}-1}$, 
and $v_N = \frac{e^{\epsilon/2}}{e^{\epsilon/2}-1}$.
By $\epsilon>0$, we have $v_S>0$ and $v_N>0$.

For each $1\le j \le |\calX|$, let 
$\bmt_j$ be the number of users whose $j$-th coordinate in the obfuscated data is $1$; i.e., $\bmt_j = \hbmm_j n$.
Since $\bmt_j$ follows the binomial distribution with parameters $n$ and $\bmm_j$, 
the mean is given by $\bbE[\bmt_j] = n\bmm_j$, and the variance of $\bmt_j$ is given by $\var(\bmt_j) = n\bmm_j(1 - \bmm_j)$. 

Then, 
by (\ref{eq:bmmj_bmpj_restRAPPOR}) and (\ref{eq:bmpj_bmmj_restRAPPOR}), 
the $l_2$-loss of $\hbmp$ can be written as follows:
\begin{align}
& \expectym \left[ l_2^2(\hbmp,\bmp) \right] \nonumber \\
&= \bbE \left[ \sum_{x \in \calX} (\hbmp(x) - \bmp(x))^2 \right] \nonumber \\
&= \sum_{j = 1}^{|\calXS|} v_S^2 \cdot \bbE \left[ (\hbmm_j - \bmm_j)^2 \right]  %\nonumber \\[-1ex]
+ \!\sum_{j = |\calXS|+1}^{|\calX|}\! v_N^2 \cdot \bbE \left[\, (\hbmm_j - \bmm_j)^2 \,\right] \nonumber \\
&= \sum_{j = 1}^{|\calXS|} v_S^2 \cdot \frac{\var(\bmt(x_j))}{n^2}
 %\nonumber \\
+ \sum_{j = |\calXS|+1}^{|\calX|} v_N^2 \cdot \frac{\var(\bmt(x_j))}{n^2} \nonumber \\
&= \frac{v_S^2}{n} \sum_{j = 1}^{|\calXS|} \bmm_j(1 - \bmm_j) + \frac{v_N^2}{n} \hspace{-1ex}\sum_{j = |\calXS|+1}^{|\calX|}\hspace{-1ex} \bmm_j(1 - \bmm_j) 
.
\label{eq:restRAPPOR-l2-loss:outdist}
\end{align}

Let $u = e^{\epsilon/2} + 1$ and $u' = e^{\epsilon/2} - 1$.
Then $v_S = \frac{u}{u'}$ and $v_N = \frac{u-1}{u'}$.
It follows from (\ref{eq:bmmj_bmpj_restRAPPOR}) that for $1\le j \le |\calXS|$, $\bmm_j = \bmp(x_j)/v_S + 1/u$, and for $|\calXS|+1 \le j \le |\calX|$, $\bmm_j = \bmp(x_j)/v_N$.
Therefore, we obtain:
\begin{align}
&\bbE \left[ l_2^2(\hbmp,\bmp) \right] \nonumber\\
&= \frac{v_S^2}{n} \sum_{j = 1}^{|\calXS|} \bigl( {\textstyle\frac{\bmp(x_j)}{v_S} + \frac{1}{u}} \bigr) \bigl( 1 - {\textstyle\frac{\bmp(x_j)}{v_S} - \frac{1}{u}} \bigr) \nonumber \\
&~~~ + \frac{v_N^2}{n} \hspace{-1ex}\sum_{j = |\calXS|+1}^{|\calX|}\hspace{-1ex} \bigl( {\textstyle\frac{\bmp(x_j)}{v_N}} \bigr) \bigl( 1 - {\textstyle\frac{\bmp(x_j)}{v_N}} \bigr) \nonumber \\
&= \frac{1}{n} \sum_{j = 1}^{|\calXS|} \bigl( {\textstyle \bmp(x_j) + \frac{v_S}{u}} \bigr) \bigl( v_S - {\textstyle \bmp(x_j) - \frac{v_S}{u}} \bigr) \nonumber \\
&~~~ + \frac{1}{n} \hspace{-1ex}\sum_{j = |\calXS|+1}^{|\calX|}\hspace{-1ex} \bmp(x_j) \bigl( v_N - \bmp(x_j) \bigr) \nonumber \\
&= \frac{1}{n} \biggl( v_S {\textstyle \bigl(1 - \frac{2}{u}\bigr) }\bmp(\calXS) \nonumber \\
&\hspace{9mm} + v_N \bmp(\calXN) + {\textstyle\frac{v_S^2(u-1)}{u^2} |\calXS|} - \!\sum_{j = 1}^{|\calX|} \bmp(x_j)^2 \biggr) \nonumber \\
&= \frac{1}{n} \biggl( {\textstyle \bigl( \frac{v_S (u-2)}{u} - v_N \bigr)} \bmp(\calXS) \nonumber \\
&\hspace{9mm} + v_N + {\textstyle\frac{v_S^2(u-1)}{u^2}}|\calXS| - \!\sum_{j = 1}^{|\calX|} \bmp(x_j)^2 \biggr) \nonumber \\
&\hspace{36ex}(\text{by $\bmp(\calXN) = 1 - \bmp(\calXS)$}) \nonumber \\
&= \frac{1}{n} \Bigl( 1 + {\textstyle\frac{(|\calXS|+1)e^{\epsilon/2} - 1}{(e^{\epsilon/2}-1)^2}} - {\textstyle \frac{1}{e^{\epsilon/2}-1}} \bmp(\calXS) - \!\sum_{j = 1}^{|\calX|} \bmp(x_j)^2 \Bigr). \nonumber
\end{align}
\end{proof}

\subsubsection{Maximum of the $l_2$ loss}

Next we show that for any $0 < \epsilon < 2\ln(\frac{|\calXN|}{2}+1)$, the $l_2$ loss is maximized by the uniform distribution $\bmpUN$ over $\calXN$.

\begin{restatable}{prop}{propLoneRestRAPapprox_l2}
\label{prop:l2_restRAPPOR:approx}
For any $0 < \epsilon < 2\ln(\frac{|\calXN|}{2}+1)$,
the $l_2$ loss $\bbE \left[ l_2^2(\hbmp,\bmp) \right]$ is maximized when $\bmp =~\bmpUN$:
\begin{align*}
\bbE \left[ l_2^2(\hbmp,\bmp) \right]
&\le \bbE \left[ l_2^2(\hbmp,\bmpUN) \right] \nonumber \\
&=\!\frac{1}{n} \Bigl( {\textstyle 1 + \frac{(|\calXS|+1)e^{\epsilon/2} - 1}{(e^{\epsilon/2}-1)^2} - \frac{1}{|\calXN|}} \Bigr).
\end{align*}
\end{restatable}

To prove this proposition, we first show the lemma below.

\begin{lemma}
\label{lem:l2_restRAPPOR:decreasing}
For $w\in[0,1]$, we define $F(w)$ by:
\begin{align*}
F(w) &= {\textstyle - \frac{1}{e^{\epsilon/2}-1}} w - 
{\textstyle \frac{w^2}{|\calXS|} - \frac{(1 - w)^2}{|\calXN|} }
{.}
\end{align*}
For any $0 < \epsilon < 2\ln(\frac{|\calXN|}{2}+1)$,
$F(w)$ is decreasing in $w$.
\end{lemma}

\begin{proof}
Let $w\in[0, 1]$.
\begin{align*}
\frac{d F(w)}{dw}
&= - \frac{1}{e^{\epsilon/2}-1} - \frac{2w}{|\calXS|} - \frac{2w - 2}{|\calXN|} \\
&= - 2 \Bigl( \frac{1}{|\calXS|} + \frac{1}{|\calXN|} \Bigr) w
- \frac{1}{e^{\epsilon/2}-1} + \frac{2}{|\calXN|} \\
&< 0. \hspace{10ex}\bigl(\text{by ${\textstyle \epsilon < 2\ln(\frac{|\calXN|}{2}+1)}$ and $w\ge0$}\bigr)
\end{align*}
Therefore, $F(w)$ is decreasing in $w$.
\end{proof}

Now we prove Proposition~\ref{prop:l2_restRAPPOR:approx} as follows.

\begin{proof}
Let 
$M = 1 + {\textstyle\frac{(|\calXS|+1)e^{\epsilon/2} - 1}{(e^{\epsilon/2}-1)^2}}$.
By Proposition~\ref{prop:l2_restRAPPOR}, we have:
\begin{align}
\bbE \left[ l_2^2(\hbmp,\bmp) \right]
&= \frac{1}{n} \Bigl( {\textstyle M - \frac{1}{e^{\epsilon/2}-1}} \bmp(\calXS) - \!\sum_{j = 1}^{|\calX|} \bmp(x_j)^2 \Bigr). \nonumber
\end{align}

As with the proof for Proposition~\ref{prop:l2_restRR:eps=0},
let $\calC_{SN}$ be the set of all distributions $\bmp^*$ over $\calX$ that satisfy:
\begin{itemize}
\item for any $1 \le j \le |\calXS|$,\,
 $\bmp^*(x_j) = \frac{\bmp^*(\calXS)}{|\calXS|}$, and
\item for any $|\calXS|+1 \le j \le |\calX|$,\,
 $\bmp^*(x_j) = \frac{\bmp^*(\calXN)}{|\calXN|} = \frac{1 - \bmp^*(\calXS)}{|\calXN|}$.
\end{itemize}
By the inequality of arithmetic and geometric means,
we obtain:
\begin{align*}
\sum_{j=1}^{|\calXS|} \bmp(x_j)^2
\ge \sqrt[|\calXS|]{\textstyle \prod_{j=1}^{|\calXS|} \bmp(x_j)^2} \cdot |\calXS|,
\end{align*}
where the equality holds iff for all 
$1 \le j \le |\calXS|$, 
$\bmp(x_j)=\frac{\bmp(\calXS)}{|\calXS|}$.
An analogous inequality holds for $\calXN$.
Therefore we obtain:
\begin{align*}
&\sum_{j=1}^{|\calX|} \bmp(x_j)^2 \nonumber\\
&\ge\! \sqrt[|\calXS|]{\prod_{j=1}^{|\calXS|} \bmp(x_j)^2} \cdot |\calXS|
+\!\sqrt[|\calXN|]{\prod_{j=|\calXS|+1}^{|\calX|} \bmp(x_j)^2} \cdot |\calXN| \\
&={\textstyle\frac{\bmp(\calXS)^2}{|\calXS|} + \frac{\bmp(\calXN))^2}{|\calXN|}} \\
&={\textstyle\frac{\bmp(\calXS)^2}{|\calXS|} + \frac{(1 - \bmp(\calXS))^2}{|\calXN|}},
\end{align*}
where the equality holds iff $\bmp\in\calC_{SN}$.

Hence we obtain:
\begin{align}
\bbE \left[ l_2^2(\hbmp,\bmp) \right]
&= \frac{1}{n} \Bigl( {\textstyle M - \frac{1}{e^{\epsilon/2}-1}} \bmp(\calXS) - \!\sum_{j = 1}^{|\calX|} \bmp(x_j)^2 \Bigr) \nonumber \\
&\le \frac{1}{n} \Bigl( {\textstyle M - \frac{1}{e^{\epsilon/2}-1}} \bmp(\calXS) - 
{\textstyle \frac{\bmp(\calXS)^2}{|\calXS|} -
 \frac{(1 - \bmp(\calXS))^2}{|\calXN|} } \Bigr) \nonumber \\
&= \frac{1}{n} \bigl( {\textstyle M + F(\bmp(\calXS))} \bigr),
\label{eq:l2_RAPPOR_FpXs}
\end{align}
where $F$ is defined in Lemma~\ref{lem:l2_restRAPPOR:decreasing}, and the equality holds iff $\bmp\in\calC_{SN}$.

By Lemma~\ref{lem:l2_restRAPPOR:decreasing} and $0 < \epsilon < 2\ln(\frac{|\calXN|}{2}+1)$, $F(\bmp(\calXS))$ is maximized when $\bmp(\calXS) = 0$.
Therefore, $\bbE \left[ l_2^2(\hbmp,\bmp) \right]$ is maximized when $\bmp(\calXS) = 0$ and $\bmp\in\calC_{SN}$, i.e., when $\bmp$ is the uniform distribution $\bmpUN$ over $\calXN$.

Therefore we obtain:
\begin{align*}
\bbE \left[ l_2^2(\hbmp,\bmp) \right]
&\le \bbE \left[ l_2^2(\hbmp,\bmpUN) \right] \nonumber \\
&=\!\frac{1}{n} \bigl( {\textstyle M + F(\bmp(\calXS))} \bigr) \nonumber \\
&=\!\frac{1}{n} \Bigl( {\textstyle 1 + \frac{(|\calXS|+1)e^{\epsilon/2} - 1}{(e^{\epsilon/2}-1)^2} - \frac{1}{|\calXN|}} \Bigr).
\end{align*}
\end{proof}

\subsubsection{$l_2$ loss in the high privacy regime}
\label{sub:l2_restRAPPOR:high_privacy}

Consider the high privacy regime where $\epsilon\approx0$. In this case, $e^{\epsilon/2} - 1 \approx \epsilon/2$. 
By using this approximation, we simplify the $l_2$ loss of the 
\colorBB{uRAP}.

By Proposition~\ref{prop:l2_restRAPPOR:approx},
the expected $l_2$ loss of the $(\calXS,\epsilon)$-\colorBB{uRAP} 
mechanism is given by:
\begin{align*}
\bbE \left[ l_2^2(\hbmp,\bmp) \right]
&\le \bbE \left[ l_2^2(\hbmp,\bmpUN) \right] \nonumber \\
&=\!\frac{1}{n} \Bigl( {\textstyle 1 + \frac{(|\calXS|+1)e^{\epsilon/2} - 1}{(e^{\epsilon/2}-1)^2} - \frac{1}{|\calXN|}} \Bigr) \nonumber \\
&\approx\!\frac{1}{n} \Bigl( {\textstyle 1 + \frac{(|\calXS|+1)(\epsilon/2+1) - 1}{\epsilon^2/4} - \frac{1}{|\calXN|}} \Bigr) \nonumber \\
&\hspace{29ex} (\text{by $e^{\epsilon/2}-1\approx\epsilon/2$}) \nonumber \\
&\approx\!\frac{1}{n} \Bigl( {\textstyle 1 + \frac{4|\calXS|}{\epsilon^2} - \frac{1}{|\calXN|}} \Bigr) \hspace{4ex}(\text{by $\epsilon\approx0$}) \nonumber \\
&=\!\frac{1}{n} \Bigl( {\textstyle \frac{4|\calXS| + \epsilon^2(1 - \frac{1}{|\calXN|})}{\epsilon^2}} \Bigr) \nonumber \\
&\approx\!\frac{4|\calXS|}{n\epsilon^2}. \hspace{19ex}(\text{by $\epsilon\approx0$}) \nonumber
\end{align*}\\

Thus, the expected $l_2$ loss of the 
\colorBB{uRAP} 
is at most 
$\frac{4|\calXS|}{n\epsilon^2}$ 
in the high privacy regime. 
It is shown in \cite{Kairouz_ICML16} that 
the expected $l_2$ loss of the $\epsilon$-RAPPOR is at most 
$\frac{4|\calX|}{n\epsilon^2} (1 - \frac{1}{|\calX|})$ 
when $\epsilon \approx 0$. 
Thus, 
the expected $l_2$ loss of the $(\calXS,\epsilon)$-\colorBB{uRAP} 
is much smaller than that of the $\epsilon$-RAPPOR when $|\calXS| \ll |\calX|$.

Note that the expected $l_2$ loss of the 
\colorBB{uRAP} 
in the worst case can also be expressed as 
$\Theta(\frac{|\calXS|}{n\epsilon^2})$ 
in this case. 
As described in Section~\ref{sub:ULDP_theoretical}, 
this is ``order'' optimal among all ULDP mechanisms.

\subsubsection{$l_2$ loss in the low privacy regime}
\label{sub:l2_restRAPPOR:low_privacy}%\subsubsection{$l_2$ loss for large 

Consider the low privacy regime where $\epsilon=\ln|\calX|$ and $|\calXS| \ll |\calX|$.
By Proposition~\ref{prop:l2_restRAPPOR:approx},
the expected $l_2$ loss of the $(\calXS,\epsilon)$-\colorBB{uRAP} 
mechanism is given by:
\begin{align}
\bbE \left[ l_2^2(\hbmp,\bmp) \right]
&\le \bbE \left[ l_2^2(\hbmp,\bmpUN) \right] \nonumber \\
&=\!\frac{1}{n} \Bigl( {\textstyle 1 + \frac{(|\calXS|+1)e^{\epsilon/2} - 1}{(e^{\epsilon/2}-1)^2} - \frac{1}{|\calXN|}} \Bigr) \nonumber \\
&=\!\frac{1}{n} \Bigl( {\textstyle 1 + \frac{(|\calXS|+1)\sqrt{|\calX|} - 1}{(\sqrt{|\calX|}-1)^2} - \frac{1}{|\calXN|}} \Bigr) \nonumber \\
&=\!\frac{1}{n} \Bigl( {\textstyle 1 + \frac{|\calXS|+1 - \frac{1}{\sqrt{|\calX|}}}{ \sqrt{|\calX|} - 2 + \frac{1}{\sqrt{|\calX|}} }
- \frac{1}{|\calX| \bigl(1 - \frac{|\calXS|}{|\calX|} \bigr)}} \Bigr) \nonumber \\
&\approx\!\frac{1}{n} \Bigl( {\textstyle 1 + \frac{|\calXS|+1}{\sqrt{|\calX|}}} \Bigr). \hspace{2.2ex}(\text{by $|\calXS| \ll |\calX|$}) 
\label{eq:l2_restRAPPOR_low_XS_ll_X}
\end{align}
When $|\calXS| \ll \sqrt{|\calX|}$, 
the right side of (\ref{eq:l2_restRAPPOR_low_XS_ll_X}) is simplified as:
\begin{align}
\frac{1}{n} \Bigl( {\textstyle 1 + \frac{|\calXS|+1}{\sqrt{|\calX|}}} \Bigr)
\approx\!\frac{1}{n}. ~~~~(\text{by $|\calXS| / \sqrt{|\calX|} \approx 0$}) \nonumber
\end{align}
Note that the expected $l_2$ loss of the non-private mechanism is 
at most $\frac{1}{n} ( 1 - \frac{1}{|\calX|} )$ \cite{Kairouz_ICML16}, and that 
$\frac{1}{n} ( 1 - \frac{1}{|\calX|} ) \approx \frac{1}{n}$ when $|\calX| \gg 1$. 
Thus, when $\epsilon=\ln|\calX|$ and 
$|\calXS| \ll \sqrt{|\calX|}$,
the 
$(\calXS,\epsilon)$-\colorBB{uRAP} 
achieves almost the same data utility as the non-private mechanism, 
whereas the expected $l_2$ loss of the $\epsilon$-RAPPOR is 
$\sqrt{|\calX|}$ 
times larger than that of the non-private mechanism \cite{Kairouz_ICML16}.

\subsection{Experimental Results of the MSE}
\label{sec:results_MSE}

Figures~\ref{fig:res1_MSE}, \ref{fig:res2_MSE}, \ref{fig:res4_MSE}, and \ref{fig:res3_MSE}
show the results of the MSE corresponding to Figures~\ref{fig:res1_TV}, \ref{fig:res2_TV},  \ref{fig:res4_TV}, and \ref{fig:res3_TV}, respectively. 
It can be seen that a tendency similar to the results of the TV is obtained for the results of the MSE, meaning that our proposed methods are effective in terms of both the $l_1$ and $l_2$ losses. 

\section{Properties of PUMs}
\label{sec:PUM:details}

\subsection{Privacy Analysis of PUMs}
\label{sub:PUM:privacy:details}
Below we show the proof of 
\colorBB{Propositions~\ref{prop:PUM_ULDP} and \ref{prop:PUM_DPcalXS}}.

\propPUMULDP*

\begin{proof}
Since $\bmQ_{cmn}$ provides $(\calZS,\calYP,\epsilon)$-\ULDP{}, 
for any output data $y \in \calYI$, there exists intermediate data 
$x \in \calXN$ such that 
$\bmQ_{cmn}(y|x) > 0$ and $\bmQ_{cmn}(y|x') = 0$ for any $x' \in \calZ\setminus\{x\}$. In addition, 
by the property of the pre-processor $f_{pre}\uid{i}$ (see (\ref{eq:f_pre_semantic})), 
if the intermediate data is $x \in \calXN$, then the input data is always $x \in \calXN$. 
Therefore, the PUM $\bmQ\uid{i}$ satisfies (\ref{eq:xs_ys_0}).

In addition, 
(\ref{eq:epsilon_OSLDP_whole}) holds for any $z,z' \in \calZ$ and any $y \in \calY$, since $\bmQ_{cmn}$ provides $(\calZS,\calYP,\epsilon)$-\ULDP{}. 
Let $Z\uid{i}$ be a random variable representing intermediate data of the $i$-th user. 
Then, 
$\Pr(Y\uid{i} = y | X\uid{i} = x) = \Pr(Y\uid{i} = y | Z\uid{i} = \bot_k)$ for $x \in \calX_{S,k}\uid{i}$, since $x \in \calX_{S,k}\uid{i}$ is deterministically mapped to $\bot_k$. 
Moreover, 
$\Pr(Y\uid{i} = y | X\uid{i} = x) = \Pr(Y\uid{i} = y | Z\uid{i} = x)$ for $x \notin \calX_S\uid{i}$. 
Thus, 
(\ref{eq:epsilon_OSLDP_whole}) holds for any $x,x' \in \calX$ and any $y \in \calY$.  
\end{proof}

\propPUMDPcalXS*
\begin{proof}
\colorBB{Since $\bmQ\uid{i} = \bmQ_{cmn} \circ f_{pre}\uid{i}$ and $f_{pre}\uid{i}$ is given by (\ref{eq:f_pre_semantic}), we have:
\begin{align}
    \bmQ\uid{i}(y|x) = 
    \begin{cases}
        \bmQ_{cmn}(y|\bot_k) & \text{(if $x \in \calX_{S,k}\uid{i}$)}\\
        \bmQ_{cmn}(y|x) & \text{(otherwise)}.\\
    \end{cases} 
    \label{eq:propPUMDPcalXS:Qi}
\end{align}
Similarly, we have:
\begin{align}
    \bmQ\uid{j}(y|x) = 
    \begin{cases}
        \bmQ_{cmn}(y|\bot_k) & \text{(if $x \in \calX_{S,k}\uid{j}$)}\\
        \bmQ_{cmn}(y|x) & \text{(otherwise)}.\\
    \end{cases} 
    \label{eq:propPUMDPcalXS:Qj}
\end{align}
Since $\bmQ_{cmn}$ provides $(\calZS,\calYP,\epsilon)$-\ULDP{}, for any $z, z' \in \calZ$ and any $y \in \calYP$, we have:
\begin{align}
    \bmQ_{cmn}(y|z) \leq e^\epsilon \bmQ_{cmn}(y|z').
    \label{eq:propPUMDPcalXS:Qcmn_ULDP}
\end{align}
By (\ref{eq:propPUMDPcalXS:Qi}), (\ref{eq:propPUMDPcalXS:Qj}), and (\ref{eq:propPUMDPcalXS:Qcmn_ULDP}), Proposition~\ref{prop:PUM_DPcalXS} holds.} 
\end{proof}

\subsection{Utility Analysis of PUMs}
\label{sub:PUM:utility:details}
\label{sec:proof_theorem_l1_decomposition}
Below we show the proof of Theorem~\ref{thm:l1_loss_decomposition}.

\lonelossPUM*

\begin{proof}
Let $\hbmp^*$ be the estimate of $\bmp$ in the case where the exact distribution $\pi_k$ is known to the analyst; i.e., $\hat{\pi}_k = \pi_k$ for any 
$k = 1, \cdots, \kappa$. 
Then the $l_1$ loss of $\hbmp$ can be written, using the triangle inequality, as follows:
\begin{align}
l_1(\hbmp, \bmp) \leq l_1(\hbmp, \hbmp^*) + l_1(\hbmp^*, \bmp).
\label{eq:l_1_triangle}
\end{align}
Since $\pi_k(x)$ is the conditional probability that 
personal data is $x \in \calX$ given that the intermediate data is $z = \bot_k$, we have
\begin{align}
\bmp(x) = \bmr(x) + \sum_{k=1}^\kappa \bmr(\bot_k) \bmpi_k(x).
\label{eq:bmp_bmpd_bmpi_kappa}
\end{align}
In addition, by substituting $\hbmp^*$ and $\bmpi_k$ for $\hbmp$ and $\hat{\bmpi}_k$ in (\ref{eq:hbmp_PUM_2}), respectively, we have 
\begin{align}
\hbmp^*(x) = \hbmr(x) + \sum_{k=1}^\kappa \hbmr(\bot_k) \bmpi_k(x).
\label{eq:bmp*_bmpd_bmpi_kappa}
\end{align}
By (\ref{eq:bmp_bmpd_bmpi_kappa}) and (\ref{eq:bmp*_bmpd_bmpi_kappa}), an upper bound of $l_1(\hbmp^*, \bmp)$ is given by:
\begin{align}
&l_1(\hbmp^*, \bmp) \nonumber\\
&= \sum_{x \in \calX} |\bmp^*(x) - \bmp(x)| \nonumber\\
&= \sum_{x \in \calX} \left |\hbmr(x) - \bmr(x) + \sum_{k=1}^\kappa (\hbmr(\bot_k) - \bmr(\bot_k)) \bmpi_k(x) \right | \nonumber%\\
\end{align}
\begin{align}
&\leq \sum_{x \in \calX} |\hbmr(x) - \bmr(x) | + \sum_{x \in \calX} \sum_{k=1}^\kappa |\hbmr(\bot_k) - \bmr(\bot_k)| \bmpi_k(x) \nonumber\\
&~~~~(\text{by the triangle inequality}) \nonumber\\
&= \sum_{x \in \calX} |\hbmr(x) - \bmr(x) | + \sum_{k=1}^\kappa |\hbmr(\bot_k) - \bmr(\bot_k)| \nonumber\\
&= \sum_{z \in \calZ} |\hbmr(z) - \bmr(z)| \nonumber\\
&= l_1(\hbmr, \bmr).
\label{eq:l1_hbmp*_bmp_l1_hbmpd_bmpd}
\end{align}
By (\ref{eq:hbmp_PUM_2}) and (\ref{eq:bmp*_bmpd_bmpi_kappa}), $l_1(\hbmp, \hbmp^*)$ is written as follows:
\begin{align}
l_1(\hbmp, \hbmp^*) %\nonumber\\
&= \sum_{x \in \calX} |\hbmp(x) - \hbmp^*(x)| \nonumber\\
&= \sum_{x \in \calX} \sum_{k=1}^\kappa \hat{\bmr}(\bot_k) |\hat{\bmpi}_k(x) - \bmpi_k(x)| \nonumber\\
&= \sum_{k=1}^\kappa \hat{\bmr}(\bot_k) \sum_{x \in \calX} |\hat{\bmpi}_k(x) - \bmpi_k(x)| \nonumber\\
&= \sum_{k=1}^\kappa \hat{\bmr}(\bot_k) l_1(\hat{\bmpi},\bmpi).
\label{eq:l1_hbmp_hbmp*_hbmpd_l1_hbmpi_bmpi}
\end{align}
By (\ref{eq:l_1_triangle}), (\ref{eq:l1_hbmp*_bmp_l1_hbmpd_bmpd}), and (\ref{eq:l1_hbmp_hbmp*_hbmpd_l1_hbmpi_bmpi}), the inequality (\ref{eq:l1_loss_decomposition}) holds. 
\end{proof}

\conference{
\begin{figure}[p]
\centering
\includegraphics[width=0.99\linewidth]{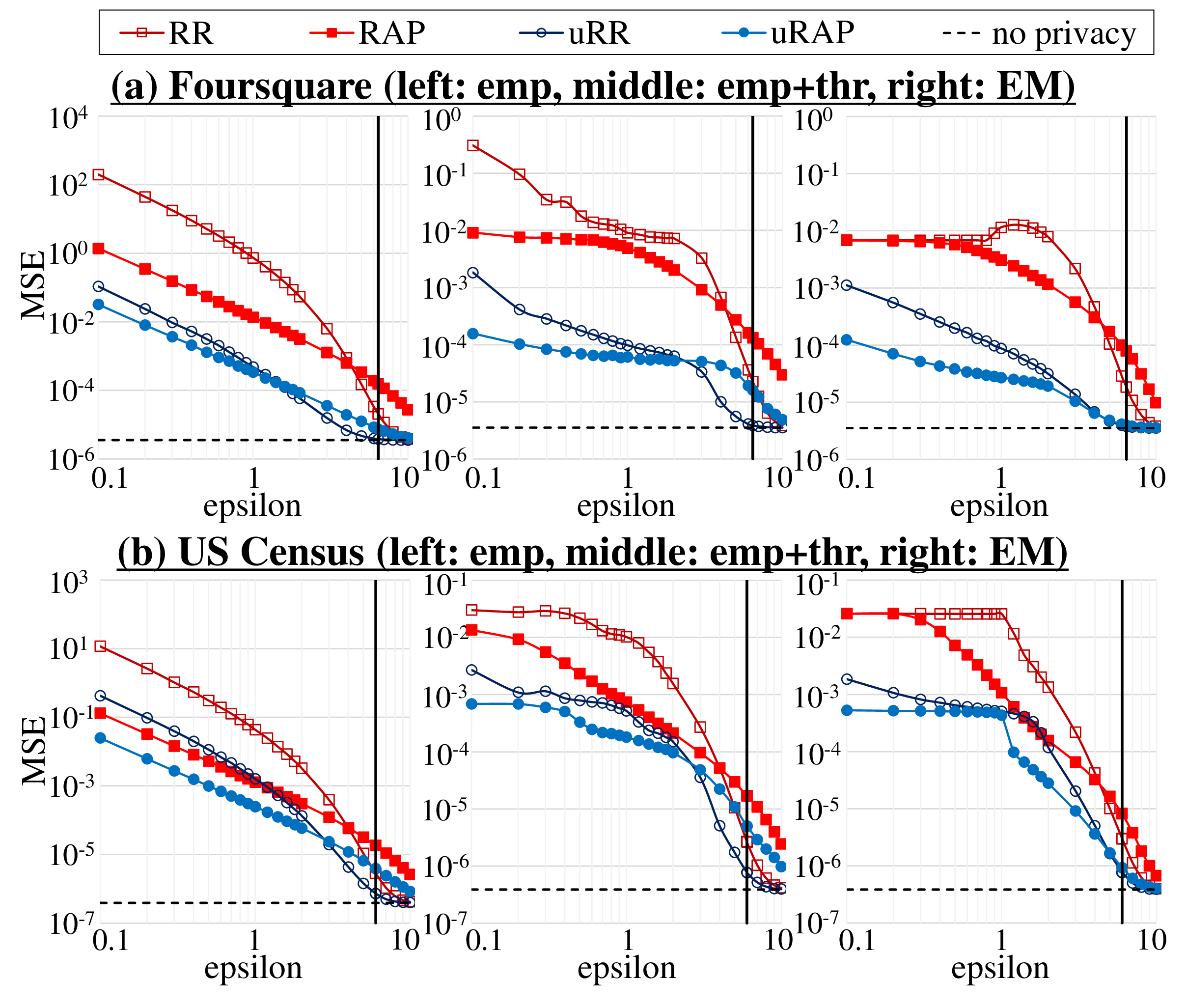}
\vspace{-5mm}
\caption{$\epsilon$ vs.~MSE (common-mechanism). 
A bold line parallel to the $y$-axis represents $\epsilon = \ln |\calX|$.}
\label{fig:res1_MSE}
\vspace{15mm}
\centering
\includegraphics[width=0.99\linewidth]{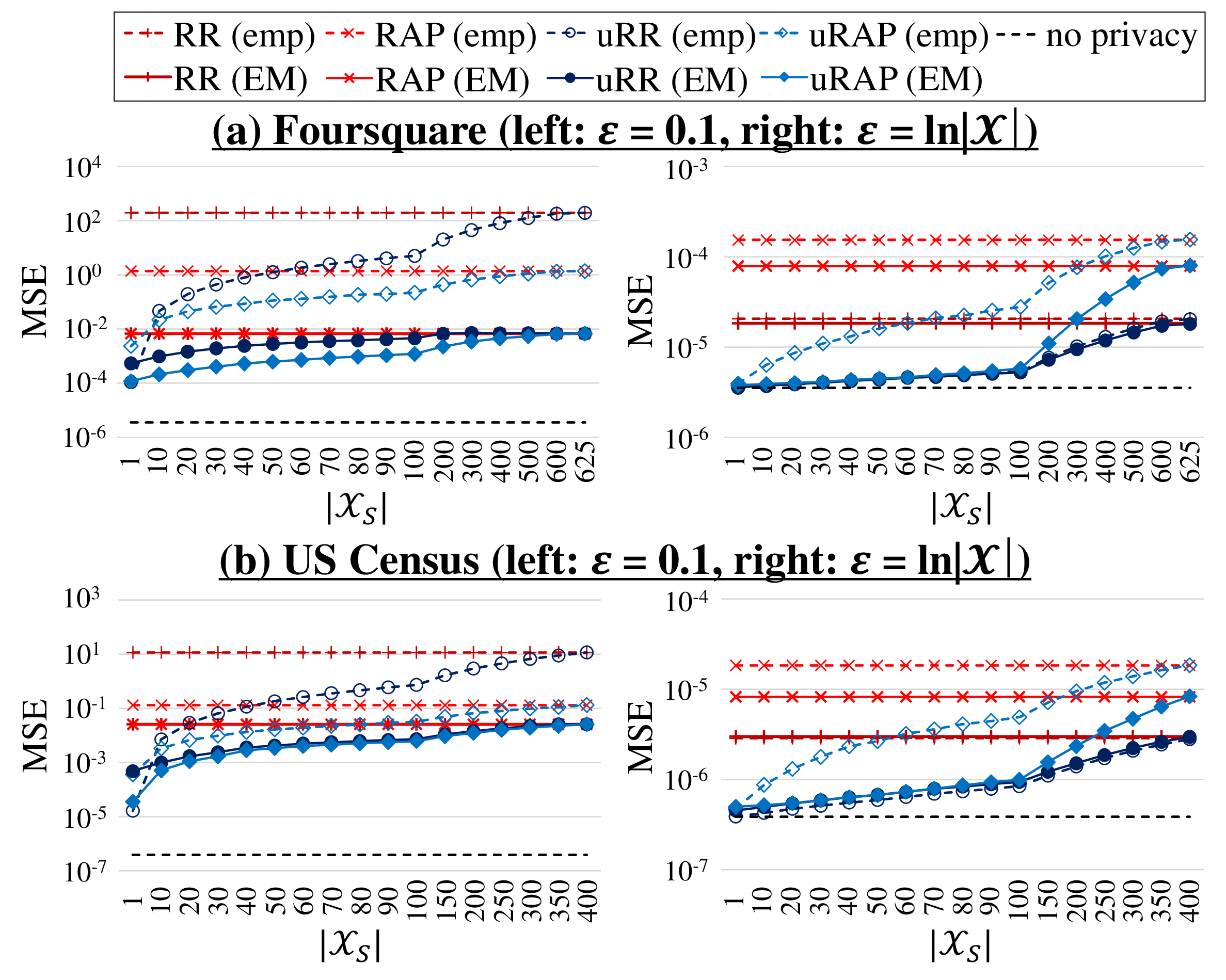}
\vspace{-5mm}
\caption{$|\calXS|$ vs.~MSE when $\epsilon = 0.1$ or $\ln |\calX|$.}
\label{fig:res2_MSE}
\end{figure}
\begin{figure}[p]
\centering
\includegraphics[width=1.0\linewidth]{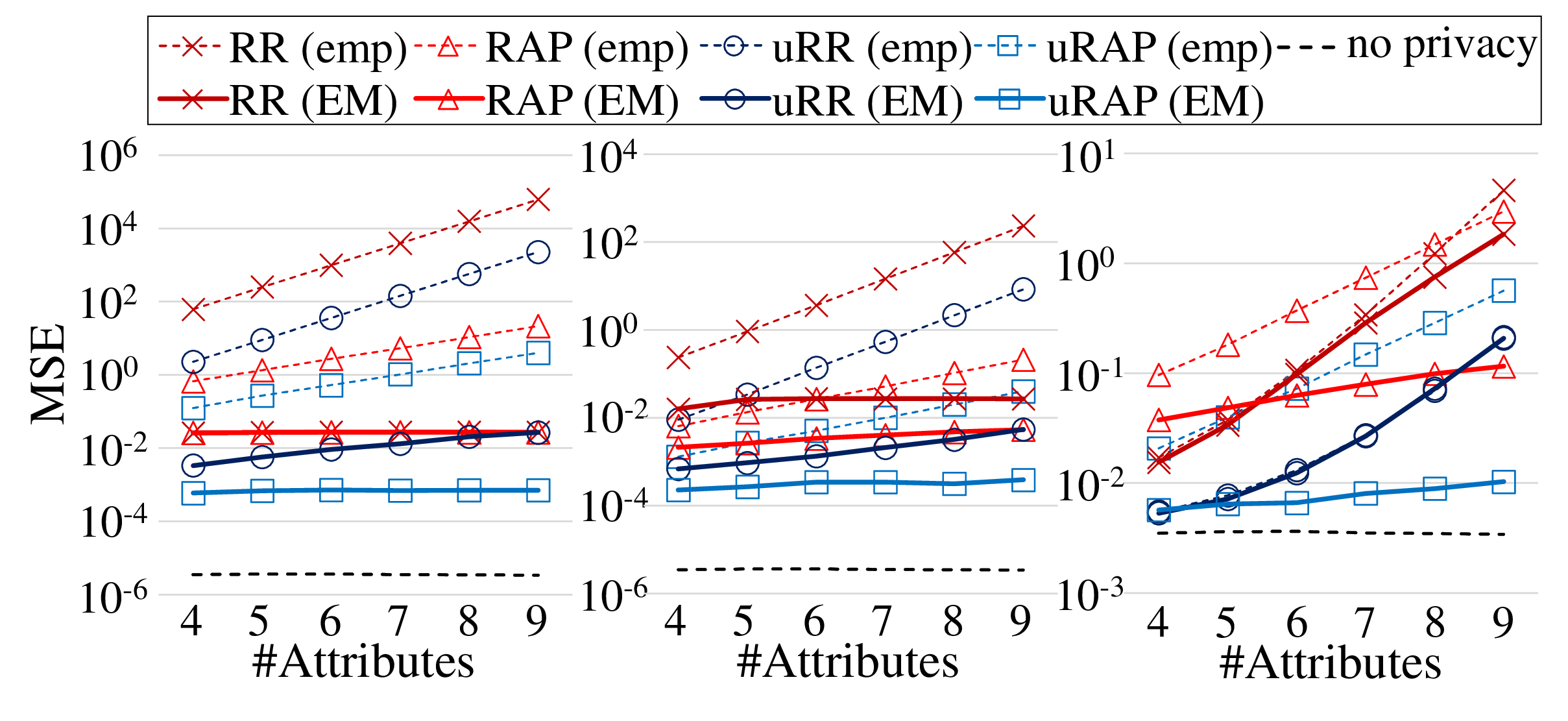}
\vspace{-5mm}
\caption{Number of attributes vs.~MSE (US Census dataset; left: $\epsilon=0.1$, middle: $\epsilon=1.0$, right: $\epsilon=6.0$).}
\label{fig:res4_MSE}
\vspace{15mm}
\centering
\includegraphics[width=0.99\linewidth]{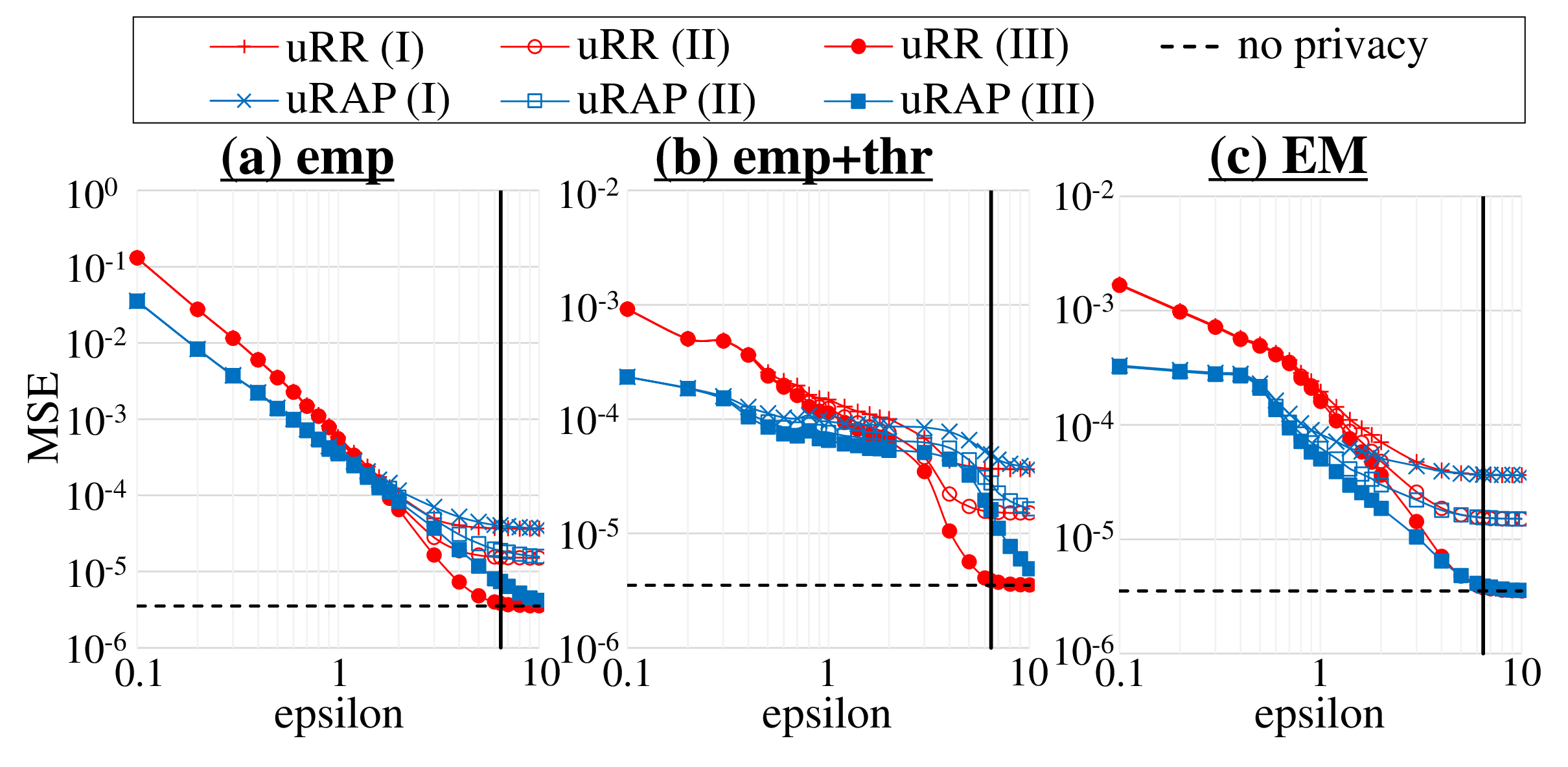}
\vspace{-5mm}
\caption{$\epsilon$ vs.~MSE (personalized-mechanism) ((I): w/o background knowledge, (II) POI distribution, (III) true distribution).
}
\label{fig:res3_MSE}
\end{figure}
}

\arxiv{
\begin{figure}[p]
\centering
\includegraphics[width=0.97\linewidth]{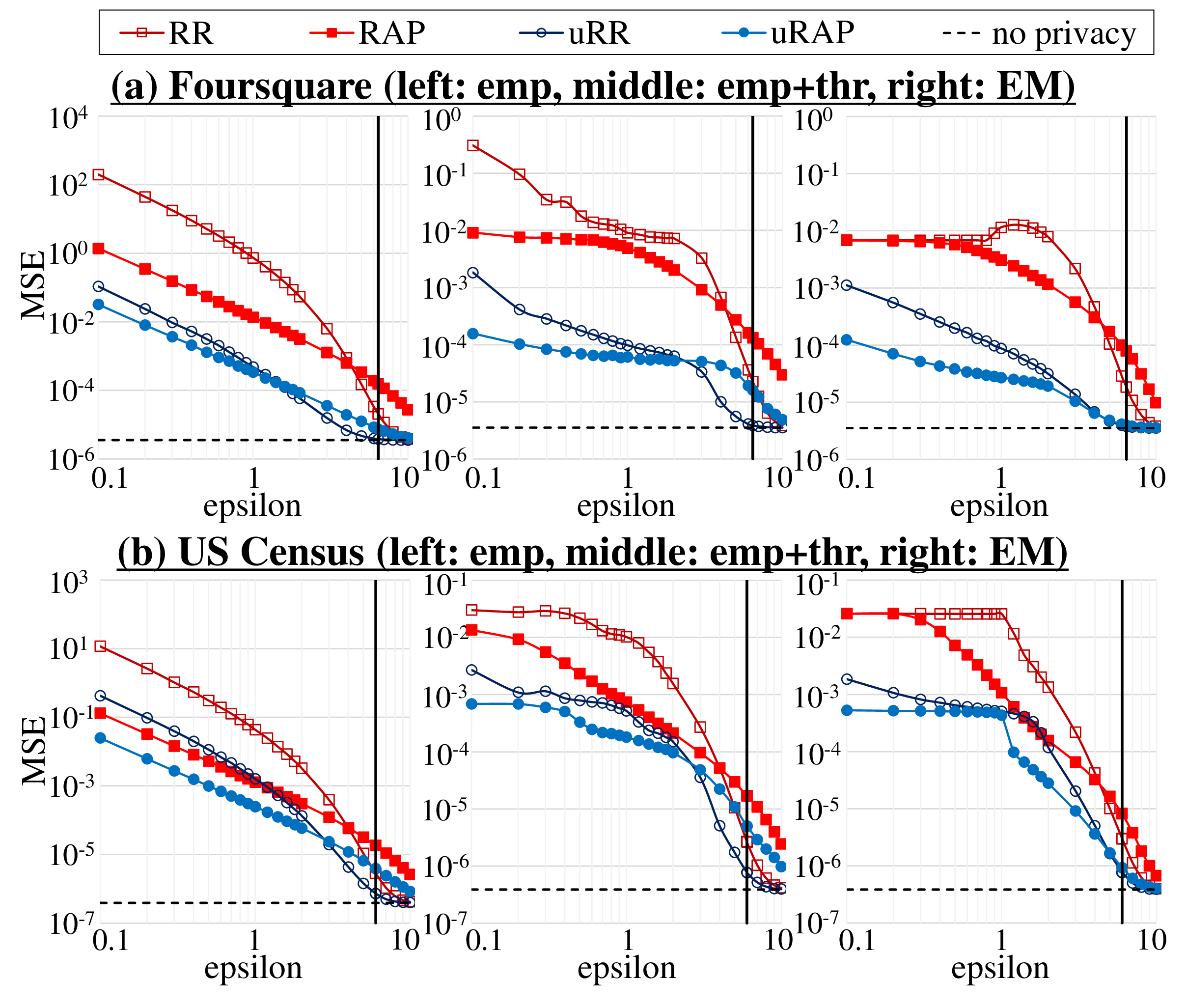}
\vspace{-5mm}
\caption{$\epsilon$ vs.~MSE (common-mechanism). 
A bold line parallel to the $y$-axis represents $\epsilon = \ln |\calX|$.}
\label{fig:res1_MSE}
\vspace{15mm}
\centering
\includegraphics[width=0.97\linewidth]{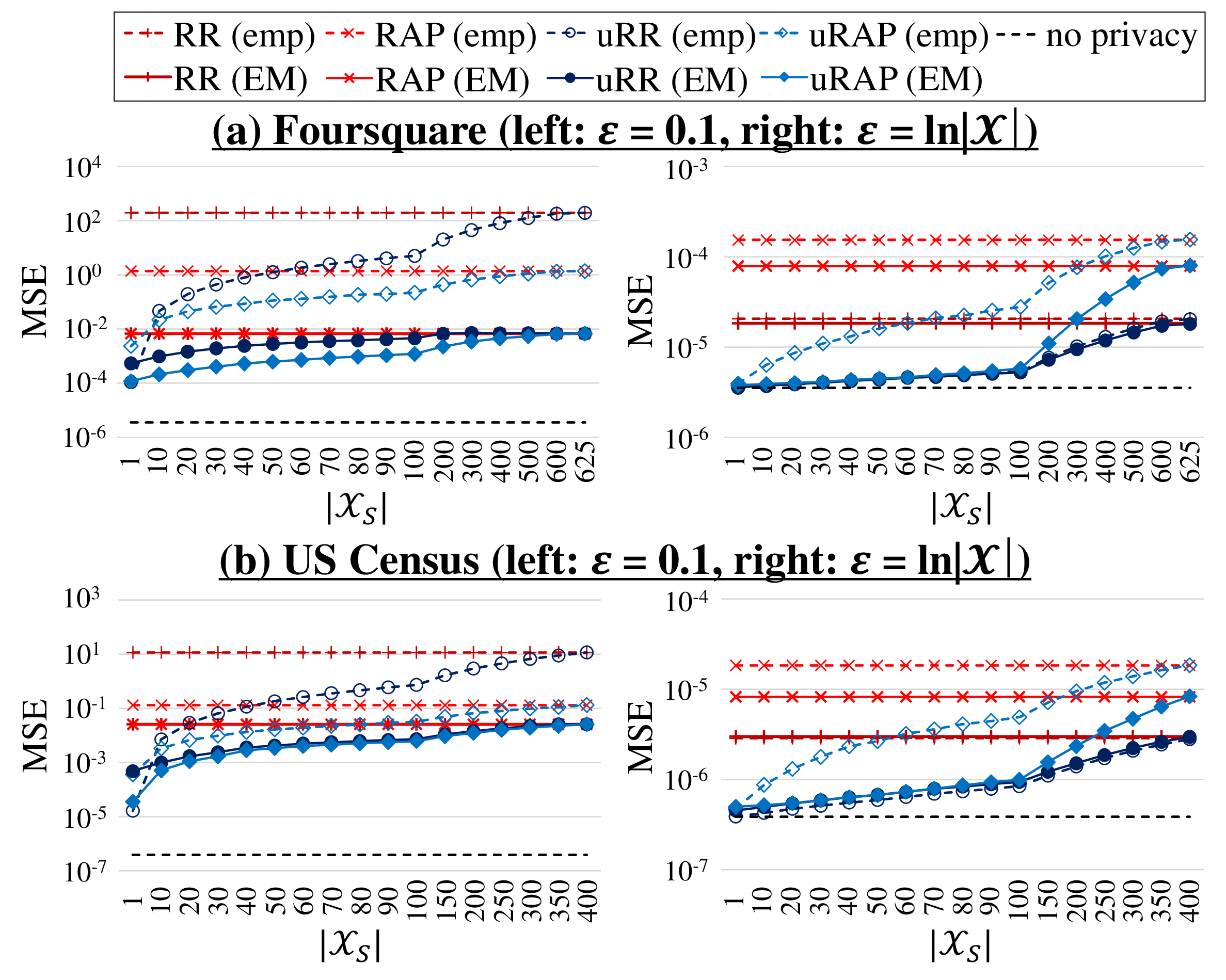}
\vspace{-5mm}
\caption{$|\calXS|$ vs.~MSE when $\epsilon = 0.1$ or $\ln |\calX|$.}
\label{fig:res2_MSE}
\end{figure}
\begin{figure}[p]
\centering
\includegraphics[width=0.97\linewidth]{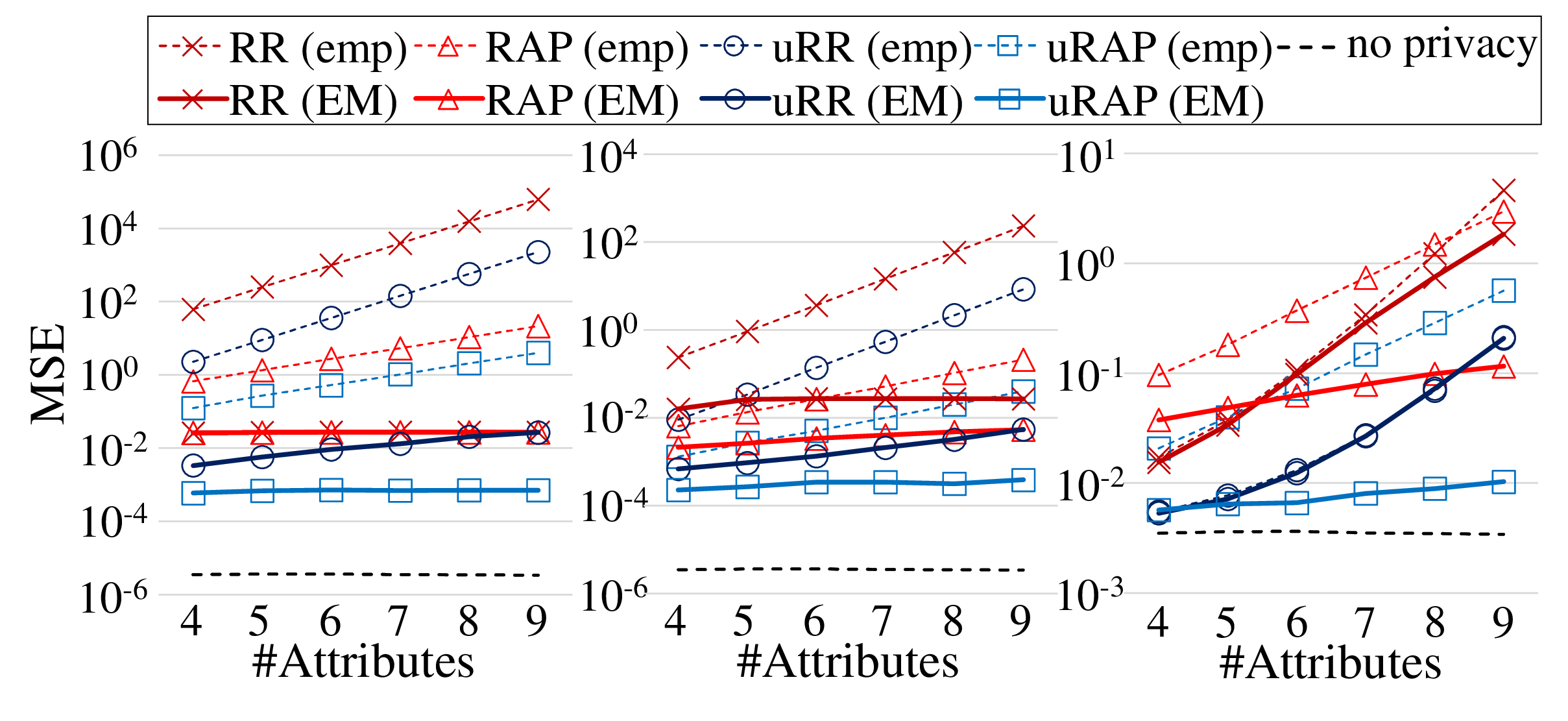}
\vspace{-5mm}
\caption{Number of attributes vs.~MSE (US Census dataset; left: $\epsilon=0.1$, middle: $\epsilon=1.0$, right: $\epsilon=6.0$).}
\label{fig:res4_MSE}
\vspace{15mm}
\centering
\includegraphics[width=0.97\linewidth]{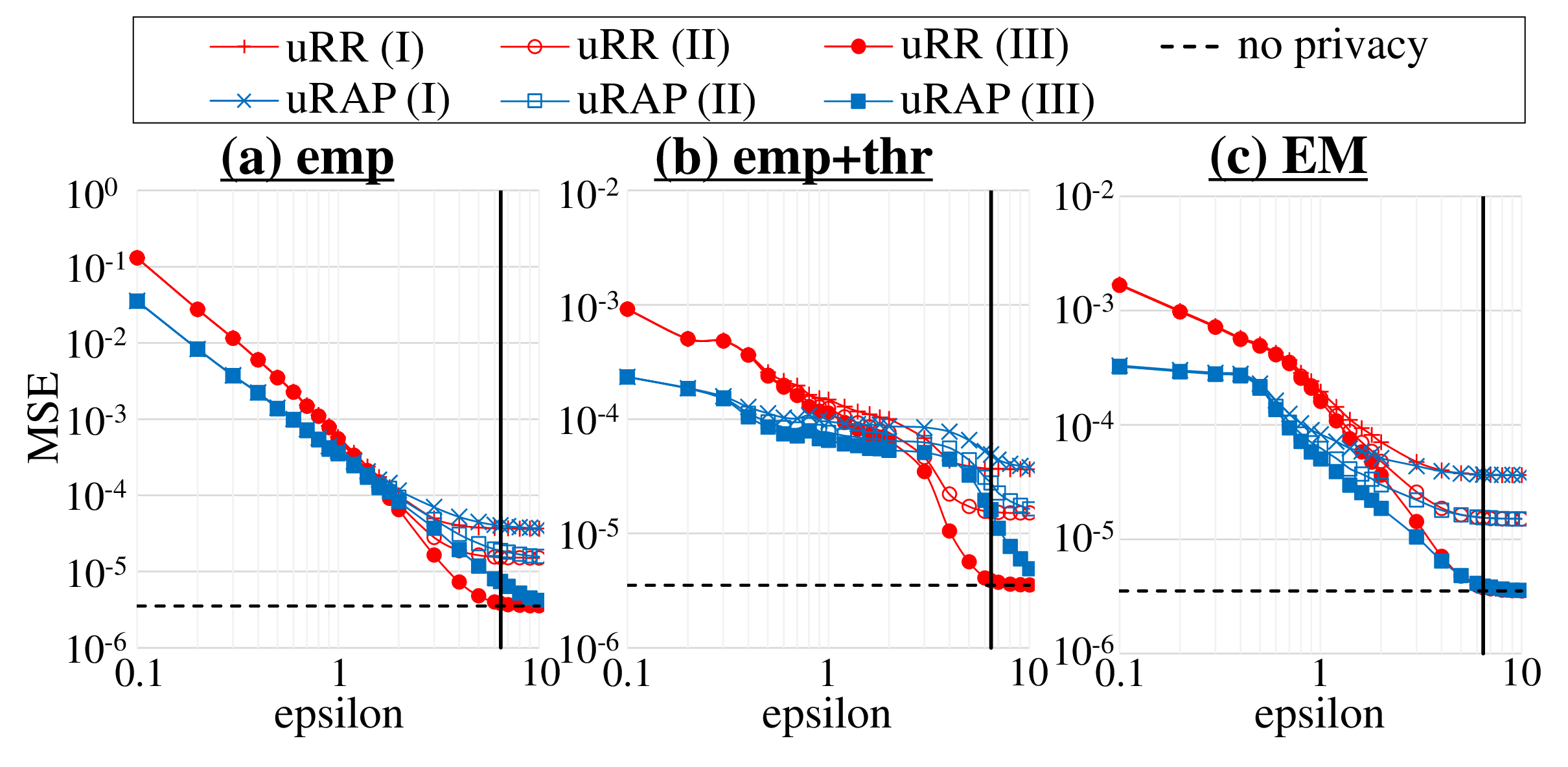}
\vspace{-5mm}
\caption{$\epsilon$ vs.~MSE (personalized-mechanism) ((I): w/o background knowledge, (II) POI distribution, (III) true distribution).
}
\label{fig:res3_MSE}
\end{figure}
}

\end{document}